\documentclass[10pt]{article}

\usepackage{amsmath}
\usepackage{amssymb}
\usepackage{amsthm}
\usepackage{latexsym}
\usepackage{color}
\usepackage{float}
\usepackage{graphicx}
\usepackage{color}
\usepackage{epstopdf}
\usepackage{enumerate}
\usepackage[T1]{fontenc}
\usepackage[english]{babel}
\usepackage[table]{xcolor}

\DeclareSymbolFont{calletters}{OMS}{cmsy}{m}{n}
\DeclareSymbolFontAlphabet{\mathcal}{calletters}
\def\be{\begin{eqnarray}}
\def\ee{\end{eqnarray}}

\def\b*{\begin{eqnarray*}}
\def\e*{\end{eqnarray*}}

\newtheorem{Theorem}{Theorem}[part]

\newtheorem{Lemma}{Lemma}[part]

\makeatletter \@addtoreset{equation}{section}

\@addtoreset{Definition}{section}

\@addtoreset{Theorem}{section}

\@addtoreset{Proposition}{section}

\@addtoreset{Property}{section}

\@addtoreset{Assumption}{section}

\@addtoreset{Corollary}{section}

\@addtoreset{Lemma}{section}

\@addtoreset{Remark}{section}

\@addtoreset{Example}{section}

\@addtoreset{Condition}{section}

\def \R{\mathbb{R}}

\def\={\;=\;}
\def\.{\;.}

\def\1{{\bf 1}}

\def\b*{\begin{eqnarray*}}
\def\e*{\end{eqnarray*}}

 \def\normeL2#1{\left\|{#1}\right\|_{L^2}}

\begin{document}
\title{Optimal market making\footnote{This research has been conducted with the support of the Research
Initiative ``Nouveaux traitements pour les donn\'ees lacunaires issues des activit\'es de cr\'edit'' financed by BNP Paribas under the aegis of the Europlace Institute of Finance.  The author would like to thank Philippe Amzelek (BNP Paribas), Laurent Carlier (BNP Paribas), \'Alvaro Cartea (Oxford University), David Evangelista (KAUST), Yuyun Fang (BNP Paribas), Jean-David Fermanian (ENSAE and CREST), Joaquin Fernandez-Tapia (Universit\'e Pierre et Marie Curie), Sebastian Jaimungal (University of Toronto), Jean-Michel Lasry (Paris Sciences et Lettres), Charles-Albert Lehalle (Capital Fund Management), Pierre-Louis Lions (Coll\`ege de France), Jiang Pu (Institut Europlace de Finance), Andrei Serjantov (BNP Paribas), Vladimir Vasiliev (BNP Paribas), and Douglas Vieira (Imperial College) for the discussions he had with them on the subject.}}
 \author{Olivier {\sc Gu\'eant} \footnote{Full Professor of Applied Mathematics. Universit\'e Paris 1 Panth\'eon-Sorbonne. Centre d'Economie de la Sorbonne. 106 Boulevard de l'H\^opital, 75013 Paris, France. The author was initially affiliated with ENSAE and CREST (3 avenue Pierre Larousse, 92245 Malakoff Cedex, France). Email: \texttt{olivier.gueant@univ-paris1.fr}}
 }
 \date{}

\maketitle
\begin{abstract}

Market makers provide liquidity to other market participants: they propose prices at which they stand ready to buy and sell a wide variety of assets. They face a complex optimization problem with both static and dynamic components. They need indeed to propose bid and offer/ask prices in an optimal way for making money out of the difference between these two prices (their bid-ask spread). Since they seldom buy and sell simultaneously, and therefore hold long and/or short inventories, they also need to mitigate the risk associated with price changes, and subsequently skew their quotes dynamically. In this paper, (i)~we propose a general modeling framework which generalizes (and reconciles) the various modeling approaches proposed in the literature since the publication of the seminal paper ``High-frequency trading in a limit order book'' by Avellaneda and Stoikov, (ii) we prove new general results on the existence and the characterization of optimal market making strategies, (iii) we obtain new closed-form approximations for the optimal quotes, (iv) we extend the modeling framework to the case of multi-asset market making and we obtain general closed-form approximations for the optimal quotes of a multi-asset market maker, and (v) we show how the model can be used in practice in the specific (and original) case of two credit indices.\\

\vspace{5mm}

\end{abstract}

\section{Introduction}

What is a market maker? In a nutshell, it is a liquidity provider. However, it is complex to give a precise definition because the exact role of market makers depends on the considered market. Furthermore, the very definition of a market maker has been blurred in recent years, because of the electronification of most markets and because of the emergence of high-frequency trading in many of them.\\

On most order-driven markets, such as many stock markets, there are nowadays several kinds of market makers. First, there are ``official'' market makers (actually, market making companies): these market makers have usually signed an agreement with a given exchange, or with a given company, for maintaining fair and orderly markets. The Designated Market Makers (DMM) on the NYSE, which succeeded the market specialists, are examples of such ``official'' market makers. They often have contractual obligations, such as participating to the opening and closing auctions and/or quoting with a reasonable bid-ask spread -- e.g.~the DMMs must quote at the National Best Bid and Offer (NBBO) a specified percentage of the time. In addition to these ``official'' market makers, other market participants in the stock markets, in particular some high-frequency traders, are often regarded as market makers (Menkveld calls them the new market makers in \cite{menkveld2013high}) because they are almost continuously present on both sides of the limit order books. They are acting as liquidity providers even though they have no obligation to do so: they just try to make money out of their high-frequency market making strategies. The electronification of most order-driven markets makes it possible for trading firms to act as liquidity providers, hence a blurring of the definition of ``market maker''.\\

On quote-driven markets, such as the corporate bond markets, the market makers are the dealers (these markets are often also called ``dealer markets''). These dealers provide liquidity to the other market participants (the ``clients'') by quoting bid and offer prices on a regular basis. However, their exact behavior depends on the considered market. On some markets, dealers' quotes are firm quotes, whereas on other markets the quotes are streamed only for information (and for a specific size/notional) and become binding when dealers answer specific requests.\\

In this paper, we consider that a market maker is somebody (or in fact an algorithm) who proposes prices at which he/she/it stands ready to buy or sell one or several assets. In particular, we do not consider any contractual constraint, and we assume that all quotes are firm quotes (for a given fixed size). The problem we consider is the determination of the optimal quotes a market maker should propose at the bid and the offer to make money while mitigating inventory risk.\\

This problem is a complex one from a quantitative viewpoint with both static and dynamic components. Market makers face indeed a classical static trade-off: high margin and low volume vs. low margin and high volume. A market maker who quotes a large spread (with no skew) trades rarely, but each transaction leads to a large Mark-to-Market (MtM) gain. Conversely, a market maker quoting a narrow spread (with no skew) trades often, but each transaction leads to a small MtM gain. In addition to this static trade-off, market makers face a dynamic problem: they must adapt their quotes dynamically to reduce their exposure to price changes. For instance, a single-asset market maker with a long inventory should price conservatively on the bid side and aggressively on the ask side, because he wants to reduce his probability to buy and increase his probability to sell. Symmetrically, if he has a short inventory, then he should price aggressively on the bid side and conservatively on the ask side.\\

Like in almost all the mathematical literature on market making, we consider the problem of a single market maker in a simplified way: (i) market prices\footnote{The exact nature of these market prices depends on the considered market. In the case of most order-driven markets (such as most stock markets), a market price may be a mid-price. It may also be a price based on the most recent transactions. In the case of the European corporate bond market, the Composite Bloomberg Bond Trader (CBBT) price is a composite price which may be regarded as a proxy for the market price of a bond. In the case of the US corporate bond market, a market price may also be built by using a mix between TRACE data (in spite of the lag) and the CBBT prices. In all cases, the market prices involved in the model should be regarded as reference prices.} are modeled by stochastic processes assumed to be exogenous to the market maker's behavior,\footnote{See \cite{cartea2015algorithmic} and \cite{gueant2013dealing} for models with adverse selection effects.} and (ii) the probability that the market maker buys (respectively sells) a security at the bid (respectively offer) price he quotes depends on the distance between the quoted price and the market price of that security -- this is the classical Avellaneda-Stoikov modeling framework  -- see \cite{avellaneda2008high}. In particular, the competition between market makers is not explicitly modeled.\\

Since the publication of the seminal paper ``High-frequency trading in a limit order book'' by Avellaneda and Stoikov (see \cite{avellaneda2008high}), market making has been one of the important research topics in quantitative finance.\footnote{Market making has always been an important topic for economists -- see for instance the model of Grossman and Miller \cite{grossman1988liquidity}. However, the dynamic approaches proposed by mathematicians have shed a new light on market making and make it possible to build algorithms for replacing human market makers. The main (old) economic paper really related to the mathematical literature on market making is the paper \cite{ho1981optimal} by Ho and Stoll published in 1981 -- see also \cite{ho1983dynamics} by the same authors. It is noteworthy that this old paper by Ho and Stoll inspired Avellaneda and Stoikov when they wrote their seminal paper \cite{avellaneda2008high}.} Therefore many models have been proposed to address the problem faced by market makers. Gu\'eant \emph{et al.} considered in \cite{gueant2013dealing} a variant of the model proposed by Avellaneda and Stoikov and showed that the four-dimensional Hamilton-Jacobi-Bellman (HJB) equation arising from the model could be simplified into a linear system of ordinary differential equations when a specific change of variables is used.\footnote{See also \cite{gueant2013general}.} The paper \cite{gueant2013dealing} also contains the Gu\'eant--Lehalle--Fernandez-Tapia formulas which are closed-form approximations of the optimal quotes of a single-asset market maker. These approximation formulas are used in practice by major banks in Europe and Asia for market making in (illiquid) quote-driven markets or for market making in some order-driven markets (in the specific case of a small tick size).\\

In the above papers, the objective function of the market maker is the expected CARA utility\footnote{CARA means Constant Absolute Risk Aversion. CARA utility functions are utility functions of the form $u(x) = - e^{-\gamma x}$ for $\gamma > 0$.} of his P\&L (sometimes with a penalty for the terminal inventory). Other models have been proposed in the literature with different objective functions. In their paper on market making with general price dynamics, Fodra and Labadie \cite{fodra2012} considered, in addition to the expected CARA utility case, the risk-neutral case and the risk-neutral case with a penalization on the terminal inventory. In a few papers, with various coauthors, and in their recent book \cite{cartea2015algorithmic} with Penalva, Cartea and Jaimungal considered as an objective function the expected value of the P\&L minus a running penalty on the inventory -- see for instance \cite{cartea2013robust}, \cite{cartea2013risk}, and \cite{cartea2014buy}.\\

The numerous researchers involved in market making modeling have also included many features in their models. Cartea and Jaimungal, with their coauthors, have proposed models with price impact, the possibility to consider short-term alpha, the existence of an adverse selection effect,\footnote{Adverse selection is also considered in \cite{gueant2013dealing}.} etc. Recently, new models have emerged to deal with ambiguity aversion: see for instance the paper \cite{cartea2013robust} by Cartea \emph{et al.} and the paper \cite{nystrom} by Nystr\"om et al. -- see also the PhD dissertation of Donnelly \cite{donnelly2014ambiguity}.\\

For strange reasons,\footnote{One reason is the interest for high-frequency trading. High-frequency trading is indeed often discussed for its influence on the price formation process of stocks. Another reason is that some market making models can be regarded as generalizations of optimal execution models built to solve problems coming from the cash-equity industry -- see for instance \cite{bayraktar2011liquidation}, \cite{gueant2012optimal} and \cite{huitema2012optimal}.} academic researchers have mainly focused on stock markets, which are certainly the least relevant markets to apply most of the models they have proposed.\footnote{Models \emph{\`a la} Avellaneda-Stoikov can hardly be applied to most stock markets for at least two reasons: (i)~the discrete nature of prices (especially in the case of stocks with a large tick size), and (ii) the fact that the very nature of the limit order books, which are queuing systems with priorities and volumes, is not taken into account. One of the only market making models really well suited to stocks is the model proposed by Guilbaud and Pham in \cite{guilbaud2013optimal} -- see also \cite{gueantbook} for a variant.} In this paper, we have instead in mind the case of a market maker in a quote-driven market, or in an order-driven market if the tick sizes of the securities are small.\\

The academic literature on market making is also mainly focused on the case of a market maker operating on a single asset. However, in practice, almost all market makers are in charge of a list of securities. For a market maker in charge of several correlated assets, applying an independent market making strategy to each asset is suboptimal in terms of risk management. It is therefore of the utmost importance to build a model accounting for the correlation structure of the security price moves, especially in the case of corporate bonds where there are often dozens of bonds issued by the same company (which are therefore highly correlated).\\

In this paper, we consider a modeling framework \emph{\`a la} Avellaneda-Stoikov with general intensity functions, instead of the exponential intensity functions of most models (see Section~2). We show that the four-variable HJB equation arising from the various optimization criteria used in the literature can be transformed into a simple system of ordinary differential equations (see Section 3). This somehow reconciles the different approaches used in the literature and enables to understand the subtle differences between the various criteria used in the literature. In particular it helps understanding what it means to be averse to price risk and to non-execution risk. We then show in Section 4 how to find closed-form approximations for the optimal quotes. These approximations generalize the Gu\'eant--Lehalle--Fernandez-Tapia formulas to the case of general intensity functions and to the case of the different optimization criteria used in the market making literature. In Section 5, we consider a problem that is very rarely dealt with in the academic literature in spite of its importance for practitioners: multi-asset market making. We show that many results obtained in the one-asset case can be generalized to our multi-asset market making model. In particular, we obtain for the first time in this paper closed-form approximations for the optimal quotes of a multi-asset market maker. This result is an important breakthrough for practitioners because most market makers are in charge of dozens of assets (or even hundreds of assets when the market maker is in fact an algorithm) and often reluctant to solve very large systems of nonlinear differential equations. In Section 6, we apply our findings to the case of two highly correlated credit indices: CDX.NA.IG (CDX North America Investment Grade) and CDX.NA.HY (CDX North America High Yield).\\

\section{Modeling framework and notations}

\subsection{Notations}

Let us fix a probability space $(\Omega, \mathcal{F}, \mathbb{P})$ equipped with a filtration $(\mathcal{F}_t)_{t\in \mathbb{R}_+}$ satisfying
the usual conditions. We assume that all stochastic processes are defined on $(\Omega, \mathcal{F},(\mathcal{F}_t)_{t\in \mathbb{R}_+}, \mathbb{P})$.\\

We consider in this section (and in the following two sections) a market maker in charge of a single asset. The reference price of this asset\footnote{There may not be a proper market price (see the above discussion), hence the wording ``reference price''.} is modeled by a process $(S_t)_t$ with the dynamics
\begin{equation}
\label{sec2:dS}
dS_t = \sigma dW_t,\quad S_0 \text{\; given},
\end{equation}
where $(W_t)_t$ is a standard Brownian motion adapted to the filtration $(\mathcal{F}_t)_{t\in \mathbb{R}_+}$.\\

This market maker proposes bid and ask quotes to buy and sell the asset. These bid and ask quotes are modeled by two stochastic processes, respectively denoted by $(S^b_t)_t$ and $(S^a_t)_t$.\\

Transactions occur at random times corresponding to the arrival times of agents willing to buy or sell the asset. The distribution of the trade times depends obviously on the liquidity of the asset, and on the bid and ask prices quoted by the market maker. We denote by $(N^b_t)_t$ and $(N^a_t)_t$ the two point processes modeling the number of transactions at the bid and at the ask, respectively. We assume that assets are traded $\Delta$ by $\Delta$, i.e., that the quantities traded do not vary across trades.\\

The inventory of the market maker, modeled by the process $(q_t)_t$, has therefore the dynamics
\begin{equation}
\label{sec2:dq}
dq_t = \Delta dN^b_t - \Delta dN^a_t, \quad q_0 \text{\; given.}
\end{equation}

We assume that the processes $(N^b_t)_t$ and $(N^a_t)_t$ are independent of the Brownian motion~$(W_t)_t$. We denote by $(\lambda_t^b)_t$ and $(\lambda_t^a)_t$ the intensity processes of $(N^b_t)_t$ and $(N^a_t)_t$, respectively. As in the classical Avellaneda-Stoikov model (see \cite{avellaneda2008high}), we assume that the intensity processes are functions of the difference between the reference price and the prices quoted by the market maker. In addition, we assume that the market maker stops proposing a bid (respectively ask) quote when his position is above (respectively below) a given threshold $Q$ (respectively $-Q$).\footnote{$Q$ is assumed to be a multiple of $\Delta$.} Formally, we assume that $(\lambda_t^b)_t$ and $(\lambda_t^a)_t$ verify
\begin{equation}
  \label{sec2:intensity} \lambda_t^b = \Lambda^b(\delta_t^b)1_{q_{t-}<Q} \quad \text{and} \quad \lambda_t^a = \Lambda^a(\delta_t^a)1_{q_{t-}>-Q},
   \end{equation} where $$\delta_t^b = S_t - S_t^b \quad \text{and} \quad \delta_t^a = S^a_t - S_t,$$
and where $\Lambda^b$ and $\Lambda^a$ are two functions satisfying the following hypotheses:\footnote{The first three hypotheses are natural. The fourth one is more technical. It ensures in particular that the functions $\pi^{b}: \delta \mapsto \delta \Lambda^{b}(\delta)$ and $\pi^{a}: \delta \mapsto \delta \Lambda^{a}(\delta)$,  which are related to the instantaneous (expected) MtM PnL associated with each side, reach a maximum on $\mathbb{R}$ (in fact on $\mathbb{R}_+$). To see this (we focus on the bid side, but the proof is similar for the ask side), let us notice that
$${\pi^{b}}'(\delta) = 0 \iff \delta + \frac{\Lambda^{b}\left(\delta\right)}{{\Lambda^{b}}'\left(\delta\right)} = 0.$$
But $\upsilon^{b}: \delta \mapsto \delta + \frac{\Lambda^{b}\left(\delta\right)}{{\Lambda^{b}}'\left(\delta\right)}$ is a strictly increasing function with $$\inf_{\delta}{\upsilon^{b}}'(\delta) = 2 - \sup_{\delta} \frac{\Lambda^{b}\left(\delta\right) {\Lambda^{b}}''\left(\delta\right)}{\left({\Lambda^{b}}'\left(\delta\right)\right)^2} > 0.$$
Therefore, the equation $\upsilon^{b}(\delta) = 0$ has a unique solution and it corresponds to a unique maximizer for $\pi^{b}$.}
\begin{itemize}
  \item $\Lambda^{b}$ and $\Lambda^{a}$ are twice continuously differentiable,
  \item $\Lambda^{b}$ and $\Lambda^{a}$ are decreasing, with $\forall \delta \in \R$, ${\Lambda^{b}}'(\delta) <0$ and ${\Lambda^{a}}'(\delta) <0$,
  \item $\lim_{\delta \to +\infty} \Lambda^{b}(\delta) = \lim_{\delta \to +\infty} \Lambda^{a}(\delta) = 0$,
  \item $\sup_{\delta}\frac{\Lambda^{b}(\delta){\Lambda^{b}}''(\delta)}{\left({\Lambda^{b}}'(\delta)\right)^2} <  2 \quad \text{and} \quad \sup_{\delta} \frac{\Lambda^{a}(\delta){\Lambda^{a}}''(\delta)}{\left({\Lambda^{a}}'(\delta)\right)^2} <  2.$
\end{itemize}

Finally, the process $(X_t)_t$ models the market maker's cash account. Given our modeling framework, $(X_t)_t$ has the dynamics
\begin{eqnarray}
\nonumber dX_t &=& S^a_t \Delta dN^a_t - S^b_t \Delta dN^b_t\\
&=&  (S_t + \delta^a_t) \Delta dN^a_t - (S_t -
\delta^b_t) \Delta dN^b_t. \label{sec2:dX}
\end{eqnarray}

\subsection{The two classical optimization problems}

In the above paragraphs, we have defined the three processes at the heart of most market making models: the reference price process $(S_t)_t$, the inventory process $(q_t)_t$, and the cash process $(X_t)_t$. We now need to define the problem faced by the market maker. Following the model proposed by Avellaneda and Stoikov in \cite{avellaneda2008high}, one can consider, as in \cite{gueant2013dealing}, that the market maker maximizes the expected value of a CARA utility function (with risk aversion parameter $\gamma > 0$) applied to the MtM value of the portfolio at a given date $T$. This MtM value at time $T$ is basically $X_T + q_T S_T$, or $X_T + q_T S_T - \ell(|q_T|)$ if we add a liquidity premium for the remaining inventory (whatever its sign) -- $\ell$ is a nondecreasing and convex function from $\R_+$ to $\R_+$. In this general framework, the goal of the market maker is to maximize
\begin{equation*}
 \mathbb{E}\left[-\exp\left(-\gamma (X_T + q_T S_T - \ell(|q_T|)) \right)\right], \qquad \textbf{(Model A)}
\end{equation*}
over $(\delta_t^b)_t \in \mathcal{A}$ and $(\delta_t^a)_t \in \mathcal{A}$, where the set of admissible controls $\mathcal{A}$ is simply the set of predictable processes bounded from below.\\

Alternatively, one can consider that the market maker maximizes the expected value of the MtM value of the portfolio at date $T$, but that holding an inventory is penalized over the time interval $[0,T]$. This is typically what is done by Cartea, Jaimungal and their coauthors (see the recent book \cite{cartea2015algorithmic} for several examples). In that kind of model, the goal of the market maker is to maximize an expression of the form
\begin{equation*}
\mathbb{E}\left[X_T + q_T S_T - \ell(|q_T|) - \frac 12 \gamma \sigma^2 \int_0^T q_t^2 dt \right], \qquad \textbf{(Model B)}
\end{equation*}
over $(\delta_t^b)_t \in \mathcal{A}$ and $(\delta_t^a)_t \in \mathcal{A}$.\\

\section{Towards a single system of ordinary differential equations for characterizing the optimal quotes}

Both Model A and Model B can be solved using the classical tools of stochastic optimal control. In particular, we show that, in both models, finding the value function (and the optimal bid and ask quotes) boils down to solving a tridiagonal system of ordinary differential equations (ODEs), and that the equations associated with Model A and Model B are part of the same family of ODEs.

\subsection{Dimensionality of the problem: a reduction from 4 to 2}

The HJB equation associated with Model A is given by
\begin{equation}
\label{sec3:HJBModelA}0= -\partial_t u(t,x,q,S) - \frac 12 \sigma^2 \partial^2_{SS} u(t,x,q,S)
\end{equation}
$$ - 1_{q<Q}\sup_{\delta^b}
\Lambda^b(\delta^b) \left[u(t,x-\Delta S+\Delta\delta^b,q+\Delta,S) - u(t,x,q,S) \right]$$
$$ - 1_{q>-Q} \sup_{\delta^a} \Lambda^a(\delta^a) \left[u(t,x+\Delta S+\Delta \delta^a,q-\Delta,S) -
u(t,x,q,S) \right],$$ for $q \in \mathcal{Q} = \lbrace -Q, -Q + \Delta, \ldots, Q-\Delta, Q \rbrace$, and $(t,S,x) \in [0,T]\times \mathbb{R}^2$, with the terminal condition \begin{equation}\label{sec3:HJBModelACT}u(T,x,q,S) = -\exp\left(-\gamma(x+qS - \ell(|q|))\right).\end{equation}

If one uses the ansatz
\begin{equation}\label{sec3:ansatzModelA}u(t,x,q,S) = -\exp\left(-\gamma(x+qS + \theta(t,q))\right),\end{equation}
then Eq.~(\ref{sec3:HJBModelA}) becomes
\begin{equation}
\label{sec3:thetaModelA}
0=-\partial_t \theta(t,q) + \frac{1}{2} \gamma \sigma^2 q^2
\end{equation}
$$-  1_{q<Q} \sup_{\delta^{b}} \frac{\Lambda^b(\delta^{b})}{\gamma}\left(1-\exp\left(-\gamma\left(\Delta \delta^{b} + \theta(t,q+\Delta) - \theta(t,q) \right)\right)\right)$$
$$- 1_{q>-Q} \sup_{\delta^{a}} \frac{\Lambda^a(\delta^{a})}{\gamma}\left(1-\exp\left(-\gamma \left(\Delta \delta^{a} + \theta(t,q-\Delta) - \theta(t,q) \right)\right)\right),
$$ for $q \in \mathcal{Q}$, and $t \in [0,T]$, and the terminal condition (\ref{sec3:HJBModelACT}) becomes $\theta(T,q) = -\ell(|q|)$.\\

The HJB equation associated with Model B is given by
\begin{equation}
\label{sec3:HJBModelB}0= -\partial_t u(t,x,q,S) + \frac 12 \gamma \sigma^2 q^2 - \frac 12 \sigma^2 \partial^2_{SS} u(t,x,q,S)
\end{equation}
$$ - 1_{q<Q}\sup_{\delta^b}
\Lambda^b(\delta^b) \left[u(t,x-\Delta S+\Delta\delta^b,q+\Delta,S) - u(t,x,q,S) \right]$$
$$ - 1_{q>-Q} \sup_{\delta^a} \Lambda^a(\delta^a) \left[u(t,x+\Delta S+\Delta \delta^a,q-\Delta,S) -
u(t,x,q,S) \right],$$ for $q \in \mathcal{Q}$, and $(t,S,x) \in [0,T]\times \mathbb{R}^2$, with the terminal condition \begin{equation}\label{sec3:HJBModelBCT}u(T,x,q,S) = x+qS - \ell(|q|).\end{equation}

If one uses the ansatz
\begin{equation}\label{sec3:ansatzModelA}u(t,x,q,S) = x+qS + \theta(t,q),\end{equation}
then Eq.~(\ref{sec3:HJBModelB}) becomes
\begin{equation}
\label{sec3:thetaModelB}
0=-\partial_t \theta(t,q) + \frac{1}{2} \gamma \sigma^2 q^2
\end{equation}
$$-  1_{q<Q} \sup_{\delta^{b}} \Lambda^b(\delta^{b})\left(\Delta \delta^{b} + \theta(t,q+\Delta) - \theta(t,q) \right)$$
$$- 1_{q>-Q} \sup_{\delta^{a}} \Lambda^a(\delta^{a})\left(\Delta \delta^{a} + \theta(t,q-\Delta) - \theta(t,q) \right),
$$ for $q \in \mathcal{Q}$, and $t \in [0,T]$, and the terminal condition (\ref{sec3:HJBModelBCT}) becomes $\theta(T,q) = -\ell(|q|)$.\\

Eqs.~(\ref{sec3:thetaModelA}) and (\ref{sec3:thetaModelB}) are in fact two systems of ODEs which belong to the same family. If we introduce for $\xi >0$ the functions
$$H^b_{\xi}(p) = \sup_{\delta} \frac{\Lambda^b(\delta)}{\xi}\left(1-\exp\left(-\xi\Delta \left(\delta - p \right)\right)\right)$$
and
$$H^a_{\xi}(p) = \sup_{\delta} \frac{\Lambda^a(\delta)}{\xi}\left(1-\exp\left(-\xi\Delta \left(\delta - p \right)\right)\right),$$
and the limit functions (for $\xi=0$) $$H^b_{0}(p) = \Delta \sup_{\delta} \Lambda^b(\delta)(\delta - p)$$ and $$H^a_{0}(p) = \Delta \sup_{\delta} \Lambda^a(\delta)(\delta - p),$$  then we can indeed consider the general equation
\begin{equation}
\label{sec3:thetagen}
0=-\partial_t \theta(t,q) + \frac{1}{2} \gamma \sigma^2 q^2
\end{equation}
$$-  1_{q<Q} H^b_{\xi}\left(\frac{\theta(t,q) - \theta(t,q+\Delta)}{\Delta}\right) - 1_{q>-Q} H^a_{\xi}\left(\frac{\theta(t,q) - \theta(t,q-\Delta)}{\Delta}\right),$$
for $q \in \mathcal{Q}$, and $t \in [0,T]$, with the terminal condition
\begin{equation}
\label{sec3:thetagenCT}
\theta(T,q) = -\ell(|q|).
\end{equation}

Eq.~(\ref{sec3:thetaModelA}) corresponds to Eq.~(\ref{sec3:thetagen}) for $\xi=\gamma$ while Eq.~(\ref{sec3:thetaModelB}) corresponds to Eq.~(\ref{sec3:thetagen}) for $\xi=0$.\\

\subsection{Existence and uniqueness of a solution $\theta$}

In the following paragraphs, we prove, for all $\xi \ge 0$, that there exists a unique solution $\theta$ to Eq.~(\ref{sec3:thetagen}) with terminal condition (\ref{sec3:thetagenCT}).\\

Let us start with a lemma on $H^b_{\xi}$ and $H^a_{\xi}$.\\

\begin{Lemma}
\label{sec3:lemmaH}
$\forall \xi \ge 0$, $H^b_{\xi}$ and $H^a_{\xi}$ are two decreasing functions of class $C^2$.\\

The supremum in the definition of $H^b_{\xi}(p)$ is attained at a unique $\tilde{\delta}^{b*}_\xi(p)$ characterized by
$$p=\tilde{\delta}^{b*}_\xi(p) - \frac 1{\xi\Delta} \log\left(1 -\xi \Delta \frac{\Lambda^b\left(\tilde{\delta}^{b*}_\xi(p)\right)}{{\Lambda^b}'\left(\tilde{\delta}^{b*}_\xi(p)\right)}\right),\quad \textrm{if\;} \xi > 0,$$ and
$$p=\tilde{\delta}^{b*}_\xi(p) + \frac{\Lambda^b\left(\tilde{\delta}^{b*}_\xi(p)\right)}{{\Lambda^b}'\left(\tilde{\delta}^{b*}_\xi(p)\right)}, \quad \textrm{if\;} \xi = 0,$$
 or equivalently by
\begin{equation}
\label{sec3:deltab}
\tilde{\delta}^{b*}_\xi(p) = {\Lambda^b}^{-1}\left(\xi H^b_{\xi}(p) - \frac{{H_{\xi}^b}'(p)}{\Delta}\right).
\end{equation}

Similarly, the supremum in the definition of $H^a_{\xi}(p)$ is attained at a unique $\tilde{\delta}^{a*}_\xi(p)$ characterized by
$$p=\tilde{\delta}^{a*}_\xi(p) - \frac 1{\xi\Delta} \log\left(1 -\xi \Delta \frac{\Lambda^a\left(\tilde{\delta}^{a*}_\xi(p)\right)}{{\Lambda^a}'\left(\tilde{\delta}^{a*}_\xi(p)\right)}\right),\quad \textrm{if\;} \xi > 0,$$ and
$$p=\tilde{\delta}^{a*}_\xi(p) + \frac{\Lambda^a\left(\tilde{\delta}^{a*}_\xi(p)\right)}{{\Lambda^a}'\left(\tilde{\delta}^{a*}_\xi(p)\right)}, \quad \textrm{if\;} \xi = 0,$$
 or equivalently
 \begin{equation}
\label{sec3:deltaa}
\tilde{\delta}^{a*}_\xi(p) = {\Lambda^a}^{-1}\left(\xi H^a_{\xi}(p) - \frac{{H_{\xi}^a}'(p)}{\Delta}\right).
\end{equation}

Furthermore, the functions $p \mapsto \tilde{\delta}^{b*}_\xi(p)$ and $p \mapsto \tilde{\delta}^{a*}_\xi(p)$  are $C^1$ and increasing.\\

\end{Lemma}

\begin{proof}

We prove the results for the bid side. The proof is similar for the ask side.\\

Let us start with $\xi>0$.\\

$\forall p \in \mathbb{R}$, let us define $g_p : \delta \mapsto \frac{\Lambda^b(\delta)}{\xi}\left(1-\exp\left(-\xi\Delta \left(\delta - p \right)\right)\right)$.\\

$g_p$ is a function of class $C^1$, positive for  $\delta \in (p, +\infty)$ and nonpositive otherwise. Because $g_p(p) = 0$ and $\lim_{\delta \to +\infty} g_p(\delta) = 0$, the supremum of $g_p$ is attained at, at least, one point $\tilde{\delta}^{b*}_\xi(p) \in (p,+\infty)$. The first order condition characterizing the suprema of $g_p$ is
$$\frac{{\Lambda^b}'(\delta)}{\xi}\left(1-\exp\left(-\xi\Delta \left(\delta - p \right)\right)\right) + \Delta \Lambda^b(\delta) \exp\left(-\xi\Delta \left(\delta - p \right)\right) = 0.$$

By rearranging the terms, we obtain $$p=\delta - \frac 1{\xi\Delta} \log\left(1 -\xi \Delta \frac{\Lambda^b\left(\delta\right)}{{\Lambda^b}'\left(\delta\right)}\right).$$

Because $\Lambda^{b}(\delta){\Lambda^{b}}''(\delta) <  2 \left({\Lambda^{b}}'(\delta)\right)^2$, the function
$$j : \delta \mapsto \delta - \frac 1{\xi\Delta} \log\left(1 -\xi \Delta \frac{\Lambda^b\left(\delta\right)}{{\Lambda^b}'\left(\delta\right)}\right)$$ is increasing\footnote{We have indeed
$$ j'(\delta) = 1 + \frac{1- \frac{{\Lambda^b}(\delta){\Lambda^b}''(\delta)}{{{\Lambda^b}'(\delta)}^2}}{1 - \xi \Delta\frac{{\Lambda^b}(\delta)}{{\Lambda^b}'(\delta)} }
 = \frac{2 - \frac{{\Lambda^b}(\delta){\Lambda^b}''(\delta)}{{{\Lambda^b}'(\delta)}^2} - \xi \Delta\frac{{\Lambda^b}(\delta)}{{\Lambda^b}'(\delta)}}{1 - \xi \Delta\frac{{\Lambda^b}(\delta)}{{\Lambda^b}'(\delta)} }
> 0.$$
} and there is therefore a unique maximizer $\tilde{\delta}^{b*}_\xi(p)$ of $g_p$, characterized by
$$p=\tilde{\delta}^{b*}_\xi(p) - \frac 1{\xi\Delta} \log\left(1 -\xi \Delta \frac{\Lambda^b\left(\tilde{\delta}^{b*}_\xi(p)\right)}{{\Lambda^b}'\left(\tilde{\delta}^{b*}_\xi(p)\right)}\right).$$

Moreover, by the implicit function theorem, $p \mapsto \tilde{\delta}^{b*}_\xi(p)$ is a function of class $C^1$ which verifies
$$
\tilde{\delta}^{b*'}_\xi(p) =  \frac{1}{j'\left(\tilde{\delta}^{b*}_\xi(p)\right)} = \frac{1 - \xi \Delta\frac{{\Lambda^b}(\tilde{\delta}^{b*}_\xi(p))}{{\Lambda^b}'(\tilde{\delta}^{b*}_\xi(p))} } {2 - \frac{{\Lambda^b}(\tilde{\delta}^{b*}_\xi(p)){\Lambda^b}''(\tilde{\delta}^{b*}_\xi(p))}{{{\Lambda^b}'(\tilde{\delta}^{b*}_\xi(p))}^2} - \xi \Delta\frac{{\Lambda^b}(\tilde{\delta}^{b*}_\xi(p))}{{\Lambda^b}'(\tilde{\delta}^{b*}_\xi(p))}}
> 0.$$
In particular, $p \mapsto \tilde{\delta}^{b*}_\xi(p)$ is increasing.\\

Moreover, the function $H^b_\xi$ is of class $C^2$, with
$${H^b_\xi}'(p) = -{\Lambda^b}(\tilde{\delta}^{b*}_\xi(p)) \Delta \exp\left(-\xi\Delta \left(\tilde{\delta}^{b*}_\xi(p) - p \right)\right)$$
and
$${H^b_\xi}''(p) = \left(-\tilde{\delta}^{b*'}(p) {\Lambda^b}'(\tilde{\delta}^{b*}_\xi(p)) + \xi\Delta {\Lambda^b}(\tilde{\delta}^{b*}_\xi(p)) \left(\tilde{\delta}^{b*'}_\xi(p) - 1 \right) \right) \Delta \exp\left(-\xi\Delta \left(\tilde{\delta}^{b*}_\xi(p) - p \right)\right).$$
In particular, $H^b_\xi$ is decreasing.\\

We also see, by using the definition of $H^b_\xi$, that  $$\tilde{\delta}^{b*}_\xi(p) = {\Lambda^b}^{-1}\left(\xi H^b_{\xi}(p) - \frac{{H_{\xi}^b}'(p)}{\Delta}\right).$$

In the $\xi= 0$ case, we define $\forall p \in \mathbb{R}$, $h_p : \delta \mapsto \Delta \Lambda^{b}(\delta) \left(\delta - p \right)$.\\

$h_p$ is a function of class $C^1$, positive for $\delta \in (p, +\infty)$ and nonpositive otherwise. By using the same reasoning as in footnote 11, we see that there is a unique maximizer $\tilde{\delta}^{b*}_0(p)$ of $h_p$, characterized by
$$p=\tilde{\delta}^{b*}_0(p) + \frac{\Lambda^{b}\left(\tilde{\delta}^{b*}_0(p)\right)}{{\Lambda^{b}}'\left(\tilde{\delta}^{b*}_0(p)\right)}.$$

As above, by the implicit function theorem, $p \mapsto {\delta}^{b*}_0(p)$ is a function of class $C^1$ which verifies
$$
\tilde{\delta}^{b*'}_0(p)= \frac{1} {2 - \frac{\Lambda^{b}(\tilde{\delta}^{b*}_0(p)){\Lambda^{b}}''(\tilde{\delta}^{b*}_0(p))}{{{\Lambda^{b}}'(\tilde{\delta}^{b*}_0(p))}^2} }
> 0.$$
In particular, $p \mapsto \tilde{\delta}^{b*}_0(p)$ is increasing.\\

Moreover, the function $H^{b}_0$ is of class $C^2$, with
$${H^{b}_0}'(p) = - \Delta {\Lambda^{b}}(\tilde{\delta}^{b*}_0(p))$$
and
$${H^{b}_0}''(p) = - \Delta \tilde{\delta}^{b*'}_0(p) {\Lambda^{b}}'(\tilde{\delta}^{b*}_0(p)).$$
In particular, $H^{b}_0$ is decreasing and we have $$\tilde{\delta}^{b*}_0(p) = {\Lambda^{b}}^{-1}\left( - \frac{{H_{0}^{b}}'(p)}{\Delta}\right).$$

This proves the lemma.\qed\\

\end{proof}

We now prove a comparison principle for Eq.~(\ref{sec3:thetagen}) which gives \emph{a priori} bounds that will enable us to prove the existence of a solution to Eq.~(\ref{sec3:thetagen}) with terminal condition (\ref{sec3:thetagenCT}).

\begin{Lemma}
\label{sec3:cp}
Let $\tau \in [0,T)$.\\

Let $\underline{\theta}: [\tau,T]\times \mathcal{Q} \rightarrow \mathbb{R}$ be a $C^1$ function with respect to time satisfying the subsolution property, i.e.,
$$\forall q\in \mathcal{Q},\quad  \underline{\theta}(T,q) \leq -\ell(|q|)$$
and $\forall (t,q) \in [\tau,T)\times \mathcal{Q}$,
$$ -\partial_t \underline{\theta}(t,q) +
\frac{1}{2}\gamma\sigma^2 q^2
-  1_{q<Q} H^b_{\xi}\left(\frac{\underline{\theta}(t,q) - \underline{\theta}(t,q+\Delta)}{\Delta}\right) - 1_{q>-Q} H^a_{\xi}\left(\frac{\underline{\theta}(t,q) - \underline{\theta}(t,q-\Delta)}{\Delta}\right) \leq 0.$$

Let $\overline{\theta}: [\tau,T]\times \mathcal{Q} \rightarrow \mathbb{R}$ be a $C^1$ function with respect to time satisfying the supersolution property, i.e.,
$$\forall q \in \mathcal{Q},\quad \overline{\theta}(T,q) \geq -\ell(|q|)$$
and $\forall (t,q) \in [\tau,T)\times \mathcal{Q}$,
$$-\partial_t \overline{\theta}(t,q) +
\frac{1}{2}\gamma\sigma^2 q^2 - 1_{q<Q} H^b_{\xi}\left(\frac{\overline{\theta}(t,q) - \overline{\theta}(t,q+\Delta)}{\Delta}\right) - 1_{q>-Q} H^a_{\xi}\left(\frac{\overline{\theta}(t,q) - \overline{\theta}(t,q-\Delta)}{\Delta}\right) \geq 0.$$

Then
$$\overline{\theta} \geq \underline{\theta}.$$

\end{Lemma}

\begin{proof}

Let $\varepsilon > 0$.\\

Let us consider a couple $(t^*_{\varepsilon}, q^*_{\varepsilon})$ such that
$$ \underline{\theta}(t^*_{\varepsilon},q^*_{\varepsilon}) - \overline{\theta}(t^*_{\varepsilon},q^*_{\varepsilon}) - \varepsilon(T-t^*_{\varepsilon}) =
\sup_{(t,q) \in [\tau,T]\times\mathcal{Q}}\underline{\theta}(t,q) - \overline{\theta}(t,q) - \varepsilon(T-t).
$$

If $t^*_{\varepsilon}\neq T$, then
$$\partial_t \underline{\theta}(t^*_{\varepsilon},q^*_{\varepsilon}) - \partial_t \overline{\theta}(t^*_{\varepsilon},q^*_{\varepsilon}) + \varepsilon \leq 0.$$

Now, by using the definition of the functions $\underline{\theta}$ and $\overline{\theta}$, the above inequality gives
$$ 1_{q^*_{\varepsilon}<Q} H^b_{\xi}\left(\frac{\overline{\theta}(t^*_{\varepsilon},q^*_{\varepsilon}) - \overline{\theta}(t^*_{\varepsilon},q^*_{\varepsilon}+\Delta)}{\Delta}\right) + 1_{q^*_{\varepsilon}>-Q} H^a_{\xi}\left(\frac{\overline{\theta}(t^*_{\varepsilon},q^*_{\varepsilon}) - \overline{\theta}(t^*_{\varepsilon},q^*_{\varepsilon}-\Delta)}{\Delta}\right)$$$$ -  1_{q^*_{\varepsilon}<Q} H^b_{\xi}\left(\frac{\underline{\theta}(t^*_{\varepsilon},q^*_{\varepsilon}) - \underline{\theta}(t^*_{\varepsilon},q^*_{\varepsilon}+\Delta)}{\Delta}\right) - 1_{q^*_{\varepsilon}>-Q} H^a_{\xi}\left(\frac{\underline{\theta}(t^*_{\varepsilon},q^*_{\varepsilon}) - \underline{\theta}(t^*_{\varepsilon},q^*_{\varepsilon}-\Delta)}{\Delta}\right) \le -\varepsilon. $$

But, by definition of $(t^*_{\varepsilon},q^*_{\varepsilon})$, since $H^b_{\xi}$ and $H^a_{\xi}$ are decreasing functions, we have
$$1_{q^*_{\varepsilon}<Q} \left(H^b_{\xi}\left(\frac{\overline{\theta}(t^*_{\varepsilon},q^*_{\varepsilon}) - \overline{\theta}(t^*_{\varepsilon},q^*_{\varepsilon}+\Delta)}{\Delta}\right) - H^b_{\xi}\left(\frac{\underline{\theta}(t^*_{\varepsilon},q^*_{\varepsilon}) - \underline{\theta}(t^*_{\varepsilon},q^*_{\varepsilon}+\Delta)}{\Delta}\right)\right) \ge 0,$$
and
$$1_{q^*_{\varepsilon}>-Q} \left(H^a_{\xi}\left(\frac{\overline{\theta}(t^*_{\varepsilon},q^*_{\varepsilon}) - \overline{\theta}(t^*_{\varepsilon},q^*_{\varepsilon}+\Delta)}{\Delta}\right) - H^a_{\xi}\left(\frac{\underline{\theta}(t^*_{\varepsilon},q^*_{\varepsilon}) - \underline{\theta}(t^*_{\varepsilon},q^*_{\varepsilon}+\Delta)}{\Delta}\right)\right) \ge 0.$$

This leads to $0  \leq - \varepsilon$. By contradiction, we must have $t^*_{\varepsilon} = T$.\\

Therefore,
$$\sup_{(t,q) \in [\tau,T]\times\mathcal{Q}}\underline{\theta}(t,q) - \overline{\theta}(t,q) - \varepsilon(T-t) = \underline{\theta}(T,q^*_{\varepsilon}) - \overline{\theta}(T,q^*_{\varepsilon}) \le 0.$$

As a consequence, $\forall (t,q) \in [\tau,T)\times \mathcal{Q}$,
$$\underline{\theta}(t,q) - \overline{\theta}(t,q) \le \varepsilon T.$$

By sending $\varepsilon$ to $0$, we obtain $\underline{\theta} \le \overline{\theta}$.\qed\\

\end{proof}

Let us now come to the existence and uniqueness of a solution to Eq.~(\ref{sec3:thetagen}) with terminal condition (\ref{sec3:thetagenCT}).

\begin{Theorem}
\label{sec3:theotheta}
  There exists a unique function $\theta : [0,T]\times \mathcal{Q} \rightarrow \mathbb{R}$, $C^1$ in time, solution of Eq.~(\ref{sec3:thetagen}) with terminal condition (\ref{sec3:thetagenCT}).
\end{Theorem}

\begin{proof}

Eq.~(\ref{sec3:thetagen}) with terminal condition (\ref{sec3:thetagenCT}) can be regarded as a backward Cauchy problem. Since $H^b_\xi$ and $H^a_\xi$ are functions of class $C^1$, by Cauchy-Lipschitz, there exists $\tau \in [0,T)$ and a function $\theta : (\tau,T]\times \mathcal{Q} \rightarrow \mathbb{R}$, $C^1$ in time, solution of Eq.~(\ref{sec3:thetagen}) on $(\tau,T]$ with terminal condition (\ref{sec3:thetagenCT}).\\

It is straightforward to verify that $\forall q \in \mathcal{Q}, t \in (\tau,T] \mapsto \theta(t,q)  + \frac 12 \gamma \sigma^2 q^2 (T-t)$ is a decreasing function. Therefore, the only reason why there would not be a global solution on $[0,T]$ is because $\sup_{q \in \mathcal{Q}} \theta(t,q)$ blows up at $\tau>0$. However, by using Lemma~\ref{sec3:cp}, we know that $$\overline{\theta}(t,q) = (H^b_\xi(0) + H^a_\xi(0))(T-t)$$ defines a supersolution of Eq.~(\ref{sec3:thetagen}) with terminal condition (\ref{sec3:thetagenCT}), and therefore that $\sup_{q \in \mathcal{Q}} \theta(t,q) \le (H^b_\xi(0) + H^a_\xi(0))(T-t)$ cannot blow up in finite time.\\

The conclusion is that $\theta$ is in fact defined on $[0,T]\times \mathcal{Q}$. Uniqueness comes then for the Cauchy-Lipschitz theorem.\qed\\
\end{proof}

The existence (and uniqueness) of a function $\theta$ solution of Eq.~(\ref{sec3:thetagen}) with terminal condition~(\ref{sec3:thetagenCT}) enables us to find a solution to the HJB equation associated with Model A or Model~B. We will use a verification argument in the next subsection in order to prove that the solution to the HJB equation we obtain by this way is indeed the value function of the stochastic optimal control problem under consideration. However, before that, a remark needs to be made on $\theta$ and on Eq.~(\ref{sec3:thetagen}) in the specific case -- often (not to say almost always) used in the academic literature -- of exponential intensities.\\

If we have $\Lambda^b(\delta)= \Lambda^a(\delta) = A e^{-k\delta} =: \Lambda(\delta)$, then we obtain (by straightforward computations)
$$H_\xi(p) := H^b_\xi(p) = H^a_\xi(p) = \frac{A\Delta}{k} C_\xi \exp(-kp),$$
where $$C_\xi = \begin{cases} \left(1 + \frac{\xi\Delta}{k}\right)^{-\frac{k}{\xi\Delta} - 1} &\mbox{if } \xi > 0 \\
e^{-1} & \mbox{if } \xi=0. \end{cases}$$
By using Eq.~(\ref{sec3:thetagen}), the function
$$ v: (t,q) \in [0,T] \times \mathcal{Q} \mapsto v_q(t) = \exp\left(\frac{k}{\Delta} \theta(t,q)\right)$$ is solution of the linear system of ordinary differential equations
\begin{equation}
\label{sec3:eqv}
\forall q \in \mathcal{Q}, \forall t \in [0,T], -\partial_t v_q(t) + \frac{1}{2\Delta}k\gamma\sigma^2 q^2 v_q(t) - A C_\xi \left(1_{q<Q} v_{q+\Delta}(t) + 1_{q>-Q} v_{q-\Delta}(t) \right) = 0,
\end{equation}
with terminal condition $\forall q \in \mathcal{Q}, v_q(T) = \exp\left(-\frac{k}{\Delta} \ell(|q|)\right)$.

\subsection{Verification argument}

We are now ready to solve the stochastic optimal control problems associated with Model~A and Model B. We start with Model A.

\begin{Theorem}
\label{sec3:theomodelA}
Let us consider the solution $\theta$ of Eq.~(\ref{sec3:thetagen}) with terminal condition (\ref{sec3:thetagenCT}) for $\xi=\gamma$.\\

Then, $u : (t,x,q,S) \mapsto -\exp(-\gamma(x+qS + \theta(t,q)))$  defines a solution to Eq.~(\ref{sec3:HJBModelA}) with terminal condition (\ref{sec3:HJBModelACT}), and
$$u(t,x,q,S) = \sup_{(\delta^b_s)_{s\ge t}, (\delta^a_s)_{s\ge t}  \in \mathcal{A}(t)} \mathbb{E}\left[- \exp\left(-\gamma\left(X^{t,x,\delta^b,\delta^a}_T+q^{t,q,\delta^b,\delta^a}_T S^{t,S}_T- \ell(|q^{t,q,\delta^b,\delta^a}_T|)\right)\right) \right],$$
where $\mathcal{A}(t)$ is the set of predictable processes on $[t,T]$, bounded from below and where
$$dS^{t,S}_s = \sigma dW_s, \qquad S^{t,S}_t = S,$$
$$dX^{t,x,\delta^b,\delta^a}_s = (S_s + \delta^a_s)  \Delta dN^a_s - (S_s - \delta^b_s)  \Delta dN^b_s , \qquad  X^{t,x,\delta^b,\delta^a}_t = x,$$
$$dq^{t,q,\delta^b,\delta^a}_s = \Delta dN^b_s - \Delta dN^a_s, \qquad  q^{t,q,\delta^b,\delta^a}_t = q,$$
where the point processes $N^b$ and $N^a$ have stochastic intensity $(\lambda^b_s)_s$ and $(\lambda^a_s)_s$ given by $\lambda^b_s = \Lambda^b(\delta^b_s) 1_{q_{s-} < Q}$ and $\lambda^a_s = \Lambda^a(\delta^a_s) 1_{q_{s-} > -Q}$.\\

The optimal bid and ask quotes $S^b_t = S_t - \delta^{b*}_t$ (for $q_{t-}< Q$) and $S^a_t = S_t + \delta^{a*}_t$ (for $q_{t-}>-Q$) are characterized by\begin{equation}
\label{sec3:deltaoptimal}
\delta^{b*}_t = \tilde{\delta}^{b*}_\gamma\left(\frac{\theta(t,q_{t-}) - \theta(t,q_{t-}+\Delta)}{\Delta}\right) \text{ and } \delta^{a*}_t = \tilde{\delta}^{a*}_\gamma\left(\frac{\theta(t,q_{t-}) - \theta(t,q_{t-}-\Delta)}{\Delta}\right),
\end{equation}where the functions $\tilde{\delta}^{b*}_\gamma(\cdot)$ and $\tilde{\delta}_\gamma^{a*}(\cdot)$ are defined in Eqs.~(\ref{sec3:deltab}) and (\ref{sec3:deltaa}).
\end{Theorem}

\begin{proof}

Let us consider $t\in [0,T)$, and two processes $(\delta^b_s)_{s\ge t}$ and $(\delta^a_s)_{s\ge t}$ in  $\mathcal{A}(t)$. We have
\begin{equation}
\label{verif1}u(T,X^{t,x,\delta^b,\delta^a}_{T-},q^{t,q,\delta^b,\delta^a}_{T-}, S^{t,S}_T) = u(t,x,q,S) + \int_t^T \partial_t u(s,X^{t,x,\delta^b,\delta^a}_{s-},q^{t,q,\delta^b,\delta^a}_{s-}, S^{t,S}_s) ds
\end{equation}$$ + \sigma \int_t^T \partial_{S} u(s,X^{t,x,\delta^b,\delta^a}_{s-},q^{t,q,\delta^b,\delta^a}_{s-}, S^{t,S}_s) dW_s + \frac 12 \sigma^2 \int_t^T \partial^2_{SS} u(s,X^{t,x,\delta^b,\delta^a}_{s-},q^{t,q,\delta^b,\delta^a}_{s-}, S^{t,S}_s) ds$$$$ + \int_t^T \left(u(s,X^{t,x,\delta^b,\delta^a}_{s-} - \Delta S^{t,S}_s + \Delta \delta^b_s ,q^{t,q,\delta^b,\delta^a}_{s-} + \Delta, S^{t,S}_s) - u(s,X^{t,x,\delta^b,\delta^a}_{s-},q^{t,q,\delta^b,\delta^a}_{s-}, S^{t,S}_s)\right)dN^b_s$$
$$ + \int_t^T \left(u(s,X^{t,x,\delta^b,\delta^a}_{s-} + \Delta S^{t,S}_s + \Delta \delta^a_s ,q^{t,q,\delta^b,\delta^a}_{s-} - \Delta, S^{t,S}_s) - u(s,X^{t,x,\delta^b,\delta^a}_{s-},q^{t,q,\delta^b,\delta^a}_{s-}, S^{t,S}_s)\right)dN^a_s.$$
Let $C < 0$. If almost surely $\forall s \in [t,T), \delta^b_s \ge -C$, then
\begin{eqnarray*}
&&\mathbb{E}\left[\int_t^T \partial_S u\left(s,X^{t,x,\delta^b,\delta^a}_{s-},q^{t,q,\delta^b,\delta^a}_{s-}, S^{t,S}_s\right)
^2ds\right]\\
&\le& \mathbb{E}\left[\int_t^T{\gamma^2q^{t,q,\delta^b,\delta^a}_{s-}}^2\exp\left(-2\gamma\left(X^{t,x,\delta^b,\delta^a}_{s-} + q^{t,q,\delta^b,\delta^a}_{s-} S^{t,S}_s + \theta(s,q^{t,q,\delta^b,\delta^a}_{s-}) \right)\right) ds\right]\\
&\le& \gamma^2 Q^2\exp\left(2\gamma \|\theta\|_{L^\infty([t,T]\times\mathcal{Q})}\right)\\
&&\times\mathbb{E}\left[\int_t^T\exp\left(-2\gamma\left(x+qS + \int_t^s \delta^b_u dN^b_u + \int_t^s \delta^a_u dN^a_u + \int_t^s \sigma q^{t,q,\delta^b,\delta^a}_{u-} dW_u \right)\right) ds\right]\\
&\le& \gamma^2 Q^2 \exp\left(2\gamma \|\theta\|_{L^\infty([t,T]\times\mathcal{Q})}\right) \exp(-2\gamma (x+qS))\\
&&\times \mathbb{E}\left[\int_t^T \exp\left(-6\gamma\int_t^s \delta^b_u dN^b_u \right) ds\right]^{\frac 13} \mathbb{E}\left[\int_t^T \exp\left(-6\gamma\int_t^s \delta^a_u dN^a_u \right) ds\right]^{\frac 13}\\
&&\times \mathbb{E}\left[\int_t^T \exp\left(-6\gamma \int_t^s \sigma q^{t,q,\delta^b,\delta^a}_{u-} dW_u\right)ds\right]^{\frac 13}\\
&\le& \gamma^2 Q^2\exp\left(2\gamma \|\theta\|_{L^\infty([t,T]\times\mathcal{Q})}\right) \exp(-2\gamma (x+qS))\\
&&\times \mathbb{E}\left[\exp\left(6\gamma C(N^b_T-N^b_t)\right)(T-t)\right]^{\frac 13} \mathbb{E}\left[\exp\left(6\gamma C(N^a_T-N^a_t)\right)(T-t)\right]^{\frac 13}\\
&&\times \left(\int_t^T \mathbb{E}\left[\exp\left(-6\gamma \int_t^s \sigma q^{t,q,\delta^b,\delta^a}_{u-} dW_u - 18\gamma^2 \int_t^s \sigma^2 {q^{t,q,\delta^b,\delta^a}_{u-}}^2 du\right)\right]ds\right)^{\frac 13}\\
&&\times \exp\left(6\gamma^2 (T-t) \sigma^2 Q^2\right)\\
&\le& \gamma^2 Q^2\exp\left(2\gamma \|\theta\|_{L^\infty([t,T]\times\mathcal{Q})}\right) \exp(-2\gamma (x+qS))\\
&&\times \exp\left(\Lambda^b(-C)(T-t) \left(\exp(6\gamma C)-1\right) \right) \exp\left(\Lambda^a(-C)(T-t) \left(\exp(6\gamma C)-1\right) \right) (T-t)^{\frac 23} \\
&&\times  \exp\left(6\gamma^2 (T-t) \sigma^2 Q^2\right)(T-t)^{\frac 13}<+\infty.\\
\end{eqnarray*}
We also have
\begin{eqnarray*}
&&\mathbb{E}\left[\int_t^T |u(s,X^{t,x,\delta^b,\delta^a}_{s-},q^{t,q,\delta^b,\delta^a}_{s-}, S^{t,S}_s)
|\Lambda^b(\delta^b_s)1_{q^{t,q,\delta^b,\delta^a}_{s-}<Q}ds\right]\\
&\le& \Lambda^b(-C)\mathbb{E}\left[\int_t^T\exp\left(-\gamma\left(X^{t,x,\delta^b,\delta^a}_{s-} + q^{t,q,\delta^b,\delta^a}_{s-} S^{t,S}_s + \theta(s,q^{t,q,\delta^b,\delta^a}_{s-}) \right)\right) ds\right]\\
&\le& \Lambda^b(-C) \exp\left(\gamma \|\theta\|_{L^\infty([t,T]\times\mathcal{Q})}\right)\\
&&\times\mathbb{E}\left[\int_t^T\exp\left(-\gamma\left(x+qS + \int_t^s \delta^b_u dN^b_u + \int_t^s \delta^a_u dN^a_u + \int_t^s \sigma q^{t,q,\delta^b,\delta^a}_{u-} dW_u \right)\right) ds\right]\\
&\le& \Lambda^b(-C) \exp\left(\gamma \|\theta\|_{L^\infty([t,T]\times\mathcal{Q})}\right) \exp(-\gamma (x+qS))\\
&&\times \mathbb{E}\left[\int_t^T \exp\left(-3\gamma\int_t^s \delta^b_u dN^b_u \right) ds\right]^{\frac 13} \mathbb{E}\left[\int_t^T \exp\left(-3\gamma\int_t^s \delta^a_u dN^a_u \right) ds\right]^{\frac 13}\\
&&\times \mathbb{E}\left[\int_t^T \exp\left(-3\gamma \int_t^s \sigma q^{t,q,\delta^b,\delta^a}_{u-} dW_u\right)ds\right]^{\frac 13}\\
&\le& \Lambda^b(-C) \exp\left(\gamma \|\theta\|_{L^\infty([t,T]\times\mathcal{Q})}\right) \exp(-\gamma (x+qS))\\
&&\times \exp\left(\Lambda^b(-C)(T-t) \left(\exp(3\gamma C)-1\right) \right) \exp\left(\Lambda^a(-C)(T-t) \left(\exp(3\gamma C)-1\right) \right) (T-t)^{\frac 23} \\
&&\times  \exp\left(\frac 32\gamma^2 (T-t) \sigma^2 Q^2\right)(T-t)^{\frac 13}\\
&<&+\infty.
\end{eqnarray*}

Similarly
$$\mathbb{E}\left[\int_t^T |u(s,X^{t,x,\delta^b,\delta^a}_{s-} - \Delta S^{t,S}_s + \Delta \delta^b_s ,q^{t,q,\delta^b,\delta^a}_{s-} + \Delta, S^{t,S}_s)
|\Lambda^b(\delta^b_s)1_{q^{t,q,\delta^b,\delta^a}_{s-}<Q}ds\right] < +\infty,$$
$$\mathbb{E}\left[\int_t^T |u(s,X^{t,x,\delta^b,\delta^a}_{s-},q^{t,q,\delta^b,\delta^a}_{s-}, S^{t,S}_s)
|\Lambda^a(\delta^a_s)1_{q^{t,q,\delta^b,\delta^a}_{s-}>-Q}ds\right] < +\infty,$$
and
$$\mathbb{E}\left[\int_t^T |u(s,X^{t,x,\delta^b,\delta^a}_{s-} + \Delta S^{t,S}_s + \Delta \delta^a_s ,q^{t,q,\delta^b,\delta^a}_{s-} - \Delta, S^{t,S}_s)
|\Lambda^a(\delta^a_s)1_{q^{t,q,\delta^b,\delta^a}_{s-}>-Q}ds\right] < +\infty.$$

By taking expectations in Eq.~(\ref{verif1}), we obtain
$$\mathbb{E}\left[u(T,X^{t,x,\delta^b,\delta^a}_{T-},q^{t,q,\delta^b,\delta^a}_{T-}, S^{t,S}_T)\right] = u(t,x,q,S)$$$$ + \mathbb{E}\left[\int_t^T \left(\partial_t u(s,X^{t,x,\delta^b,\delta^a}_{s-},q^{t,q,\delta^b,\delta^a}_{s-}, S^{t,S}_s) + \frac 12 \sigma^2 \partial^2_{SS} u(s,X^{t,x,\delta^b,\delta^a}_{s-},q^{t,q,\delta^b,\delta^a}_{s-}, S^{t,S}_s)\right) ds\right]$$$$ + \mathbb{E}\left[\int_t^T \left(u(s,X^{t,x,\delta^b,\delta^a}_{s-} - \Delta S^{t,S}_s + \Delta \delta^b_s ,q^{t,q,\delta^b,\delta^a}_{s-} + \Delta, S^{t,S}_s)\right.\right.$$$$\left. - u(s,X^{t,x,\delta^b,\delta^a}_{s-},q^{t,q,\delta^b,\delta^a}_{s-}, S^{t,S}_s)\right)\Lambda^b(\delta^b_s)1_{q^{t,q,\delta^b,\delta^a}_{s-}<Q}ds$$
$$ + \int_t^T \left(u(s,X^{t,x,\delta^b,\delta^a}_{s-} + \Delta S^{t,S}_s + \Delta \delta^a_s ,q^{t,q,\delta^b,\delta^a}_{s-} - \Delta, S^{t,S}_s)\right.$$$$\left.\left. - u(s,X^{t,x,\delta^b,\delta^a}_{s-},q^{t,q,\delta^b,\delta^a}_{s-}, S^{t,S}_s)\right)\Lambda^a(\delta^a_s)1_{q^{t,q,\delta^b,\delta^a}_{s-}>-Q}ds\right].$$

By definition of $u$, we have therefore
$$\mathbb{E}\left[u(T,X^{t,x,\delta^b,\delta^a}_{T-},q^{t,q,\delta^b,\delta^a}_{T-}, S^{t,S}_T)\right] = u(t,x,q,S)$$$$ + \mathbb{E}\left[\int_t^T u(s,X^{t,x,\delta^b,\delta^a}_{s-},q^{t,q,\delta^b,\delta^a}_{s-}, S^{t,S}_s)\left( -\gamma \partial_t \theta(s,q^{t,q,\delta^b,\delta^a}_{s-})  + \frac 12 \gamma^2 \sigma^2 {q^{t,q,\delta^b,\delta^a}_{s-}}^2\right) ds\right]$$$$ + \mathbb{E}\Bigg[\int_t^T -u(s,X^{t,x,\delta^b,\delta^a}_{s-},q^{t,q,\delta^b,\delta^a}_{s-}, S^{t,S}_s)\Lambda^b(\delta^b_s)1_{q^{t,q,\delta^b,\delta^a}_{s-}<Q}$$$$\left(1-\exp\left(-\gamma\left(\Delta \delta^b_s + \theta(s,q^{t,q,\delta^b,\delta^a}_{s-}+\Delta) - \theta(s,q^{t,q,\delta^b,\delta^a}_{s-})\right)\right)\right)ds\Bigg]$$
$$ + \mathbb{E}\Bigg[\int_t^T -u(s,X^{t,x,\delta^b,\delta^a}_{s-},q^{t,q,\delta^b,\delta^a}_{s-}, S^{t,S}_s)\Lambda^a(\delta^a_s)1_{q^{t,q,\delta^b,\delta^a}_{s-}>-Q}$$$$\left(1-\exp\left(-\gamma\left(\Delta \delta^a_s + \theta(s,q^{t,q,\delta^b,\delta^a}_{s-}-\Delta) - \theta(s,q^{t,q,\delta^b,\delta^a}_{s-})\right)\right)\right)ds\Bigg].$$

By definition of $\theta$, we have the inequality
$$\mathbb{E}\left[u(T,X^{t,x,\delta^b,\delta^a}_{T},q^{t,q,\delta^b,\delta^a}_{T}, S^{t,S}_T)\right] = \mathbb{E}\left[u(T,X^{t,x,\delta^b,\delta^a}_{T-},q^{t,q,\delta^b,\delta^a}_{T-}, S^{t,S}_T)\right]  \le u(t,x,q,S),$$
i.e.,
$$\mathbb{E}\left[-\exp\left(-\gamma(X^{t,x,\delta^b,\delta^a}_{T}+q^{t,q,\delta^b,\delta^a}_{T} S^{t,S}_T - \ell(|q^{t,q,\delta^b,\delta^a}_{T}|))\right)\right] \le u(t,x,q,S).$$

Furthermore, by Lemma \ref{sec3:lemmaH}, there is equality in the above inequality if $(\delta^b_s)_{s\ge t}$ and $(\delta^a_s)_{s\ge t}$ are given (in closed-loop) by Eq.~(\ref{sec3:deltaoptimal}).\\

Therefore, $u$ is indeed the value function
$$u(t,x,q,S) = \sup_{(\delta^b_s)_{s\ge t}, (\delta^a_s)_{s\ge t}  \in \mathcal{A}(t)} \mathbb{E}\left[- \exp\left(-\gamma\left(X^{t,x,\delta^b,\delta^a}_T+q^{t,q,\delta^b,\delta^a}_T S^{t,S}_T- \ell(|q^{t,q,\delta^b,\delta^a}_T|)\right)\right) \right],$$
and the optimal quotes are given in closed-loop by Eq.~(\ref{sec3:deltaoptimal}).\qed\\

\end{proof}

For Model B, a similar result holds.

\begin{Theorem}
\label{sec3:theomodelB}
Let us consider the solution $\theta$ of Eq.~(\ref{sec3:thetagen}) with terminal condition (\ref{sec3:thetagenCT}) for $\xi=0$.\\

Then, $u : (t,x,q,S) \mapsto x+qS + \theta(t,q)$  defines a solution to Eq.~(\ref{sec3:HJBModelB}) with terminal condition (\ref{sec3:HJBModelBCT}), and
$$u(t,x,q,S) =$$$$ \sup_{(\delta^b_s)_{s\ge t}, (\delta^a_s)_{s\ge t}  \in \mathcal{A}(t)} \mathbb{E}\left[X^{t,x,\delta^b,\delta^a}_T+q^{t,q,\delta^b,\delta^a}_T S^{t,S}_T- \ell(|q^{t,q,\delta^b,\delta^a}_T|) - \frac 12 \gamma \sigma^2 \int_t^T {q^{t,q,\delta^b,\delta^a}_s}^2 ds  \right],$$
where $\mathcal{A}(t)$ is the set of predictable processes on $[t,T]$, bounded from below and where
$$dS^{t,S}_s = \sigma dW_s, \qquad S^{t,S}_t = S,$$
$$dX^{t,x,\delta^b,\delta^a}_s = (S_s + \delta^a_s)  \Delta dN^a_s - (S_s - \delta^b_s)  \Delta dN^b_s , \qquad  X^{t,x,\delta^b,\delta^a}_t = x,$$
$$dq^{t,q,\delta^b,\delta^a}_s = \Delta dN^b_s - \Delta dN^a_s, \qquad  q^{t,q,\delta^b,\delta^a}_t = q,$$
where the point processes $N^b$ and $N^a$ have stochastic intensity $(\lambda^b_s)_s$ and $(\lambda^a_s)_s$ given by $\lambda^b_s = \Lambda^b(\delta^b_s) 1_{q_{s-} < Q}$ and $\lambda^a_s = \Lambda^a(\delta^a_s) 1_{q_{s-} > -Q}$.\\

The optimal bid and ask quotes $S^b_t = S_t - \delta^{b*}_t$ (for $q_{t-}< Q$) and $S^a_t = S_t + \delta^{a*}_t$ (for $q_{t-}>-Q$) are given by
\begin{equation}
\label{sec3:deltaoptimal2}
\delta^{b*}_t = \tilde{\delta}^{b*}_0\left(\frac{\theta(t,q_{t-}) - \theta(t,q_{t-}+\Delta)}{\Delta}\right) \text{ and } \delta^{a*}_t = \tilde{\delta}^{a*}_0\left(\frac{\theta(t,q_{t-}) - \theta(t,q_{t-}-\Delta)}{\Delta}\right),
 \end{equation}
 where the functions $\tilde{\delta}^{b*}_{0}(\cdot)$ and $\tilde{\delta}^{a*}_0(\cdot)$ are defined in Eqs.~(\ref{sec3:deltab}) and (\ref{sec3:deltaa}).
\end{Theorem}
\vspace{10pt}

\begin{proof}

Let us consider $t\in [0,T)$, and two processes $(\delta^b_s)_{s\ge t}$ and $(\delta^a_s)_{s\ge t}$ in  $\mathcal{A}(t)$. We have
\begin{equation}
\label{verif2}u(T,X^{t,x,\delta^b,\delta^a}_{T-},q^{t,q,\delta^b,\delta^a}_{T-}, S^{t,S}_T) = u(t,x,q,S) + \int_t^T \partial_t u(s,X^{t,x,\delta^b,\delta^a}_{s-},q^{t,q,\delta^b,\delta^a}_{s-}, S^{t,S}_s) ds
\end{equation}$$ + \sigma \int_t^T \partial_{S} u(s,X^{t,x,\delta^b,\delta^a}_{s-},q^{t,q,\delta^b,\delta^a}_{s-}, S^{t,S}_s) dW_s + \frac 12 \sigma^2 \int_t^T \partial^2_{SS} u(s,X^{t,x,\delta^b,\delta^a}_{s-},q^{t,q,\delta^b,\delta^a}_{s-}, S^{t,S}_s) ds$$$$ + \int_t^T \left(u(s,X^{t,x,\delta^b,\delta^a}_{s-} - \Delta S^{t,S}_s + \Delta \delta^b_s ,q^{t,q,\delta^b,\delta^a}_{s-} + \Delta, S^{t,S}_s) - u(s,X^{t,x,\delta^b,\delta^a}_{s-},q^{t,q,\delta^b,\delta^a}_{s-}, S^{t,S}_s)\right)dN^b_s$$
$$ + \int_t^T \left(u(s,X^{t,x,\delta^b,\delta^a}_{s-} + \Delta S^{t,S}_s + \Delta \delta^a_s ,q^{t,q,\delta^b,\delta^a}_{s-} - \Delta, S^{t,S}_s) - u(s,X^{t,x,\delta^b,\delta^a}_{s-},q^{t,q,\delta^b,\delta^a}_{s-}, S^{t,S}_s)\right)dN^a_s.$$

If almost surely $\forall s \in [t,T), \delta^b_s \ge -C$, then
\vspace{10pt}
\begin{eqnarray*}
&&\mathbb{E}\left[\int_t^T \partial_S u\left(s,X^{t,x,\delta^b,\delta^a}_{s-},q^{t,q,\delta^b,\delta^a}_{s-}, S^{t,S}_s\right)
^2ds\right]\\
&\le& \mathbb{E}\left[\int_t^T {q^{t,q,\delta^b,\delta^a}_{s-}}^2ds\right]\\
&\le& Q^2(T-t)\\
&<&+\infty.
\end{eqnarray*}

We also have
\begin{eqnarray*}
&&\mathbb{E}\left[\int_t^T |u(s,X^{t,x,\delta^b,\delta^a}_{s-} - \Delta S^{t,S}_s + \Delta \delta^b_s ,q^{t,q,\delta^b,\delta^a}_{s-} + \Delta, S^{t,S}_s) \right.\\
&&\left.- u(s,X^{t,x,\delta^b,\delta^a}_{s-},q^{t,q,\delta^b,\delta^a}_{s-}, S^{t,S}_s)|\Lambda^b(\delta^b_s)1_{q^{t,q,\delta^b,\delta^a}_{s-}<Q}ds\right]\\
&\le& \mathbb{E}\left[\int_t^T \Lambda^b(\delta^b_s) |\Delta \delta^b_s + \theta(s,q^{t,q,\delta^b,\delta^a}_{s-}+\Delta) - \theta(s,q^{t,q,\delta^b,\delta^a}_{s-})|  ds\right]\\
&\le& 2 \Lambda^b(-C) \|\theta\|_{L^\infty([t,T]\times\mathcal{Q})}(T-t) + \Delta (T-t) \sup_{\delta>-C} |\delta| \Lambda^b(\delta)\\
&\le& 2 \Lambda^b(-C) \|\theta\|_{L^\infty([t,T]\times\mathcal{Q})}(T-t) +  (T-t) \max( \Delta C \Lambda^b(-C), H^b_0(0))\\
&<&+\infty.
\end{eqnarray*}

Similarly
$$\mathbb{E}\left[\int_t^T |u(s,X^{t,x,\delta^b,\delta^a}_{s-} + \Delta S^{t,S}_s + \Delta \delta^a_s ,q^{t,q,\delta^b,\delta^a}_{s-} - \Delta, S^{t,S}_s) \right.$$$$\left.- u(s,X^{t,x,\delta^b,\delta^a}_{s-},q^{t,q,\delta^b,\delta^a}_{s-}, S^{t,S}_s)|\Lambda^a(\delta^a_s)1_{q^{t,q,\delta^b,\delta^a}_{s-}>-Q}ds\right] < +\infty.$$

By taking expectations in Eq.~(\ref{verif2}), we obtain
$$\mathbb{E}\left[u(T,X^{t,x,\delta^b,\delta^a}_{T-},q^{t,q,\delta^b,\delta^a}_{T-}, S^{t,S}_T)\right] = u(t,x,q,S)$$$$ + \mathbb{E}\left[\int_t^T \left(\partial_t u(s,X^{t,x,\delta^b,\delta^a}_{s-},q^{t,q,\delta^b,\delta^a}_{s-}, S^{t,S}_s) + \frac 12 \sigma^2 \partial^2_{SS} u(s,X^{t,x,\delta^b,\delta^a}_{s-},q^{t,q,\delta^b,\delta^a}_{s-}, S^{t,S}_s)\right) ds\right]$$$$ + \mathbb{E}\left[\int_t^T \left(u(s,X^{t,x,\delta^b,\delta^a}_{s-} - \Delta S^{t,S}_s + \Delta \delta^b_s ,q^{t,q,\delta^b,\delta^a}_{s-} + \Delta, S^{t,S}_s)\right.\right.$$$$\left. - u(s,X^{t,x,\delta^b,\delta^a}_{s-},q^{t,q,\delta^b,\delta^a}_{s-}, S^{t,S}_s)\right)\Lambda^b(\delta^b_s)1_{q^{t,q,\delta^b,\delta^a}_{s-}<Q}ds$$
$$ + \int_t^T \left(u(s,X^{t,x,\delta^b,\delta^a}_{s-} + \Delta S^{t,S}_s + \Delta \delta^a_s ,q^{t,q,\delta^b,\delta^a}_{s-} - \Delta, S^{t,S}_s)\right.$$$$\left.\left. - u(s,X^{t,x,\delta^b,\delta^a}_{s-},q^{t,q,\delta^b,\delta^a}_{s-}, S^{t,S}_s)\right)\Lambda^a(\delta^a_s)1_{q^{t,q,\delta^b,\delta^a}_{s-}>-Q}ds\right].$$

By definition of $u$, we have therefore
$$\mathbb{E}\left[u(T,X^{t,x,\delta^b,\delta^a}_{T-},q^{t,q,\delta^b,\delta^a}_{T-}, S^{t,S}_T)\right] = u(t,x,q,S) + \mathbb{E}\left[\int_t^T \partial_t \theta(s,q^{t,q,\delta^b,\delta^a}_{s-})  ds\right]$$$$ + \mathbb{E}\Bigg[\int_t^T \Lambda^b(\delta^b_s)1_{q^{t,q,\delta^b,\delta^a}_{s-}<Q}\left(\Delta \delta^b_s + \theta(s,q^{t,q,\delta^b,\delta^a}_{s-}+\Delta) - \theta(s,q^{t,q,\delta^b,\delta^a}_{s-})\right)ds\Bigg]$$
$$ + \mathbb{E}\Bigg[\int_t^T \Lambda^a(\delta^a_s)1_{q^{t,q,\delta^b,\delta^a}_{s-}>-Q}\left(\Delta \delta^a_s + \theta(s,q^{t,q,\delta^b,\delta^a}_{s-}-\Delta) - \theta(s,q^{t,q,\delta^b,\delta^a}_{s-})\right)ds\Bigg].$$

By definition of $\theta$, we have the inequality
\begin{eqnarray*}
&&\mathbb{E}\left[u(T,X^{t,x,\delta^b,\delta^a}_{T},q^{t,q,\delta^b,\delta^a}_{T}, S^{t,S}_T)\right]\\
&=& \mathbb{E}\left[u(T,X^{t,x,\delta^b,\delta^a}_{T-},q^{t,q,\delta^b,\delta^a}_{T-}, S^{t,S}_T)\right]\\
&\le& u(t,x,q,S) + \mathbb{E}\Bigg[\int_t^T \frac 12 \gamma \sigma^2 {q^{t,q,\delta^b,\delta^a}_{s-}}^2 ds\Bigg]\\
&\le& u(t,x,q,S) + \mathbb{E}\Bigg[\int_t^T \frac 12 \gamma \sigma^2 {q^{t,q,\delta^b,\delta^a}_{s}}^2 ds\Bigg]\\
\end{eqnarray*}
i.e.,
$$\mathbb{E}\left[X^{t,x,\delta^b,\delta^a}_{T}+q^{t,q,\delta^b,\delta^a}_{T} S^{t,S}_T - \ell(|q^{t,q,\delta^b,\delta^a}_{T}|) - \int_t^T \frac 12 \gamma \sigma^2 {q^{t,q,\delta^b,\delta^a}_{s}}^2 ds \right] \le u(t,x,q,S).$$

Furthermore, by Lemma \ref{sec3:lemmaH}, there is equality in the above inequality if $(\delta^b_s)_{s\ge t}$ and $(\delta^a_s)_{s\ge t}$ are given (in closed-loop) by Eq.~(\ref{sec3:deltaoptimal2}).\\

Therefore, $u$ is indeed the value function
$$u(t,x,q,S) =$$$$ \sup_{(\delta^b_s)_{s\ge t}, (\delta^a_s)_{s\ge t}  \in \mathcal{A}(t)} \mathbb{E}\left[X^{t,x,\delta^b,\delta^a}_{T}+q^{t,q,\delta^b,\delta^a}_{T} S^{t,S}_T - \ell(|q^{t,q,\delta^b,\delta^a}_{T}|) - \int_t^T \frac 12 \gamma \sigma^2 {q^{t,q,\delta^b,\delta^a}_{s}}^2 ds\right],$$
and the optimal quotes are given in closed-loop by Eq.~(\ref{sec3:deltaoptimal2}).\qed\\

\end{proof}

\subsection{Comments on the results}

In both Model A and Model B, the dynamic optimization problem faced by the market maker was initially characterized by a Hamilton-Jacobi-Bellman equation with 4 variables: the time $t$, and 3 state variables (the cash $x$, the inventory $q$, and the reference price $S$). Computing a numerical approximation for the solution of  a 4-dimensional HJB equation such as Eq.~(\ref{sec3:HJBModelA}) or Eq.~(\ref{sec3:HJBModelB}) is always time-consuming. Therefore, the results obtained in Theorem 3.2 and Theorem 3.3 are very useful: they state that the optimal quotes of a market maker in both Model A and Model B can in fact be computed by solving a tridiagonal system of nonlinear ordinary differential equations. This corresponds to a reduction of the dimensionality of the problem from 4 to 2. Furthermore, the systems of nonlinear ordinary differential equations are similar for Model A and Model B: they correspond to Eq.~(\ref{sec3:thetagen}) -- with terminal condition (\ref{sec3:thetagenCT}) -- with $\xi = \gamma$ for Model A and with $\xi = 0$ for Model B.\\

The objective functions of Model A and Model B lead to similar equations, but it is interesting to understand the differences between the two modeling approaches. In fact, the penalization term $$ \frac 12 \gamma \sigma^2 \int_0^T q_t^2 dt$$ in Model B leads to the term $\frac 12 \gamma \sigma^2 q^2$ in the ODEs characterizing $\theta$ (when $\xi = 0$), and this term arises also in the ODE associated with Model A (when $\xi=\gamma$) because of the market maker's aversion to price risk. However, in Model A, the market maker is not only averse to price risk, but also to the risk of not finding a counterparty to trade with -- this is what we call non-execution risk. There is indeed a source of risk coming from the process $(W_t)_t$, and another source of risk coming from the processes $(N_t^b)_t$ and $(N^a_t)_t$, and risk aversion in Model A applies to both kinds of risk. In other words, things work as if the market maker of Model A was risk averse to both kinds of risk, while the market maker of Model B is only averse to the risk associated with price changes. In particular, the parameter $\xi$ can be regarded as some form of risk aversion parameter applying to non-execution risk only: it is equal to $\gamma$ in the case of Model A, and equal to $0$ in the case of Model B.\\

\section{Closed-form and almost-closed-form approximations}

In \cite{gueant2013dealing}, the authors show in the specific case where $\Lambda^b(\delta)= \Lambda^a(\delta) = A e^{-k\delta} =: \Lambda(\delta)$ that there is an asymptotic regime far from $T$ for the optimal quotes in Model A.\footnote{The authors of \cite{gueant2013dealing}  use the linear system of ODEs (\ref{sec3:eqv}) in the case $\Delta = 1$ and $\xi=\gamma$.} In other words, far from the terminal time $T$, the optimal quotes in \cite{gueant2013dealing} are well approximated by functions that only depend on the inventory $q$ -- and not on the time variable $t$. In practice, in markets (such as most dealer-driven OTC markets) for which there is no natural terminal time $T$, this result is not surprising -- even, somehow, reassuring -- and only the asymptotic formula should be used. Furthermore, the authors of \cite{gueant2013dealing} proposed closed-form approximations for the asymptotic values of the optimal quotes. In this section, we propose new approximation formulas which generalize those obtained in \cite{gueant2013dealing} to a more general set of intensity functions, and to both Model~A and Model~B (only Model~A was considered in \cite{gueant2013dealing}). These more general approximations are based on heuristic arguments, and we will see in the numerical experiments of Section 6 when they are (or are not) satisfactory.

\subsection{Approximation with an elliptic partial differential equation}

To compute the optimal quotes given in Eqs.~(\ref{sec3:deltaoptimal}) and (\ref{sec3:deltaoptimal2}), the first step consists in computing the function $\theta$ solution of the system of ODEs (\ref{sec3:thetagen}), with terminal condition~(\ref{sec3:thetagenCT}). In order to approximate the optimal quotes, we first approximate therefore the function $\theta$.\\

To carry out our reasoning, we suppose that the intensity functions $\Lambda^b$ and $\Lambda^a$ are identical (equal to $\Lambda$), and that $H_\xi : = H^b_\xi = H^a_\xi$ verifies $H_\xi''(0)  > 0$.\footnote{The condition $H''_{\xi}(0) > 0$ is always verified when $\xi = 0$ (see the proof of Lemma \ref{sec3:lemmaH}). A sufficient condition in general is
$$\forall \delta \in \mathbb{R},\quad  \xi \Delta \frac{\Lambda(\delta)^2 \Lambda''(\delta)}{\Lambda'(\delta)^3} < 1.$$ This condition (obtained by using the expression of $H''_\xi$ in the proof of Lemma \ref{sec3:lemmaH}) is verified for instance if $\Lambda$ is convex (exponential intensities enter this category).}\\

Our heuristic reasoning consists in replacing the function $\theta : [0,T]\times \mathcal{Q} \to \mathbb{R}$ by a function $\tilde{\theta} : [0,T]\times \mathbb{R} \to \mathbb{R}$ and to replace the system of ODEs (\ref{sec3:thetagen}) characterizing $\theta$, i.e.,
$$
0=-\partial_t \theta(t,q) + \frac{1}{2} \gamma \sigma^2 q^2
$$
$$-  1_{q<Q} H^b_{\xi}\left(\frac{\theta(t,q) - \theta(t,q+\Delta)}{\Delta}\right) - 1_{q>-Q} H^a_{\xi}\left(\frac{\theta(t,q) - \theta(t,q-\Delta)}{\Delta}\right)$$
by the PDE
  $$0 = - \partial_t \tilde\theta(t,q) + \frac 12 \gamma \sigma^2 q^2 - 2H_{\xi}\left(0\right)$$\begin{equation}\label{sec4:thetatilde} - H_{\xi}''(0) (\partial_q \tilde\theta(t,q))^2 + \Delta H_{\xi}'(0) \partial^2_{qq} \tilde\theta(t,q).\end{equation}

This PDE comes from an expansion in $\epsilon$ of the expression\footnote{We remove here the boundaries associated with $-Q$ and $Q$.}
$$
0=-\partial_t \tilde{\theta}(t,q) + \frac{1}{2} \gamma \sigma^2 q^2
$$
$$-  H_{\xi}\left(\frac{\tilde{\theta}(t,q) - \tilde{\theta}(t,q+\epsilon\Delta)}{\Delta}\right) - H_{\xi}\left(\frac{\tilde\theta(t,q) - \tilde\theta(t,q-\epsilon\Delta)}{\Delta}\right),$$ applied to $\epsilon =1$.\footnote{Another way to see this expansion is to consider an expansion of order 2 in $\Delta$ (an expansion of order 1 would correspond, after rescaling $H_\xi$, to a fluid-limit regime where non-execution risk vanishes) combined with an approximation of $H_{\xi}$ by using the first three terms of its Taylor expansion (in 0).}\\

We have indeed
\begin{eqnarray*}
&& H_{\xi}\left(\frac{\tilde{\theta}(t,q) - \tilde{\theta}(t,q+\epsilon\Delta)}{\Delta}\right) + H_{\xi}\left(\frac{\tilde\theta(t,q) - \tilde\theta(t,q-\epsilon\Delta)}{\Delta}\right)\\
&=& H_{\xi}\left(-\epsilon \partial_q \tilde{\theta}(t,q) - \frac 12 \epsilon^2 \Delta \partial^2_{qq}\tilde{\theta}(t,q) + o(\epsilon^2)\right) + H_{\xi}\left(\epsilon \partial_q \tilde{\theta}(t,q) - \frac 12 \epsilon^2 \Delta \partial^2_{qq}\tilde{\theta}(t,q) + o(\epsilon^2)\right)\\
&=& 2 H_{\xi}(0) - \epsilon^2 \Delta H'_{\xi}(0) \partial^2_{qq}\tilde{\theta}(t,q) + \epsilon^2 H''_{\xi}(0) \left(\partial_q \tilde{\theta}(t,q)\right)^2 + o(\epsilon^2).\\
\end{eqnarray*}

By considering
$$\tilde{v}(t,q) = \exp\left(-\frac{H''_{\xi}(0)}{\Delta H'_{\xi}(0) } \tilde{\theta}(t,q)\right),$$
the nonlinear PDE (\ref{sec4:thetatilde}) becomes the linear PDE\footnote{One can see the proximity with Eq.~(\ref{sec3:eqv}).}
\begin{equation}
\label{sec4:v}0 = \partial_t \tilde{v}(t,q) - \frac{H''_{\xi}(0)}{\Delta H'_{\xi}(0)}\left( 2 H_{\xi}(0)  - \frac 12 \gamma \sigma^2 q^2\right) \tilde{v}(t,q) -\Delta H'_{\xi}(0) \partial^2_{qq} \tilde{v}(t,q),\end{equation}
and the terminal condition relevant with our problem is
$$\tilde{v}(T,q) = \exp\left(\frac{H''_{\xi}(0)}{\Delta H'_{\xi}(0) } \ell(|q|)\right).$$

Eq.~(\ref{sec4:v}) is a linear PDE and it can be studied using basic tools of spectral theory. Our goal is to study the asymptotic behavior of $\tilde{v}(t,q)$ when $T$ tends to infinity, and to use the formulas obtained in this asymptotic regime in order to approximate successively $\tilde{v}$, $\tilde{\theta}$, $\theta$, and ultimately the optimal quotes $(\delta_t^{b*})_t$ and $(\delta^{a*}_t)_t$.\\

\subsection{Generalization of the Gu\'eant-Lehalle-Fernandez-Tapia's formulas}

By classical spectral theory,\footnote{The basic reasoning consists in proving that the operator $\tilde{v} \mapsto -\frac 12\frac{H''_{\xi}(0)}{\Delta H'_{\xi}(0)} \gamma \sigma^2 q^2 \tilde{v} + \Delta H'_{\xi}(0) \partial^2_{qq} \tilde{v}$ is a positive self-adjoint operator with a compact inverse (see Chapter 6 of \cite{brezis} for more details). Therefore, this operator can be diagonalized in an orthonormal basis. Its minimum eigenvalue can be shown to be simple by using the same methodology as in \cite{gueant2013dealing}.} we know that $$\tilde{v}(t,q) \sim_{T \to +\infty} (\tilde{v}(T,\cdot),\tilde{f}^0)\tilde{f}^0(q) \exp\left(\nu (T-t)\right),$$ where $\nu$ and $\tilde{f}^0$ are respectively the minimum and a minimizer of the functional
$$ \tilde{f} \in \lbrace \tilde{g} \in H^1(\mathbb{R}), \|\tilde{g}\|_{L^2(\mathbb{R})}=1\rbrace \mapsto \int_{-\infty}^{\infty} \left(\alpha x^2 \tilde{f}(x)^2 + \eta \tilde{f}'(x)^2\right) dx,$$
   with $$\alpha = - \frac 12 \frac{H''_{\xi}(0)}{\Delta H'_{\xi}(0)} \gamma \sigma^2 \text{ and } \eta = -\Delta H'_{\xi}(0),$$ and where $(\cdot,\cdot)$ designates the scalar product in $L^2(\mathbb{R})$.\\

In particular, 
   $$\tilde{f}^0(q) \propto \exp\left(-\frac{1}{2}\sqrt{\frac{\alpha}{\eta}}q^2\right).$$

From 
\begin{eqnarray*}
\tilde{v}(t,q) & \sim_{T \to +\infty}& C \exp\left(-\frac{1}{2}\sqrt{\frac{\alpha}{\eta}}q^2\right) \exp(\nu (T-t)),\\
\end{eqnarray*}
where $C$ is a constant, independent of $(t,q)$, we deduce: 
\begin{eqnarray*}
\tilde{\theta}(t,q) +  \frac{\Delta H'_{\xi}(0) }{H''_{\xi}(0)} \nu(T-t) &\to_{T \to +\infty}& - \frac{\Delta H'_{\xi}(0) }{H''_{\xi}(0)}\left(\log(C) -\frac{1}{2}\sqrt{\frac{\alpha}{\eta}}q^2\right)\\
\end{eqnarray*}
i.e.,
\begin{eqnarray*}
\tilde{\theta}(t,q) +  \frac{\Delta H'_{\xi}(0) }{H''_{\xi}(0)} \nu(T-t) &\to_{T \to +\infty}& - \frac{\Delta H'_{\xi}(0) }{H''_{\xi}(0)}\left(\log(C) -\frac{1}{2}\sqrt{\frac{\gamma\sigma^2}{2H''_{\xi}(0)}}q^2\right).\\
\end{eqnarray*}

As a consequence, we consider the approximations
$$\frac{\theta(t,q) - \theta(t,q+\Delta)}{\Delta} \simeq \frac{2q+\Delta}{2} \sqrt{\frac{\gamma\sigma^2}{2H''_{\xi}(0)}}$$
and
$$\frac{\theta(t,q) - \theta(t,q-\Delta)}{\Delta} \simeq -\frac{2q-\Delta}{2} \sqrt{\frac{\gamma\sigma^2}{2H''_{\xi}(0)}}.$$

These approximations are independent of $t$ and of the final penalty function $\ell$. They can be plugged into Eqs.~(\ref{sec3:deltaoptimal}) and (\ref{sec3:deltaoptimal2}) to obtain the general approximation formulas
\begin{equation}
\label{sec4:gueantformulab}\delta^{b*}_t \simeq \delta_{\text{approx}}^{b*}(q_{t-}) := \tilde{\delta}^{*}_\xi\left(\frac{2q_{t-}+\Delta}{2} \sqrt{\frac{\gamma\sigma^2}{2H''_{\xi}(0)}}\right)
\end{equation}
and
\begin{equation}
\label{sec4:gueantformulaa}
\delta^{a*}_t \simeq \delta_{\text{approx}}^{a*}(q_{t-}) := \tilde{\delta}^{*}_\xi\left(-\frac{2q_{t-}-\Delta}{2} \sqrt{\frac{\gamma\sigma^2}{2H''_{\xi}(0)}}\right),
\end{equation} where
\begin{equation}
\label{sec4:deltatilde}
\tilde{\delta}^{*}_\xi(p) = {\Lambda}^{-1}\left(\xi H_{\xi}(p) - \frac{{H_{\xi}}'(p)}{\Delta}\right).
\end{equation}

In particular, if $\Lambda(\delta) = A e^{-k\delta}$, then
$$\tilde{\delta}_\xi^*(p) = \begin{cases} p + \frac{1}{\xi \Delta} \log\left(1 + \frac{\xi\Delta}{k}\right) &\mbox{if } \xi > 0 \\
p + \frac 1k & \mbox{if } \xi=0, \end{cases}$$
and we obtain
\begin{equation}
\label{sec4:glftformulab}\delta_{\text{approx}}^{b*}(q) =  \begin{cases} \frac{1}{\xi \Delta} \log\left(1 + \frac{\xi\Delta}{k}\right) + \frac{2q+\Delta}{2} \sqrt{\frac{\gamma\sigma^2}{2A\Delta k} \left(1 + \frac{\xi\Delta}{k}\right)^{\frac{k}{\xi\Delta} + 1} } &\mbox{if } \xi > 0\\
\frac{1}k + \frac{2q+\Delta}{2} \sqrt{\frac{\gamma\sigma^2e}{2A\Delta k} } & \mbox{if } \xi=0, \end{cases}
\end{equation}
and
\begin{equation}
\label{sec4:glftformulaa}
\delta_{\text{approx}}^{a*}(q) = \begin{cases} \frac{1}{\xi \Delta} \log\left(1 + \frac{\xi\Delta}{k}\right) - \frac{2q-\Delta}{2} \sqrt{\frac{\gamma\sigma^2}{2A\Delta k} \left(1 + \frac{\xi\Delta}{k}\right)^{\frac{k}{\xi\Delta} + 1} } &\mbox{if } \xi > 0\\
\frac{1}k - \frac{2q-\Delta}{2} \sqrt{\frac{\gamma\sigma^2e}{2A\Delta k} } & \mbox{if } \xi=0. \end{cases}
\end{equation}

In particular, we recover, in the specific case where $\Delta = 1$ and $\xi=\gamma$, the Gu\'eant-Lehalle-Fernandez-Tapia's formula of \cite{gueant2013dealing} and \cite{gueantbook} often used in the industry.

\subsection{Comments on the approximations}

The approximations obtained above deserve a few comments. First, in the general case (i.e., even when the intensity function $\Lambda$ is not exponential), the approximations are almost in closed form, in the sense that they are only functions of the parameters and of transforms of $\Lambda$. In practice, one simply needs to compute $\Lambda^{-1}$, $H_\xi$, $H'_\xi$, and $H''_\xi$, in order to compute the approximations (\ref{sec4:gueantformulab}) and (\ref{sec4:gueantformulaa}). Second, the above approximations enable to better understand the optimal strategy of a market maker, and the role played by the different parameters. In particular, they enable to better understand the different types of risk faced by a market maker.\\

By using Eqs.~(\ref{sec4:gueantformulab}) and (\ref{sec4:gueantformulaa}), we see that
$$
\frac{d\ }{dq}\delta_{\text{approx}}^{b*}(q) = \sqrt{\frac{\gamma\sigma^2}{2H''_{\xi}(0)}}\tilde{\delta}^{*'}_\xi\left(\frac{2q+\Delta}{2} \sqrt{\frac{\gamma\sigma^2}{2H''_{\xi}(0)}}\right) > 0
$$
and
$$
\frac{d\ }{dq} \delta_{\text{approx}}^{a*}(q) =  \sqrt{\frac{\gamma\sigma^2}{2H''_{\xi}(0)}} \tilde{\delta}^{*'}_\xi\left(-\frac{2q-\Delta}{2} \sqrt{\frac{\gamma\sigma^2}{2H''_{\xi}(0)}}\right)< 0.
$$
This means that a market maker proposes lower prices at the bid and at the ask when his inventory increases, and conversely, higher prices at the bid and at the ask when his inventory decreases. In particular, a market maker with a positive or negative inventory always skews his bid and ask prices in order to increase his chance to go back to a flat position.\\

In the particular case of exponential intensities, it is interesting to notice that the approximation of the bid-ask spread is independent of $q$, and the skew is linear in $q$:
\begin{equation}
\label{sec4:spread}\delta_{\text{approx}}^{b*}(q) + \delta_{\text{approx}}^{a*}(q) =  \begin{cases} \frac{2}{\xi \Delta} \log\left(1 + \frac{\xi\Delta}{k}\right) + \Delta \sqrt{\frac{\gamma\sigma^2}{2A\Delta k} \left(1 + \frac{\xi\Delta}{k}\right)^{\frac{k}{\xi\Delta} + 1} } &\mbox{if } \xi > 0\\
\frac{2}k + \Delta \sqrt{\frac{\gamma\sigma^2e}{2A\Delta k} } & \mbox{if } \xi=0, \end{cases}
\end{equation}
\begin{equation}
\label{sec4:skew}\delta_{\text{approx}}^{b*}(q) - \delta_{\text{approx}}^{a*}(q) =  \begin{cases} 2q \sqrt{\frac{\gamma\sigma^2}{2A\Delta k} \left(1 + \frac{\xi\Delta}{k}\right)^{\frac{k}{\xi\Delta} + 1} } &\mbox{if } \xi > 0\\
2q\sqrt{\frac{\gamma\sigma^2e}{2A\Delta k} } & \mbox{if } \xi=0. \end{cases}
\end{equation}

As far as volatility is concerned, we have
$$
\frac{d\ }{d\sigma}\delta_{\text{approx}}^{b*}(q) = \frac{2q+\Delta}{2} \sqrt{\frac{\gamma}{2H''_{\xi}(0)}}  \tilde{\delta}^{*'}_\xi\left(\frac{2q+\Delta}{2} \sqrt{\frac{\gamma\sigma^2}{2H''_{\xi}(0)}}\right)
$$
and
$$
\frac{d\ }{d\sigma}\delta_{\text{approx}}^{a*}(q) = -\frac{2q-\Delta}{2} \sqrt{\frac{\gamma}{2H''_{\xi}(0)}}  \tilde{\delta}^{*'}_\xi\left(-\frac{2q-\Delta}{2} \sqrt{\frac{\gamma\sigma^2}{2H''_{\xi}(0)}}\right).
$$
Therefore we have three cases:
\begin{itemize}
  \item if $q=0$, then $\frac{d\ }{d\sigma}\delta_{\text{approx}}^{b*}(q)  = \frac{d\ }{d\sigma}\delta_{\text{approx}}^{a*}(q) > 0$. In other words, an increase in volatility leads to an increase in the bid-ask spread, symmetric around the reference price (no skew).
  \item if $q \ge \Delta$, then $\frac{d\ }{d\sigma}\delta_{\text{approx}}^{b*}(q) >0$ and $\frac{d\ }{d\sigma}\delta_{\text{approx}}^{a*}(q) < 0$. In other words, an increase in volatility leads to lower bid and ask prices: it increases the skew in absolute value, \emph{ceteris paribus}.
  \item if $q \le -\Delta$, then $\frac{d\ }{d\sigma}\delta_{\text{approx}}^{b*}(q) <0$ and $\frac{d\ }{d\sigma}\delta_{\text{approx}}^{a*}(q) > 0$. In other words, an increase in volatility leads to higher bid and ask prices: it increases the skew in absolute value, \emph{ceteris paribus}.
\end{itemize}

In the particular case of exponential intensities, it is interesting to notice that the bid-ask spread is approximated by an affine function of $\sigma$, and the skew by a linear function of $\sigma$ (see Eqs.~(\ref{sec4:spread}) and  (\ref{sec4:skew})).\\

As far as liquidity is concerned, if we replace $\Lambda$ by $\beta\Lambda$, for $\beta >0$, then we see that $H_\xi$ is replaced by $\beta H_\xi$, and that $\tilde{\delta}_\xi$ is unchanged (see Eq.~(\ref{sec4:deltatilde})). Therefore, we see from Eqs.~(\ref{sec4:gueantformulab}) and (\ref{sec4:gueantformulaa}), that replacing $\Lambda$ by $\beta\Lambda$ is equivalent to replacing $\sigma^2$ by $\frac{\sigma^2}\beta$. In other words, an increase in liquidity is equivalent to a decrease in volatility and, conversely, a decrease in liquidity has the same effects as an increase in volatility.\\

As far as risk aversion is concerned, the differences between Model A and Model B help to clarify the different roles played by $\gamma$.\\

In the case of Model B, where $\xi=0$, we see from Eqs.~(\ref{sec4:gueantformulab}) and (\ref{sec4:gueantformulaa}), that an increase in~$\gamma$ is equivalent to an increase in $\sigma^2$. In particular, an increase in $\gamma$ increases the bid-ask spread and increases the skew in absolute value. This is expected, since $\gamma$, in Model B, penalizes positive and negative inventory.\\

In the case of Model A, the situation is different, but the introduction of the variable $\xi$ helps to understand what is at stake. As already mentioned, everything works as if $\xi$ was a risk aversion parameter for non-execution risk and $\gamma$ a risk aversion parameter for price risk. To analyze the different effects, we consider the specific case of exponential intensities. We see in Eq.~(\ref{sec4:spread}) that the approximation of the bid-ask spread is made of two parts:
\begin{enumerate}
  \item $\frac{2}{\xi \Delta} \log\left(1 + \frac{\xi\Delta}{k}\right)$, which is decreasing in $\xi$. This term is related to the static risk faced by a market maker, associated with transaction uncertainty only. When $\xi=\gamma$ increases, a market maker reduces his bid-ask spread to lower the uncertainty with respect to transactions.
  \item $\Delta \sqrt{\frac{\gamma\sigma^2}{2A\Delta k} \left(1 + \frac{\xi\Delta}{k}\right)^{\frac{k}{\xi\Delta} + 1} }$, which is increasing in $\gamma$ and $\xi = \gamma$. This term, that only appears with volatility, is related to the dynamic risk faced by a market maker. This risk is complex and definitely more subtle than the classical risk that the price moves. In fact, both $\xi$ and $\gamma$ appear in the formula because the risk faced by a market maker is actually the risk that the price moves adversely without him being able to unwind his position rapidly enough (because of trade uncertainty). The higher the risk aversion to this combination of price risk and non-execution risk, the larger the bid-ask spread, because a market maker wants to avoid holding large inventories (in absolute value).
\end{enumerate}
As far as the skew is concerned, only the second effect matters. This is confirmed by Eq.~(\ref{sec4:skew}), and we see that the skew in absolute value is increasing with $\gamma$ and $\xi=\gamma$.\\

Comparative statics is always interesting to understand the role played by the different parameters involved in a model. Here, we have carried out comparative statics on almost-closed-form and closed-form approximations, and not on the original optimal bid and ask quotes, which can only be computed numerically. We will see in Section 6 the differences between the actual optimal bid and ask quotes and the approximations proposed in this section.\footnote{In particular, in the case of exponential intensities, the actual bid-ask spread is not independent of $q$.}

\section{Multi-asset market making strategies}

In most papers of the academic literature on market making, only single-asset market making is tackled. In practice, however, market makers are often in charge of a book of several assets. An evident case is the one of corporate bonds, since there are usually dozens of bonds issued by the same company, and the same market maker is in charge of all these bonds. As a consequence, optimal quotes for a specific bond should not depend on the market maker's inventory in that bond, but instead on the risk profile of the whole bond portfolio with respect to the issuer. In particular, when a market maker has a short inventory in an asset and an almost equivalent long inventory in another asset, highly correlated with the first, there may be no reason for him to skew his bid and ask quotes on these two assets, contrary to what single-asset market making models would suggest. In this section, we generalize our market making model to the multi-asset case. In particular, we obtain closed-form approximations for the optimal quotes of a multi-asset market maker.

\subsection{Modeling framework and notations}

We consider a market maker in charge of $d$ assets. For $i \in \{ 1, \ldots, d\}$, the reference price of asset $i$ is modeled by a process $(S^i_t)_t$ with the following dynamics
\begin{equation}
\label{sec5:dS}
dS^i_t = \sigma^i dW^i_t,\quad S^i_0 \text{\; given},
\end{equation}
where $(W^1_t, \ldots, W^d_t)_t$ is a $d$-dimensional Brownian motion adapted to the filtration $(\mathcal{F}_t)_{t\in \mathbb{R}_+}$, with nonsingular correlation matrix. We denote by $\Sigma = (\rho^{i,j} \sigma^i\sigma^j)_{1 \le i,j \le d}$, the variance-covariance matrix associated with the process $(S_t)_t = (S^1_t, \ldots, S^d_t)_t$.\\

This market maker proposes bid and ask quotes to buy and sell the $d$ assets. These bid and ask quotes are modeled by $2d$ stochastic processes, respectively denoted by $(S^{1,b}_t)_t, \ldots, (S^{d,b}_t)_t$  and $(S^{1,a}_t)_t, \ldots, (S^{d,a}_t)_t$.\\

As in the single-asset case, we denote by $(N^{i,b}_t)_t$ and $(N^{i,a}_t)_t$, for each $i \in \{1, \ldots, d \}$, the two point processes modeling the number of transactions at the bid and at the ask, respectively, for asset $i$. We assume that the asset $i$ is traded $\Delta^i$ units by $\Delta^i$ units.\\

The inventory of the market maker, modeled by the $d$-dimensional process $(q_t)_t = (q^1_t, \ldots, q^d_t)_t$, has therefore the following dynamics:
\begin{equation}
\label{sec5:dq}
\forall i \in \{1, \ldots, d\}, dq^i_t = \Delta^i dN^{i,b}_t - \Delta^i dN^{i,a}_t, \quad q^i_0 \text{\; given.}
\end{equation}

We assume that the processes $(N^{1,b}_t, \ldots, N^{d,b}_t)_{t}$ and $(N^{1,a}_t, \ldots, N^{d,a}_t)_t$ are independent of the Brownian motion $(W^1_t, \ldots, W^d_t)_t$. For $i \in \{1, \ldots, d\}$, we denote by $(\lambda_t^{i,b})_t$ and $(\lambda_t^{i,a})_t$ the intensity processes of $(N^{i,b}_t)_t$ and $(N^{i,a}_t)_t$, respectively. We assume that $(\lambda_t^{i,b})_t$ and $(\lambda_t^{i,a})_t$ verify
\begin{equation}
  \label{sec5:intensity} \lambda_t^{i,b} = \Lambda^{i,b}(\delta_t^{i,b})1_{q^i_{t-}<Q^i} \quad \text{and} \quad \lambda_t^{i,a} = \Lambda^{i,a}(\delta_t^{i,a})1_{q^i_{t-}>-Q^i},
   \end{equation} where $$\delta_t^{i,b} = S^i_t - S_t^{i,b} \quad \text{and} \quad \delta_t^{i,a} = S^{i,a}_t - S^i_t,$$
and where $\Lambda^{i,b}$ and $\Lambda^{i,a}$ are two functions satisfying the following hypotheses:
\begin{itemize}
  \item $\Lambda^{i,b}$ and $\Lambda^{i,a}$ are twice continuously differentiable,
  \item $\Lambda^{i,b}$ and $\Lambda^{i,a}$ are decreasing, with $\forall \delta \in \R$, ${\Lambda^{i,b}}'(\delta) <0$ and ${\Lambda^{i,a}}'(\delta) <0$,
  \item $\lim_{\delta \to +\infty} \Lambda^{i,b}(\delta) = \lim_{\delta \to +\infty} \Lambda^{i,a}(\delta) = 0$,
    \item $\sup_{\delta}\frac{\Lambda^{i,b}(\delta){\Lambda^{i,b}}''(\delta)}{\left({\Lambda^{i,b}}'(\delta)\right)^2} <  2 \quad \text{and} \quad \sup_{\delta} \frac{\Lambda^{i,a}(\delta){\Lambda^{i,a}}''(\delta)}{\left({\Lambda^{i,a}}'(\delta)\right)^2} <  2.$
\end{itemize}

Finally, the process $(X_t)_t$ modeling the market maker's cash account has the dynamics
\begin{eqnarray}
\nonumber dX_t &=& \sum_{i=1}^d S^{i,a}_t \Delta^i dN^{i,a}_t - S^{i,b}_t \Delta^i dN^{i,b}_t\\
&=&  \sum_{i=1}^d (S^i_t + \delta^{i,a}_t) \Delta^i dN^{i,a}_t - (S^i_t -
\delta^{i,b}_t) \Delta^i dN^{i,b}_t. \label{sec5:dX}
\end{eqnarray}

In the $d$-dimensional generalization of Model A, the problem consists in maximizing
\begin{equation*}
 \mathbb{E}\left[-\exp\left(-\gamma \left(X_T + \sum_{i=1}^d q^i_T S^i_T - \ell_d(q^1_T, \ldots, q^d_T)\right) \right)\right], \qquad \textbf{(Model A)}
\end{equation*}
over $(\delta_t^{1,b}, \ldots, \delta_t^{d,b})_t \in \mathcal{A}^d$ and $(\delta_t^{1,a}, \ldots, \delta_t^{d,a})_t \in \mathcal{A}^d$, where $\ell_d$ is a penalty function.\\

In the $d$-dimensional generalization of Model B, the problem consists instead in maximizing
\begin{equation*}
\mathbb{E}\left[X_T + \sum_{i=1}^d q^i_T S^i_T - \ell_d(q^1_T, \ldots, q^d_T) - \frac 12 \gamma \int_0^T \sum_{i=1}^d \sum_{j=1}^d \rho^{i,j} \sigma^i\sigma^j q^i_t q^j_t  dt \right], \qquad \textbf{(Model B)}
\end{equation*}
over $(\delta_t^{1,b}, \ldots, \delta_t^{d,b})_t \in \mathcal{A}^d$ and $(\delta_t^{1,a}, \ldots, \delta_t^{d,b})_t \in \mathcal{A}^d$.

\subsection{Towards a general system of ordinary differential equations}

For solving the two stochastic optimal control problems of Model A and Model B, we use similar changes of variables as in Section 3. In particular, we show that finding the value function (and the optimal bid and ask quotes) in both models boils down to solving a system of ordinary differential equations, and that, as in the single-asset case, the equations associated with Model~A and Model B are part of the same family of ODEs.\\

The HJB equation associated with Model A is given by\footnote{We denote by $(e^1, \ldots, e^d)$ the canonical basis of $\mathbb{R}^d$.}
\begin{equation}
\label{sec5:HJBModelA}0= -\partial_t u(t,x,q,S) - \frac 12 \sum_{i=1}^d \sum_{j=1}^d \rho^{i,j} \sigma^i\sigma^j \partial^2_{S^iS^j} u(t,x,q,S)
\end{equation}
$$ - \sum_{i=1}^d 1_{q^i<Q^i}\sup_{\delta^{i,b}}
\Lambda^{i,b}(\delta^{i,b}) \left[u(t,x-\Delta^i S^i+\Delta^i\delta^{i,b},q+\Delta^i e^i,S) - u(t,x,q,S) \right]$$
$$ - \sum_{i=1}^d 1_{q^i>-Q^i} \sup_{\delta^{i,a}} \Lambda^{i,a}(\delta^{i,a}) \left[u(t,x+\Delta^i S^i+\Delta^i \delta^{i,a},q-\Delta^ie^i,S) -
u(t,x,q,S) \right],$$ for $\forall i \in \{ 1, \ldots, d\}, q^i \in \mathcal{Q}^i = \lbrace -Q^i, -Q^i + \Delta^i, \ldots, Q^i-\Delta^i, Q^i \rbrace$, and $(t,S,x) \in [0,T]\times \mathbb{R}^d\times \mathbb{R}$, with the terminal condition \begin{equation}\label{sec5:HJBModelACT}u(T,x,q,S) = -\exp\left(-\gamma\left(x+\sum_{i=1}^d q^iS^i - \ell_d(q^1,\ldots, q^d)\right)\right).\end{equation}

If one uses the ansatz
\begin{equation}\label{sec5:ansatzModelA}u(t,x,q,S) = -\exp\left(-\gamma\left(x+\sum_{i=1}^d q^iS^i + \theta(t,q)\right)\right),\end{equation}
then Eq.~(\ref{sec5:HJBModelA}) becomes
\begin{equation}
\label{sec5:thetaModelA}
0=-\partial_t \theta(t,q) + \frac{1}{2} \gamma \sum_{i=1}^d\sum_{j=1}^d \rho^{i,j} \sigma^i\sigma^j q^i q^j
\end{equation}
$$-  \sum_{i=1}^d 1_{q^i<Q^i}\sup_{\delta^{i,b}} \frac{\Lambda^{i,b}(\delta^{i,b})}{\gamma}\left(1-\exp\left(-\gamma\left(\Delta^i \delta^{i,b} + \theta(t,q+\Delta^ie^i) - \theta(t,q) \right)\right)\right)$$
$$- \sum_{i=1}^d  1_{q^i>-Q^i} \sup_{\delta^{i,a}} \frac{\Lambda^{i,a}(\delta^{i,a})}{\gamma}\left(1-\exp\left(-\gamma \left(\Delta^i \delta^{i,a} + \theta(t,q-\Delta^ie^i) - \theta(t,q) \right)\right)\right),
$$ for $\forall i \in \{ 1, \ldots, d\}, q^i \in \mathcal{Q}^i$, and $t \in [0,T]$, and the terminal condition (\ref{sec5:HJBModelACT}) becomes $\theta(T,q) = -\ell_d(q^1,\ldots,q^d)$.\\

The HJB equation associated with Model B is given by
\begin{equation}
\label{sec5:HJBModelB}0= -\partial_t u(t,x,q,S) + \frac{1}{2} \gamma \sum_{i=1}^d\sum_{j=1}^d \rho^{i,j} \sigma^i\sigma^j q^i q^j
\end{equation}
$$- \frac 12 \sum_{i=1}^d \sum_{j=1}^d \rho^{i,j} \sigma^i\sigma^j \partial^2_{S^iS^j} u(t,x,q,S)$$
$$ - \sum_{i=1}^d 1_{q^i<Q^i}\sup_{\delta^{i,b}}
\Lambda^{i,b}(\delta^{i,b}) \left[u(t,x-\Delta^i S^i+\Delta^i\delta^{i,b},q+\Delta^i e^i,S) - u(t,x,q,S) \right]$$
$$ - \sum_{i=1}^d 1_{q^i>-Q^i} \sup_{\delta^{i,a}} \Lambda^{i,a}(\delta^{i,a}) \left[u(t,x+\Delta^i S^i+\Delta^i \delta^{i,a},q-\Delta^ie^i,S) -
u(t,x,q,S) \right],$$ for $\forall i \in \{ 1, \ldots, d\}, q^i \in \mathcal{Q}^i$ and $(t,S,x) \in [0,T]\times \mathbb{R}^d\times \mathbb{R}$, with the terminal condition \begin{equation}\label{sec5:HJBModelBCT}u(T,x,q,S) = x+\sum_{i=1}^d q^iS^i - \ell_d(q^1,\ldots,q^d).\end{equation}

If one uses the ansatz
\begin{equation}\label{sec5:ansatzModelB}u(t,x,q,S) = x+\sum_{i=1}^d q^iS^i + \theta(t,q),\end{equation}
then Eq.~(\ref{sec5:HJBModelB}) becomes
\begin{equation}
\label{sec5:thetaModelB}
0=-\partial_t \theta(t,q) + \frac{1}{2} \gamma \sum_{i=1}^d\sum_{j=1}^d \rho^{i,j} \sigma^i\sigma^j q^i q^j
\end{equation}
$$-  \sum_{i=1}^d 1_{q^i<Q^i} \sup_{\delta^{i,b}} \Lambda^{i,b}(\delta^{i,b})\left(\Delta^i \delta^{i,b} + \theta(t,q+\Delta^ie^i) - \theta(t,q) \right)$$
$$- \sum_{i=1}^d 1_{q^i>-Q^i} \sup_{\delta^{i,a}} \Lambda^{i,a}(\delta^{i,a})\left(\Delta^i \delta^{i,a} + \theta(t,q-\Delta^ie^i) - \theta(t,q) \right),
$$ for $\forall i \in \{ 1, \ldots, d\}, q^i \in \mathcal{Q}^i$, and $t \in [0,T]$, and the terminal condition (\ref{sec5:HJBModelBCT}) becomes $\theta(T,q) = -\ell_d(q^1,\ldots,q^d)$.\\

As in the single-asset case, Eqs.~(\ref{sec5:thetaModelA}) and (\ref{sec5:thetaModelB}) are in fact two systems of ODEs which belong to the same family. Let us introduce for $\xi >0$ the functions
$$H^{i,b}_{\xi}(p) = \sup_{\delta} \frac{\Lambda^{i,b}(\delta)}{\xi}\left(1-\exp\left(-\xi\Delta^i \left(\delta - p \right)\right)\right)$$
and
$$H^{i,a}_{\xi}(p) = \sup_{\delta} \frac{\Lambda^{i,a}(\delta)}{\xi}\left(1-\exp\left(-\xi\Delta^i \left(\delta - p \right)\right)\right),$$
and the limit functions (for $\xi=0$) $$H^{i,b}_{0}(p) = \Delta^i \sup_{\delta} \Lambda^{i,b}(\delta)(\delta - p),$$ and $$H^{i,a}_{0}(p) = \Delta^i \sup_{\delta} \Lambda^{i,a}(\delta)(\delta - p).$$  Then, we can consider the general equation
\begin{equation}
\label{sec5:thetagen}
0=-\partial_t \theta(t,q) + \frac{1}{2} \gamma \sum_{i=1}^d\sum_{j=1}^d \rho^{i,j} \sigma^i\sigma^j q^i q^j
\end{equation}
$$-  \sum_{i=1}^d 1_{q^i<Q^i} H^{i,b}_{\xi}\left(\frac{\theta(t,q) - \theta(t,q+\Delta^ie^i)}{\Delta^i}\right) - \sum_{i=1}^d 1_{q^i>-Q^i} H^{i,a}_{\xi}\left(\frac{\theta(t,q) - \theta(t,q-\Delta^ie^i)}{\Delta^i}\right),$$
for $\forall i \in \{ 1, \ldots, d\}, q^i \in \mathcal{Q}^i$, and $t \in [0,T]$, with the terminal condition
\begin{equation}
\label{sec5:thetagenCT}
\theta(T,q) = -\ell_d(q^1,\ldots,q^d).
\end{equation}

Eq.~(\ref{sec5:thetaModelA}) corresponds to Eq.~(\ref{sec5:thetagen}) for $\xi=\gamma$ while Eq.~(\ref{sec5:thetaModelB}) corresponds to Eq.~(\ref{sec5:thetagen}) for $\xi=0$.\\

\subsection{Solution of the market making problem}

In order to characterize the optimal quotes in our multi-asset market making model, we proceed as in the single-asset case. In particular, we start by proving that there exists a solution of Eq.~(\ref{sec5:thetagen}) with terminal condition (\ref{sec5:thetagenCT}).\\

\begin{Theorem}
\label{sec5:theotheta}
  There exists a unique function $\theta : [0,T]\times \prod_{i=1}^{d}\mathcal{Q}^i \rightarrow \mathbb{R}$, $C^1$ in time, solution of Eq.~(\ref{sec5:thetagen}) with terminal condition (\ref{sec5:thetagenCT}).
\end{Theorem}

\begin{proof}

Eq.~(\ref{sec5:thetagen}) with terminal condition (\ref{sec5:thetagenCT}) is a backward Cauchy problem. Because the functions $H^{i,b}_\xi$ and $H^{i,a}_\xi$ are functions of class $C^1$ for all $i \in \{1,\ldots,d\}$, the Cauchy-Lipschitz theorem applies, and there exists $\tau \in [0,T)$ and a function $\theta : (\tau,T]\times \prod_{i=1}^{d}\mathcal{Q}^i  \rightarrow \mathbb{R}$, $C^1$ in time, solution of Eq.~(\ref{sec5:thetagen}) on $(\tau,T]$ with terminal condition (\ref{sec5:thetagenCT}).\\

$\forall q \in \prod_{i=1}^{d}\mathcal{Q}^i$, the function $t \in (\tau,T] \mapsto \theta(t,q)  + \frac 12 \gamma \sum_{i=1}^d\sum_{j=1}^d \rho^{i,j} \sigma^i\sigma^j q^i q^j (T-t)$ is a decreasing function. Therefore, the only reason why there would not be a global solution on $[0,T]$ is because $\sup_{q \in \prod_{i=1}^{d}\mathcal{Q}^i} \theta(t,q)$ blows up at $\tau>0$. However, by using a comparison principle similar to that of Lemma~\ref{sec3:cp}, we easily see that $$\sup_{q \in \prod_{i=1}^{d}\mathcal{Q}^i} \theta(t,q) \le \sum_{i=1}^d(H^{i,b}_\xi(0) + H^{i,a}_\xi(0))(T-t).$$ Therefore, $\sup_{q \in \prod_{i=1}^{d}\mathcal{Q}^i} \theta(t,q)$ cannot blow up in finite time, and $\theta$ is in fact defined on $[0,T]\times \prod_{i=1}^{d}\mathcal{Q}^i$.\\

Uniqueness comes then for the Cauchy-Lipschitz theorem.\qed\\
\end{proof}

We are now ready to state the two theorems characterizing the optimal quotes in Model~A and Model B. The proofs of these results are based on verification arguments, and are (\emph{mutatis mutandis}) identical to those in the single-asset case.\\

Let us start with Model A.\\

\begin{Theorem}
\label{sec5:theomodelA}
Let us consider the solution $\theta$ of Eq.~(\ref{sec5:thetagen}) with terminal condition (\ref{sec5:thetagenCT}) for $\xi=\gamma$.\\

Then, $u : (t,x,q,S) \mapsto -\exp(-\gamma(x+\sum_{i=1}^d q^iS^i + \theta(t,q)))$  defines a solution to Eq.~(\ref{sec5:HJBModelA}) with terminal condition (\ref{sec5:HJBModelACT}), and
$$u(t,x,q,S) = \sup_{(\delta^{1,b}_s,\ldots, \delta^{d,b}_s)_{s\ge t}, (\delta^{1,a}_s,\ldots, \delta^{d,a}_s )_{s\ge t}  \in \mathcal{A}(t)^d} \mathbb{E}\Bigg[- \exp\Bigg(-\gamma\Bigg(X^{t,x,\delta^b,\delta^a}_T+\sum_{i=1}^d q^{i,t,q^i,\delta^b,\delta^a}_T S^{i,t,S^i}_T$$$$- \ell_d(q^{1,t,q^1,\delta^b,\delta^a}_T, \ldots, q^{d,t,q^d,\delta^b,\delta^a}_T)\Bigg)\Bigg) \Bigg],$$
where
$$\forall i \in \{1, \ldots, d\}, \quad dS^{i,t,S^i}_s = \sigma^i dW^i_s, \qquad S^{i,t,S^i}_t = S^i,$$
$$dX^{t,x,\delta^b,\delta^a}_s = \sum_{i=1}^d (S^i_s + \delta^{i,a}_s)  \Delta^i dN^{i,a}_s - (S^i_s - \delta^{i,b}_s)  \Delta^i dN^{i,b}_s , \qquad  X^{t,x,\delta^b,\delta^a}_t = x,$$
$$\forall i \in \{1, \ldots, d\}, \quad dq^{i,t,q^i,\delta^b,\delta^a}_s = \Delta^i dN^{i,b}_s - \Delta^i dN^{i,a}_s, \qquad  q^{i,t,q^i,\delta^b,\delta^a}_t = q^i,$$
and where $\forall i \in \{1, \ldots, d\}$, the point processes $N^{i,b}$ and $N^{i,a}$ have stochastic intensity $(\lambda^{i,b}_s)_s$ and $(\lambda^{i,a}_s)_s$ given by $\lambda^{i,b}_s = \Lambda^{i,b}(\delta^{i,b}_s) 1_{q^i_{s-} < Q^i}$ and $\lambda^{i,a}_s = \Lambda^{i,a}(\delta^{i,a}_s) 1_{q^i_{s-} > -Q^i}$.\\

For $i \in \{1, \ldots, d\},$ the optimal bid and ask quotes $S^{i,b}_t = S^i_t - \delta^{i,b*}_t$ (for $q^i_{t-}< Q^i$) and $S^{i,a}_t = S^i_t + \delta^{i,a*}_t$ (for $q^i_{t-}>-Q^i$) are characterized by\begin{equation}
\label{sec5:deltaoptimalA}
\delta^{i,b*}_t = \tilde{\delta}^{i,b*}_\gamma\left(\frac{\theta(t,q_{t-}) - \theta(t,q_{t-}+\Delta^i e^i)}{\Delta^i}\right) \text{ and } \delta^{i,a*}_t = \tilde{\delta}^{i,a*}_\gamma\left(\frac{\theta(t,q_{t-}) - \theta(t,q_{t-}-\Delta^i e^i)}{\Delta^i}\right),
\end{equation}where the functions $\tilde{\delta}^{i,b*}_\gamma(\cdot)$ and $\tilde{\delta}_\gamma^{i,a*}(\cdot)$ are defined by
$$\tilde{\delta}^{i,b*}_\gamma(p) = {\Lambda^{i,b}}^{-1}\left(\gamma H^{i,b}_{\gamma}(p) - \frac{{H_{\gamma}^{i,b}}'(p)}{\Delta^i}\right) \text{ and } \tilde{\delta}^{i,a*}_\gamma(p) = {\Lambda^{i,a}}^{-1}\left(\gamma H^{i,a}_{\gamma}(p) - \frac{{H_{\gamma}^{i,a}}'(p)}{\Delta^i}\right).$$

\end{Theorem}

For model B, the result is the following:

\begin{Theorem}
\label{sec5:theomodelB}
Let us consider the solution $\theta$ of Eq.~(\ref{sec5:thetagen}) with terminal condition (\ref{sec5:thetagenCT}) for $\xi=0$.\\

Then, $u : (t,x,q,S) \mapsto x+\sum_{i=1}^d q^iS^i + \theta(t,q)$  defines a solution to Eq.~(\ref{sec5:HJBModelB}) with terminal condition (\ref{sec5:HJBModelBCT}), and
$$u(t,x,q,S) = \sup_{(\delta^{1,b}_s,\ldots, \delta^{d,b}_s)_{s\ge t}, (\delta^{1,a}_s,\ldots, \delta^{d,a}_s )_{s\ge t}  \in \mathcal{A}(t)^d} \mathbb{E}\Bigg[X^{t,x,\delta^b,\delta^a}_T+\sum_{i=1}^d q^{i,t,q^i,\delta^b,\delta^a}_T S^{i,t,S^i}_T$$$$- \ell_d(q^{1,t,q^1,\delta^b,\delta^a}_T, \ldots, q^{d,t,q^d,\delta^b,\delta^a}_T) - \frac 12 \gamma \int_0^T \sum_{i=1}^d \sum_{j=1}^d \rho^{i,j} \sigma^i\sigma^j q^{i,t,q^i,\delta^b,\delta^a}_t q^{j,t,q^j,\delta^b,\delta^a}_t  dt\Bigg],$$
where:
$$\forall i \in \{1, \ldots, d\}, \quad dS^{i,t,S^i}_s = \sigma^i dW^i_s, \qquad S^{i,t,S^i}_t = S^i,$$
$$dX^{t,x,\delta^b,\delta^a}_s = \sum_{i=1}^d (S^i_s + \delta^{i,a}_s)  \Delta^i dN^{i,a}_s - (S^i_s - \delta^{i,b}_s)  \Delta^i dN^{i,b}_s , \qquad  X^{t,x,\delta^b,\delta^a}_t = x,$$
$$\forall i \in \{1, \ldots, d\}, \quad dq^{i,t,q^i,\delta^b,\delta^a}_s = \Delta^i dN^{i,b}_s - \Delta^i dN^{i,a}_s, \qquad  q^{i,t,q^i,\delta^b,\delta^a}_t = q^i,$$
and where $\forall i \in \{1, \ldots, d\}$, the point processes $N^{i,b}$ and $N^{i,a}$ have stochastic intensity $(\lambda^{i,b}_s)_s$ and $(\lambda^{i,a}_s)_s$ given by $\lambda^{i,b}_s = \Lambda^{i,b}(\delta^{i,b}_s) 1_{q^i_{s-} < Q^i}$ and $\lambda^{i,a}_s = \Lambda^{i,a}(\delta^{i,a}_s) 1_{q^i_{s-} > -Q^i}$.\\

For $i \in \{1, \ldots, d\},$ the optimal bid and ask quotes $S^{i,b}_t = S^i_t - \delta^{i,b*}_t$ (for $q^i_{t-}< Q^i$) and $S^{i,a}_t = S^i_t + \delta^{i,a*}_t$ (for $q^i_{t-}>-Q^i$) are characterized by
\vspace{3mm}\begin{equation}
\label{sec5:deltaoptimalB}
\delta^{i,b*}_t = \tilde{\delta}^{i,b*}_0\left(\frac{\theta(t,q_{t-}) - \theta(t,q_{t-}+\Delta^i e^i)}{\Delta^i}\right) \text{ and } \delta^{i,a*}_t = \tilde{\delta}^{i,a*}_0\left(\frac{\theta(t,q_{t-}) - \theta(t,q_{t-}-\Delta^i e^i)}{\Delta^i}\right),
\end{equation}where the functions $\tilde{\delta}^{i,b*}_0(\cdot)$ and $\tilde{\delta}_0^{i,a*}(\cdot)$ are defined by
$$\tilde{\delta}^{i,b*}_0(p) = {\Lambda^{i,b}}^{-1}\left(- \frac{{H_{0}^{i,b}}'(p)}{\Delta^i}\right) \text{ and } \tilde{\delta}^{i,a*}_0(p) = {\Lambda^{i,a}}^{-1}\left(- \frac{{H_{0}^{i,a}}'(p)}{\Delta^i}\right).$$
\end{Theorem}

\subsection{About closed-form approximations}

In the single-asset case, closed-form approximations were obtained in Section 4, in the special case where $\Lambda^b = \Lambda^a =: \Lambda$ and $H''_\xi(0) > 0$.\footnote{Very recently, closed-form approximations have also been found in the case of asymmetric intensities -- see~\cite{egv2017}.} In the multi-asset case, if we assume that $\forall i \in \{1, \ldots, d\}, \Lambda^{i,b} = \Lambda^{i,a} =: \Lambda^i$, and $H^{i''}_\xi(0) > 0$, then it is natural to wonder whether the same techniques can be used in order to obtain closed-form approximations.\\

The answer is in fact that the change of variables used to derive closed-form approximations does not work in general in dimension higher than 1. However, the idea of transforming Eq.~\eqref{sec5:thetagen} into a multidimensional equivalent of Eq.~\eqref{sec4:thetatilde} enables to obtain results, without using the Hopf-Cole transform -- i.e., without relying on a  multidimensional equivalent of Eq.~\eqref{sec4:v}.\\

Following the same reasoning as in Section 4, we can indeed introduce the PDE
$$0 = - \partial_t \tilde\theta(t,q) + \frac 12 \gamma \sum_{i=1}^d\sum_{j=1}^d \rho^{i,j} \sigma^i\sigma^j q^i q^j - 2\sum_{i=1}^d H^i_{\xi}\left(0\right)$$\begin{equation}\label{sec5:thetatilde}- \sum_{i=1}^d H_{\xi}^{i''}(0) (\partial_{q^i} \tilde\theta(t,q))^2 + \Delta^i H_{\xi}^{i'}(0) \partial^2_{q^iq^i} \tilde\theta(t,q),\end{equation}
with final condition $\tilde\theta(T,q_1,\ldots,q_d) = -\ell_d(q_1,\ldots,q_d)$.\\

In the case where $\ell_d(q_1,\ldots,q_d) = \sum_{i=1}^d\sum_{j=1}^d a^{i,j} q^i q^j$ with $(a^{i,j})_{i,j}$ a symmetric positive matrix, it is easy to see that Eq.~\eqref{sec5:thetatilde} can be solved in closed-form by using the ansatz
$$\tilde{\theta}(t,q) = \theta_0(t) - q'\theta_2(t)q,$$
where $\theta_2(t)$ is a $d\times d$ symmetric matrix (see the companion paper \cite{egv2017}).\\

In particular, we show in \cite{egv2017} that $\theta_2(t)$ verifies:
$$\theta_2(t) \to_{T \to +\infty} \frac 12 \sqrt{\frac \gamma 2} \Gamma,$$ where$$\Gamma = D^{-\frac 12}\left(D^{\frac 12} \Sigma D^{\frac 12}\right)^{\frac 12}D^{-\frac 12}, \qquad D =
\left(\begin{matrix}
H_{\xi}^{1''}(0) &  0  & \ldots & 0\\
0  &  H_{\xi}^{2''}(0) & \ldots & 0\\
\vdots & \vdots & \ddots & \vdots\\
0  &   0       &\ldots & H_{\xi}^{d''}(0)
\end{matrix}\right).
$$

As a consequence, we can consider the approximations
$$\frac{\theta(t,q) - \theta(t,q+\Delta^i e^i)}{\Delta^i} \simeq \sqrt{\frac \gamma 2 }\left( \Gamma^{ii} \frac{2q^i+\Delta^i}{2} + \sum_{1\le j \le d, j\not= i} \Gamma^{ij} q^j \right)$$
and
$$\frac{\theta(t,q) - \theta(t,q-\Delta^i e^i)}{\Delta^i} \simeq -\sqrt{\frac \gamma 2 } \left(\Gamma^{ii} \frac{2q^i-\Delta^i}{2} + \sum_{1\le j \le d, j\not= i} \Gamma^{ij} q^j \right).$$

These approximations can be plugged into Eqs.~(\ref{sec5:deltaoptimalA}) and (\ref{sec5:deltaoptimalB}) to obtain the general approximation formulas
\begin{equation}
\label{sec5:gueantformulab}\delta^{i,b*}_t \simeq \delta_{\text{approx}}^{i,b*}(q_{t-}) := \tilde{\delta}^{i*}_\xi\left(\sqrt{\frac \gamma 2 } \left( \Gamma^{ii} \frac{2q_{t-}^i+\Delta^i}{2} + \sum_{1\le j \le d, j\not= i} \Gamma^{ij} q_{t-}^j \right)\right)
\end{equation}
and
\begin{equation}
\label{sec5:gueantformulaa}
\delta^{i,a*}_t \simeq \delta_{\text{approx}}^{i,a*}(q_{t-}) := \tilde{\delta}^{i*}_\xi\left(-\sqrt{\frac \gamma 2 } \left( \Gamma^{ii} \frac{2q_{t-}^i-\Delta^i}{2} + \sum_{1\le j \le d, j\not= i} \Gamma^{ij} q_{t-}^j \right)\right),
\end{equation} where
\begin{equation}
\label{sec5:deltatilde}
\tilde{\delta}^{i*}_\xi(p) = {\Lambda^i}^{-1}\left(\xi H^i_{\xi}(p) - \frac{{H^i_{\xi}}'(p)}{\Delta^i}\right).
\end{equation}

These approximation formulas are interesting because we see the cross-effects coming from the non-diagonal terms of the matrix $\Gamma$.

\section{Application: the case of two credit indices}

In this section, we apply our single-asset and multi-asset market making models, along with the associated closed-form approximations, to the case of two credit (or CDS) indices: the investment grade (IG) index CDX.NA.IG and the high yield (HY) index CDX.NA.HY. We consider a market maker who is in charge of proposing bid and ask quotes for these two indices, and we will assume throughout this section that this market maker is only concerned with spread risk and not with default risk -- this hypothesis is always made by practitioners for market making fixed-income and credit instruments.\\

Without going into the details of these indices,\footnote{See www.markit.com for more details.} we need to specify their main financial characteristics. Basically, for the IG index, the protection buyer pays quarterly (at fixed dates in order to ease compensation) a coupon corresponding to an annualized rate of 100~bps, and pays upfront an amount (positive or negative) corresponding to an upfront rate (positive or negative) determined by the market. In practice, for market making, the upfront rate is the relevant variable because a round trip on the index leads to a PnL corresponding to the difference between upfront rates (times the notional of the  transaction). However, in practice, this index is quoted in spread -- this spread being computed using a basic CDS model. For the HY index, the protection buyer pays quarterly (at fixed dates) a coupon corresponding to an annualized rate of 500 bps, and pays upfront an amount (positive or negative) corresponding to an upfront rate (positive or negative) determined by the market. Unlike the IG index, the HY index is quoted in upfront rate, or more precisely as $100(1- \text{upfront\; rate})$. It is also noteworthy that, in practice, buying the IG index means buying protection, whereas buying the HY index means selling protection. For simplifying the exposition, we will consider that buying always means buying protection, and that the index quotes are the upfront rates. The conversion of our numerical results into market standard quotes can easily be carried out by using a basic CDS model.\\

In order to apply our models to these credit indices, we need first to estimate the value of the different parameters. This has been done thanks to the data provided by BNP Paribas in the framework of the Research Initiative ``Nouveaux traitements pour les donn\'ees lacunaires issues des activit\'es de cr\'edit'', which is financed by BNP Paribas under the aegis of the Europlace Institute of Finance. For estimating the volatilty and correlation parameters $\sigma^{IG}, \sigma^{HY},$ and $\rho$, mid-prices (prices here are upfront rates) have been considered. For the intensity functions, exponential intensities have been considered and the parameters $A^{IG}$, $k^{IG}$, $A^{HY}$, and $k^{HY}$ have been estimated with classical likelihood maximization techniques using real quotes posted by the bank and the trades occurring between the bank and other market participants.\footnote{The period of estimation was over the first semester of 2016.}\\

If we consider that the two theoretical assets correspond to $\$1$ of each index respectively, the value of the parameters are the following (figures are rounded):
\begin{center}
\begin{tabular}{|c|c|c|}
  \hline
   & IG index & HY index \\
  \hline \hline
  $\sigma\; (\$.s^{-\frac 12})$ & $\sigma^{IG} = 5.83\cdot10^{-6}$  & $\sigma^{HY} = 2.15\cdot10^{-5}$ \\
  \hline
  $\rho$ & \multicolumn{2}{c|}{$\rho = 0.9$}  \\
  \hline
  $A\; (s^{-1})$  & $A^{IG} = 9.10\cdot10^{-4}$ & $A^{HY} = 1.06\cdot10^{-3}$ \\
  \hline
  $k\; (\$^{-1})$ & $k^{IG} = 1.79\cdot10^{4}$ & $k^{HY} = 5.47\cdot10^{3}$ \\
  \hline
\end{tabular}
\end{center}

Coming now to the order sizes, we consider orders of size $\Delta^{IG} = \$50$ million for the IG index, and orders of size $ \Delta^{HY} = \$10$ million for the HY index.\\

As far as risk aversion is concerned, we consider a reference value $\gamma = 6 \cdot10^{-5} \$^{-1}$.\\

Regarding risk limits, we consider that $\frac{Q^{IG}}{\Delta^{IG}} = \frac{Q^{HY}}{\Delta^{HY}} = 4$.\\

Finally, we always consider a final time $T = 7200 \;s$, corresponding to 2 hours. We will see indeed on the examples below that the asymptotic regime is reached very rapidly, in far less than 2 hours.\\

We can consider first the case of the IG index alone. We approximated the solution $\theta$ of the systems of ODEs (\ref{sec3:thetagen}) by using an implicit scheme and a Newton's method at each time step to deal with the nonlinearity. Then we obtained the feedback control function $$(t,q^{IG}) \mapsto (\delta^{IG,b}(t,q^{IG}), \delta^{IG,a}(t,q^{IG}))$$ which gives the optimal bid and ask quotes\footnote{In fact the difference between the reference price and the actual quote, as in the rest of the paper.} at time $t$ when $q^{IG}_{t-} = q^{IG}$.\\

We see in Figure \ref{convergence1} that the asymptotic regime is reached after less than 1 hour.\\

\begin{figure}[H]
  \centering
  \includegraphics[width=0.65\textwidth]{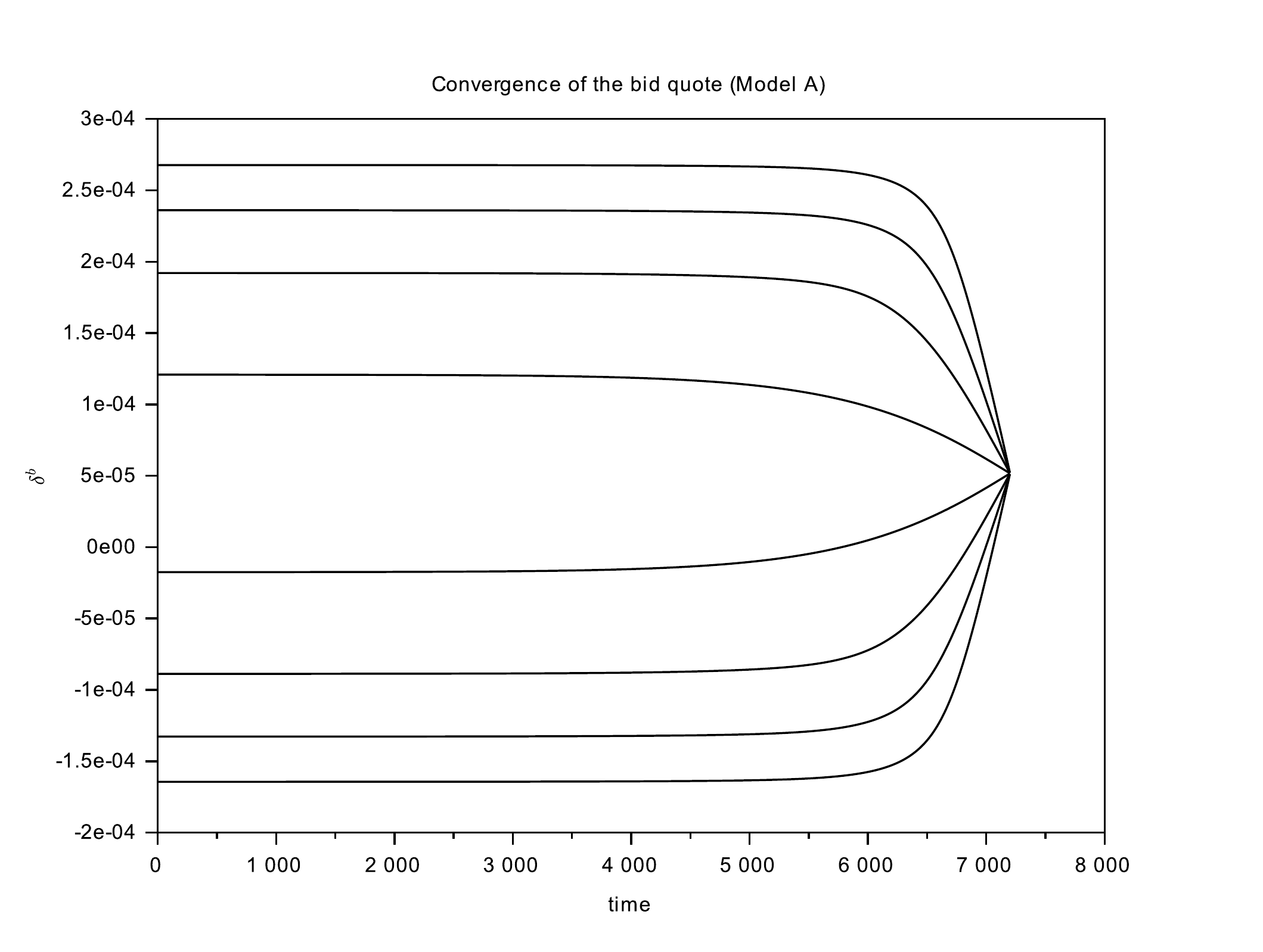}
  \caption{$t \mapsto \delta^{IG,b}(t,q^{IG})$ in Model A for the different values of $q^{IG}$.}\label{convergence1}
\end{figure}

In Figures \ref{IG_bid} and \ref{IG_ask}, we plot the initial (i.e., asymptotic) values of the bid and ask quotes, obtained with Model A, for the IG index, when it is considered on a stand-alone basis. We see that the market maker quotes conservatively at the bid and aggressively at the ask when he is long, and conversely that he quotes conservatively at the ask and aggressively at the bid when he is short.\\

\begin{figure}[H]
  \centering
  \includegraphics[width=0.65\textwidth]{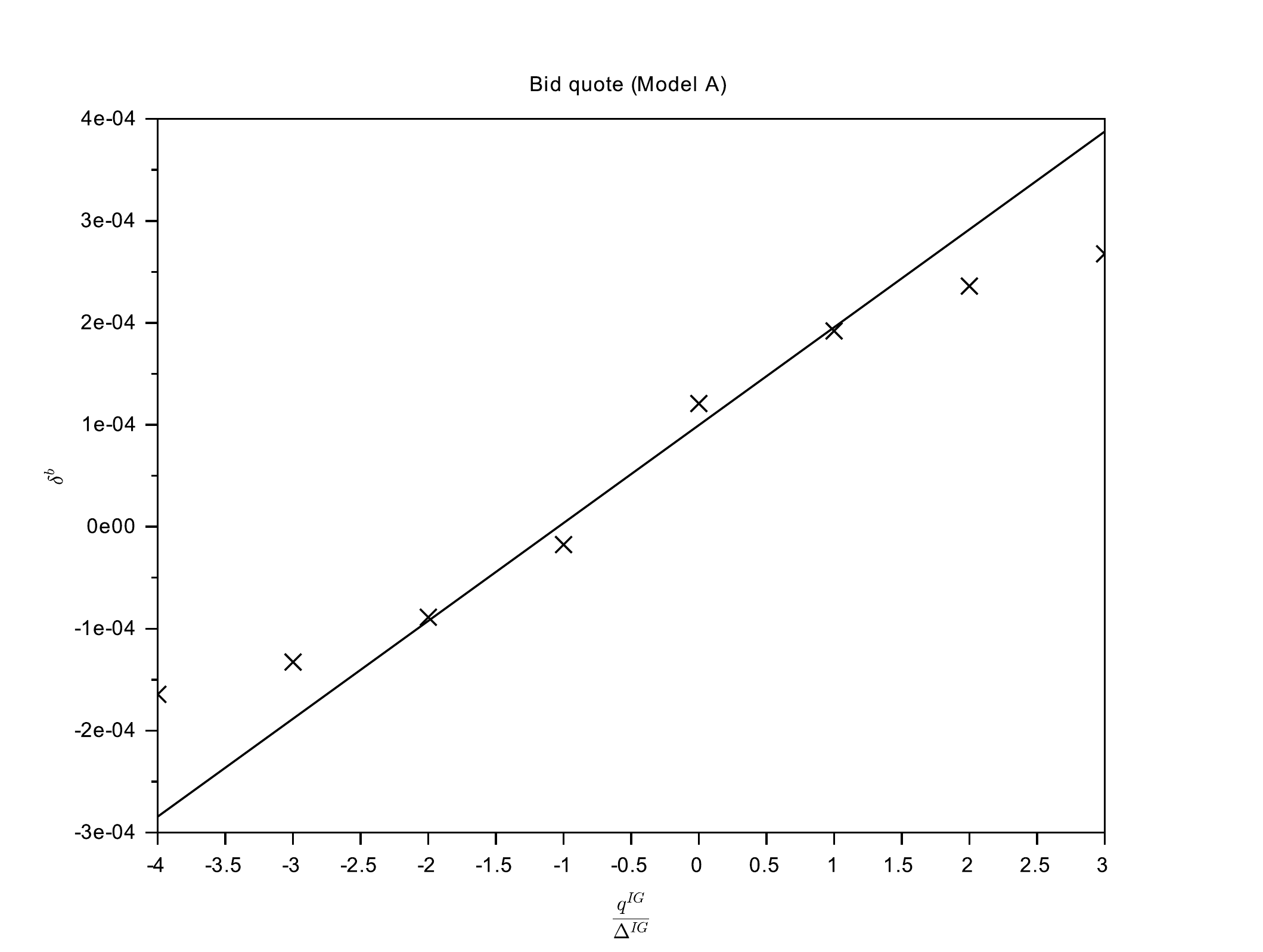}
  \caption{$q^{IG} \mapsto \delta^{IG,b}(0,q^{IG})$ (crosses) and the associated closed-form approximations (line) obtained with Eq. (\ref{sec4:glftformulab})  -- in the case of Model A.}\label{IG_bid}
\end{figure}

\begin{figure}[H]
  \centering
  \includegraphics[width=0.65\textwidth]{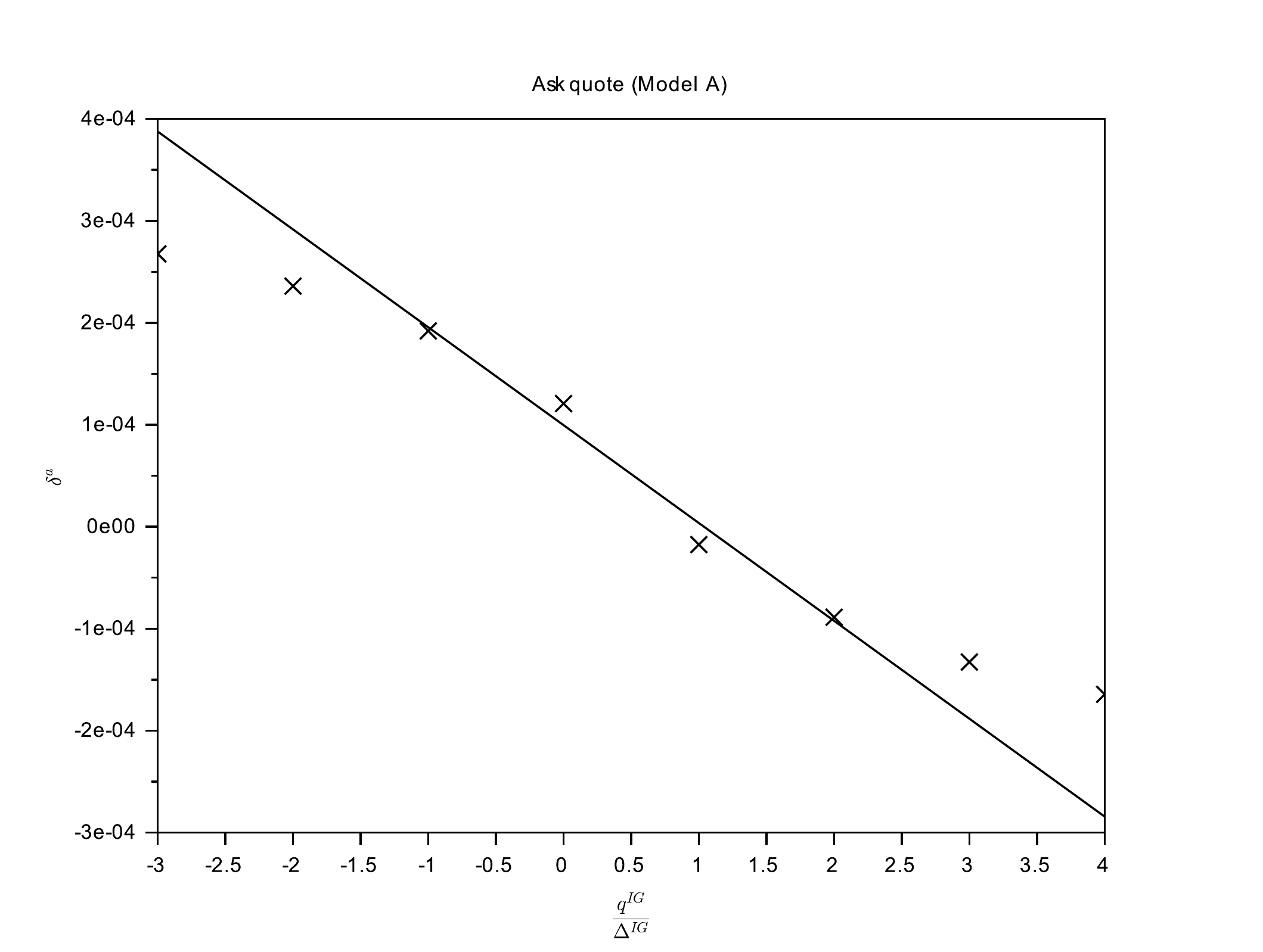}
  \caption{$q^{IG} \mapsto \delta^{IG,a}(0,q^{IG})$ (crosses) and the associated closed-form approximations (line)  obtained with Eq. (\ref{sec4:glftformulaa})  -- in the case of Model A.}\label{IG_ask}
\end{figure}

We also see that the closed-form approximations are satisfactory for small values of the inventory (in absolute value), but more questionable for larger values. In particular, the optimal quotes are not affine functions of the inventory as the closed-form approximations suggest.\\

The difference between actual values, obtained through the numerical approximation of the solution of a system of ODEs, and closed-form approximations can also be seen in Figures \ref{IG_spread} and \ref{IG_skew}, which represent the bid-ask spread and the skew of a market maker quoting optimally. The bid-ask spread is indeed not constant, and the skew is not linear on our example.\\

\begin{figure}[H]
  \centering
  \includegraphics[width=0.65\textwidth]{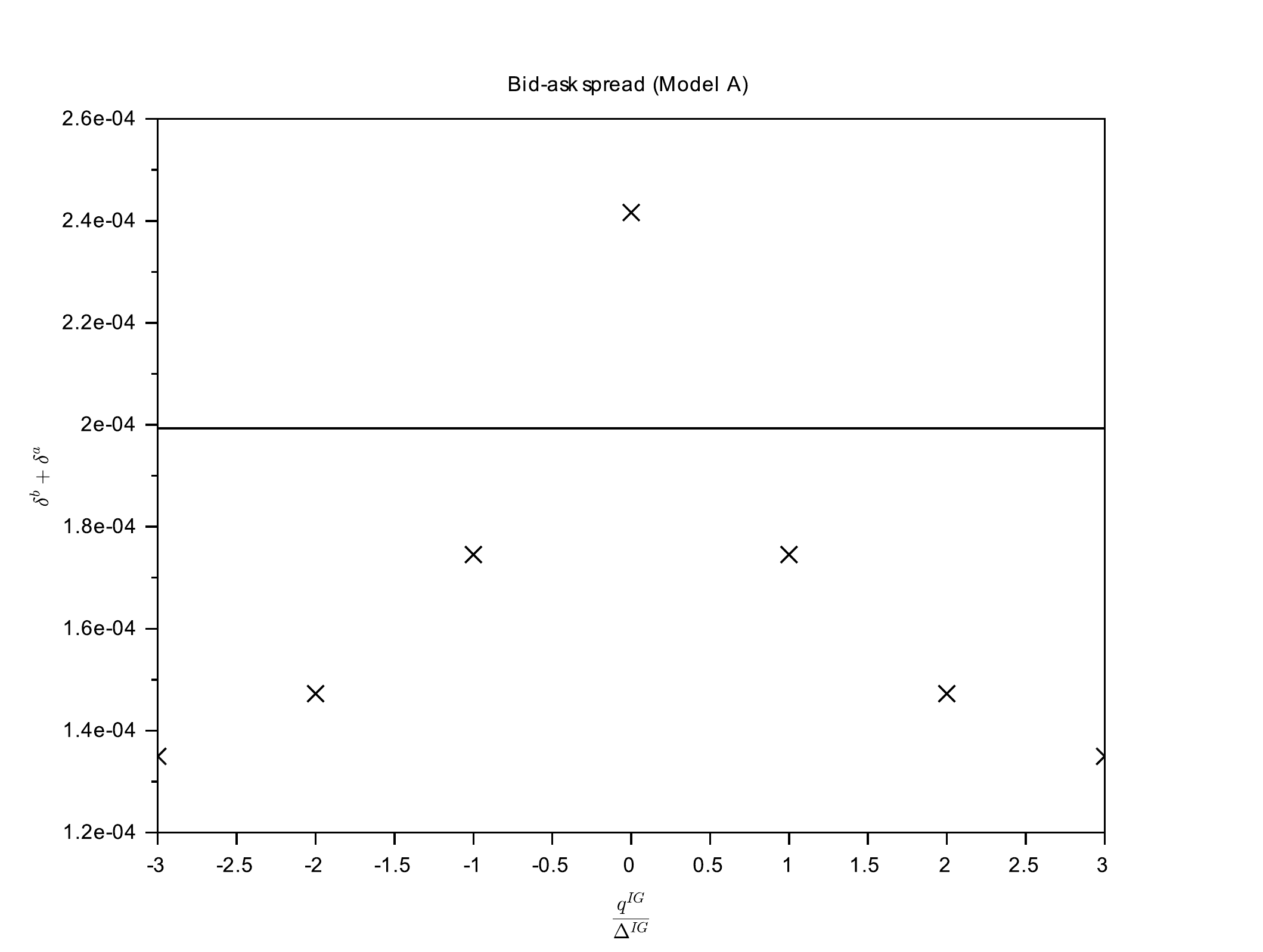}
  \caption{$q^{IG} \mapsto \delta^{IG,b}(0,q^{IG}) + \delta^{IG,a}(0,q^{IG})$ (crosses) and the associated closed-form approximations (line)  obtained with Eq. (\ref{sec4:spread})  -- in the case of Model A.}\label{IG_spread}
\end{figure}

\begin{figure}[H]
  \centering
  \includegraphics[width=0.65\textwidth]{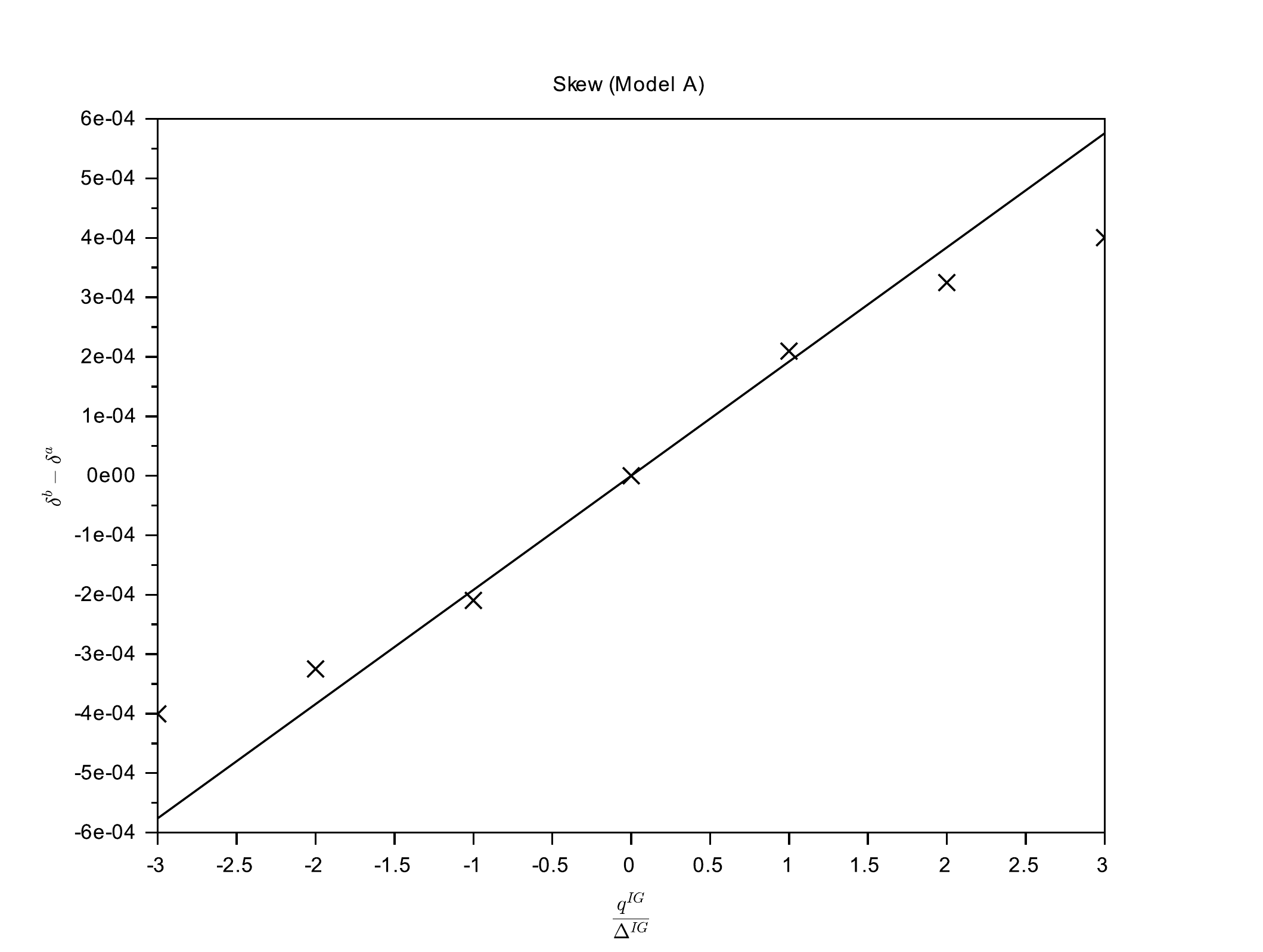}
  \caption{$q^{IG} \mapsto \delta^{IG,b}(0,q^{IG}) - \delta^{IG,a}(0,q^{IG})$ (crosses) and the associated closed-form approximations (line)  obtained with Eq. (\ref{sec4:skew})  -- in the case of Model A.}\label{IG_skew}
\end{figure}

However, if we consider market conditions with less volatility, then the closed-form approximations are far better -- see Figures \ref{low1} and \ref{low2} where we computed the optimal bid and ask quotes (in Model A) for a value of $\sigma^{IG}$ divided by 2. The quality of the approximations depends therefore strongly on the considered market and on the market context. Practitioners must subsequently understand in depth the trade-off between accuracy and computational time (especially when there are hundreds of assets) in order to choose between the two methods. \\

\begin{figure}[H]
  \centering
  \includegraphics[width=0.65\textwidth]{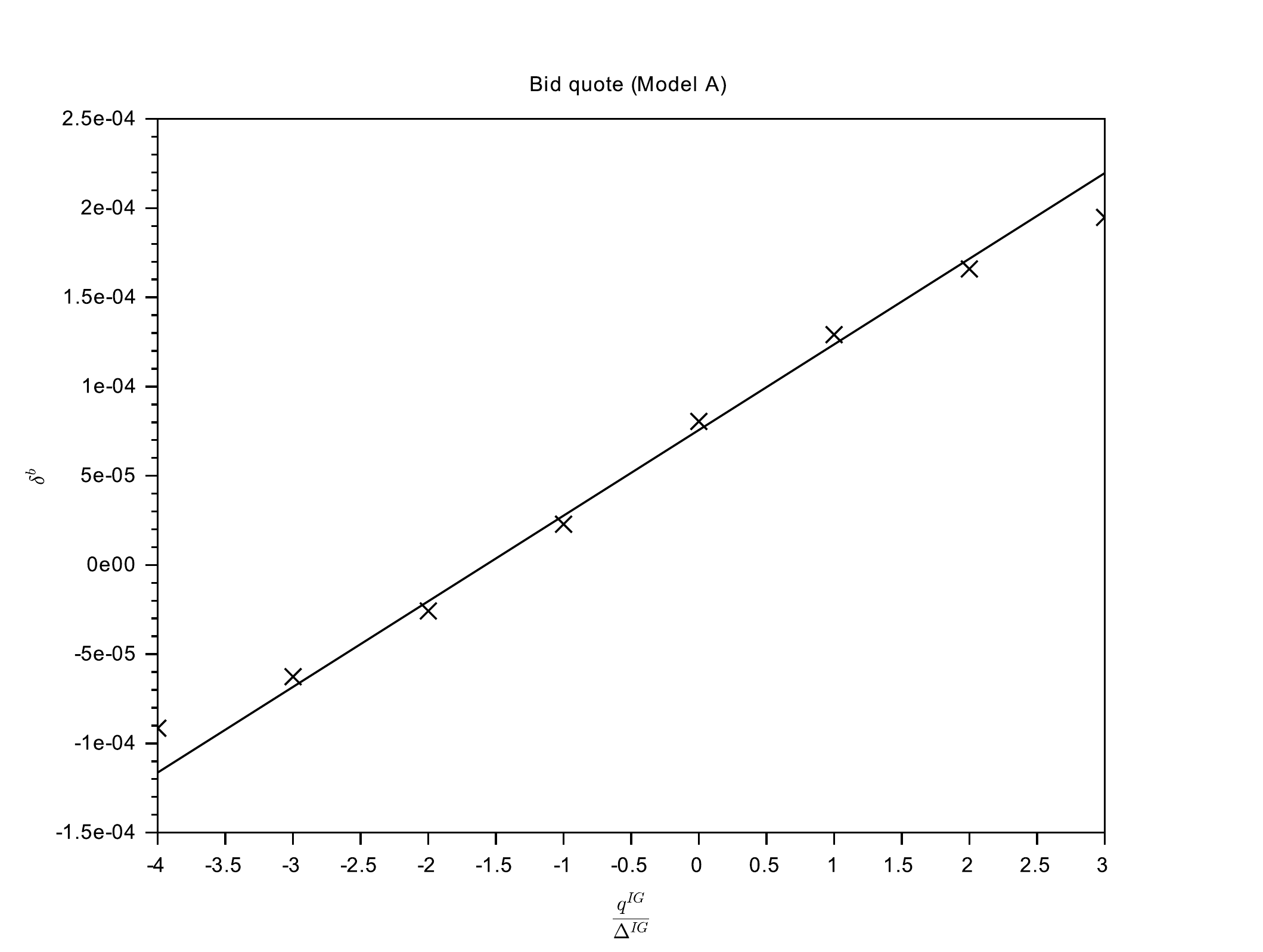}
  \caption{$q^{IG} \mapsto \delta^{IG,b}(0,q^{IG})$ (crosses) and the associated closed-form approximations (line) obtained with Eq. (\ref{sec4:glftformulab})  -- in the case of Model A, when $\sigma^{IG}$ is reduced by half.}\label{low1}
\end{figure}

\begin{figure}[H]
  \centering
  \includegraphics[width=0.65\textwidth]{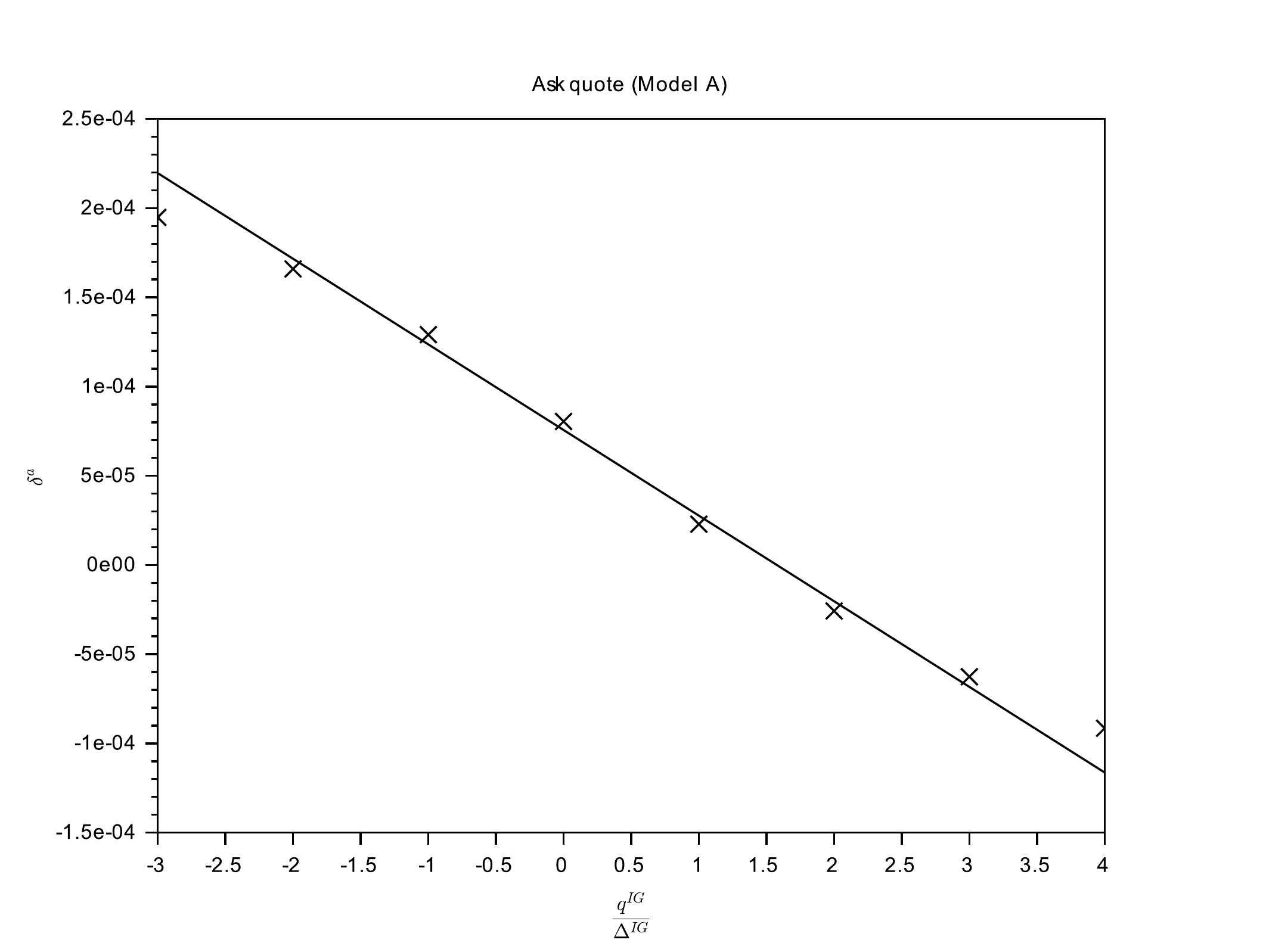}
  \caption{$q^{IG} \mapsto \delta^{IG,a}(0,q^{IG})$ (crosses) and the associated closed-form approximations (line)  obtained with Eq. (\ref{sec4:glftformulaa})  -- in the case of Model A, when $\sigma^{IG}$ is reduced by half.}\label{low2}
\end{figure}

So far in this section, we have only considered optimal quotes in Model A. We see in Figures~\ref{comp1} and \ref{comp2} that the differences between the two models is in fact very small. In other words, although Model B ignores part of the risk (or more precisely aversion to part of the risk), it constitutes a very interesting simplification of Model A.\\

\begin{figure}[H]
  \centering
  \includegraphics[width=0.65\textwidth]{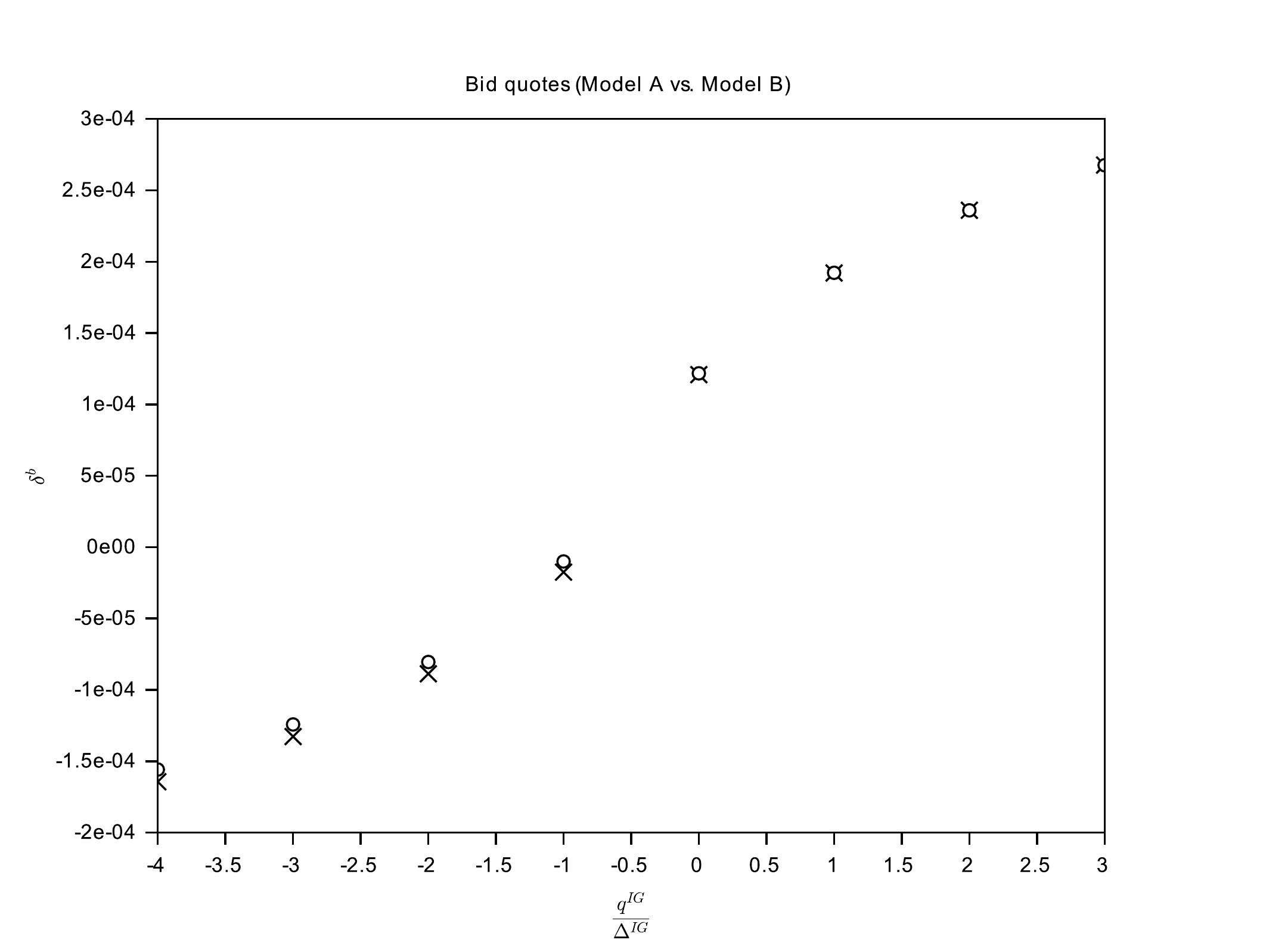}
  \caption{$q^{IG} \mapsto \delta^{IG,b}(0,q^{IG})$ in Model A (crosses) and $q^{IG} \mapsto \delta^{IG,b}(0,q^{IG})$ in Model B (circles).}\label{comp1}
\end{figure}

\begin{figure}[H]
  \centering
  \includegraphics[width=0.65\textwidth]{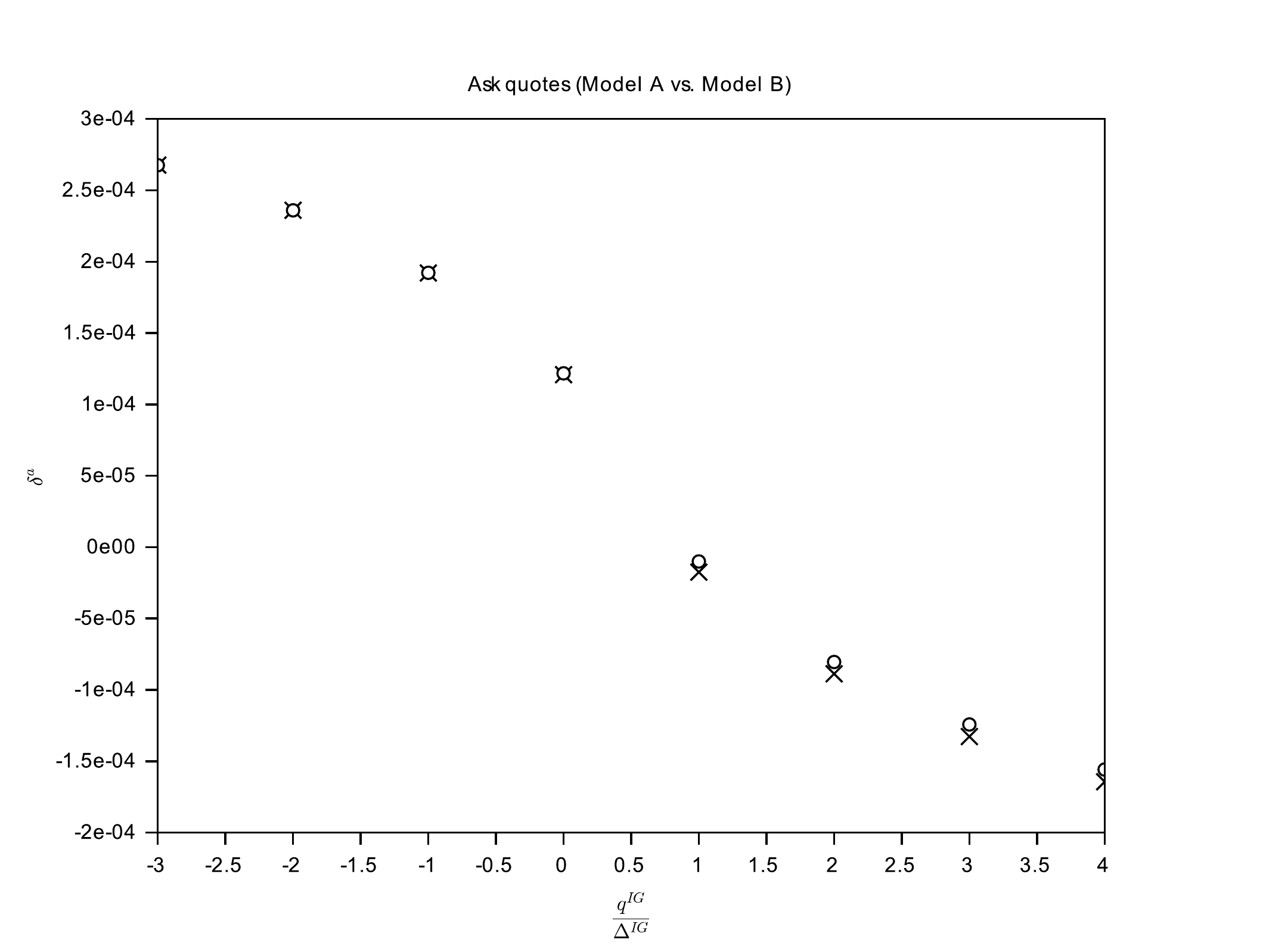}
  \caption{$q^{IG} \mapsto \delta^{IG,a}(0,q^{IG})$ in Model A (crosses) and $q^{IG} \mapsto \delta^{IG,a}(0,q^{IG})$ in Model B (circles).}\label{comp2}
\end{figure}

Let us now come to the case of the HY index alone. Like for the IG index, we approximated the solution $\theta$ of the systems of ODEs (\ref{sec3:thetagen}) by using an implicit scheme and a Newton's method at each time step to deal with the nonlinearity. Then we obtained the feedback control function $$(t,q^{HY}) \mapsto (\delta^{HY,b}(t,q^{HY}), \delta^{HY,a}(t,q^{HY}))$$ which gives the optimal bid and ask quotes at time $t$ when $q^{HY}_{t-} = q^{HY}$.\\

We see in Figure \ref{convergence2} that the asymptotic regime is reached after nearly 1 hour.\\

\begin{figure}[H]
  \centering
  \includegraphics[width=0.65\textwidth]{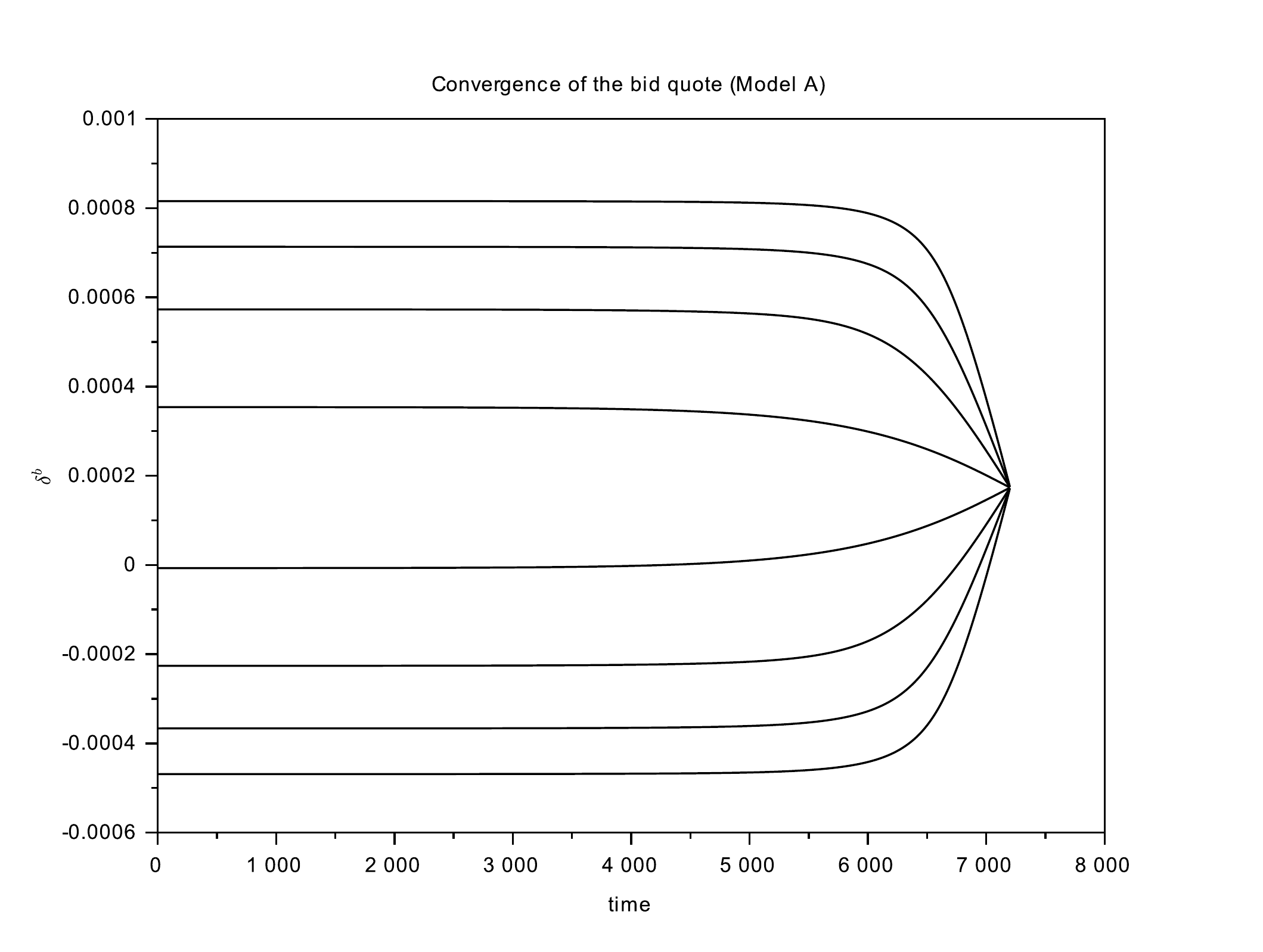}
  \caption{$t \mapsto \delta^{HY,b}(t,q^{HY})$ in Model A for the different values of $q^{HY}$.}\label{convergence2}
\end{figure}

In Figures \ref{HY_bid} and \ref{HY_ask}, we plot the initial (i.e., asymptotic) values of the bid and ask quotes, obtained with Model A, for the HY index, when it is considered on a stand-alone basis. As above, we see that the market maker quotes conservatively at the bid and aggressively at the ask when he is long, and conversely that he quotes conservatively at the ask and aggressively at the bid when he is short. We also see that the closed-form approximations are satisfactory only for small values of the inventory (in absolute value).\\

\begin{figure}[H]
  \centering
  \includegraphics[width=0.65\textwidth]{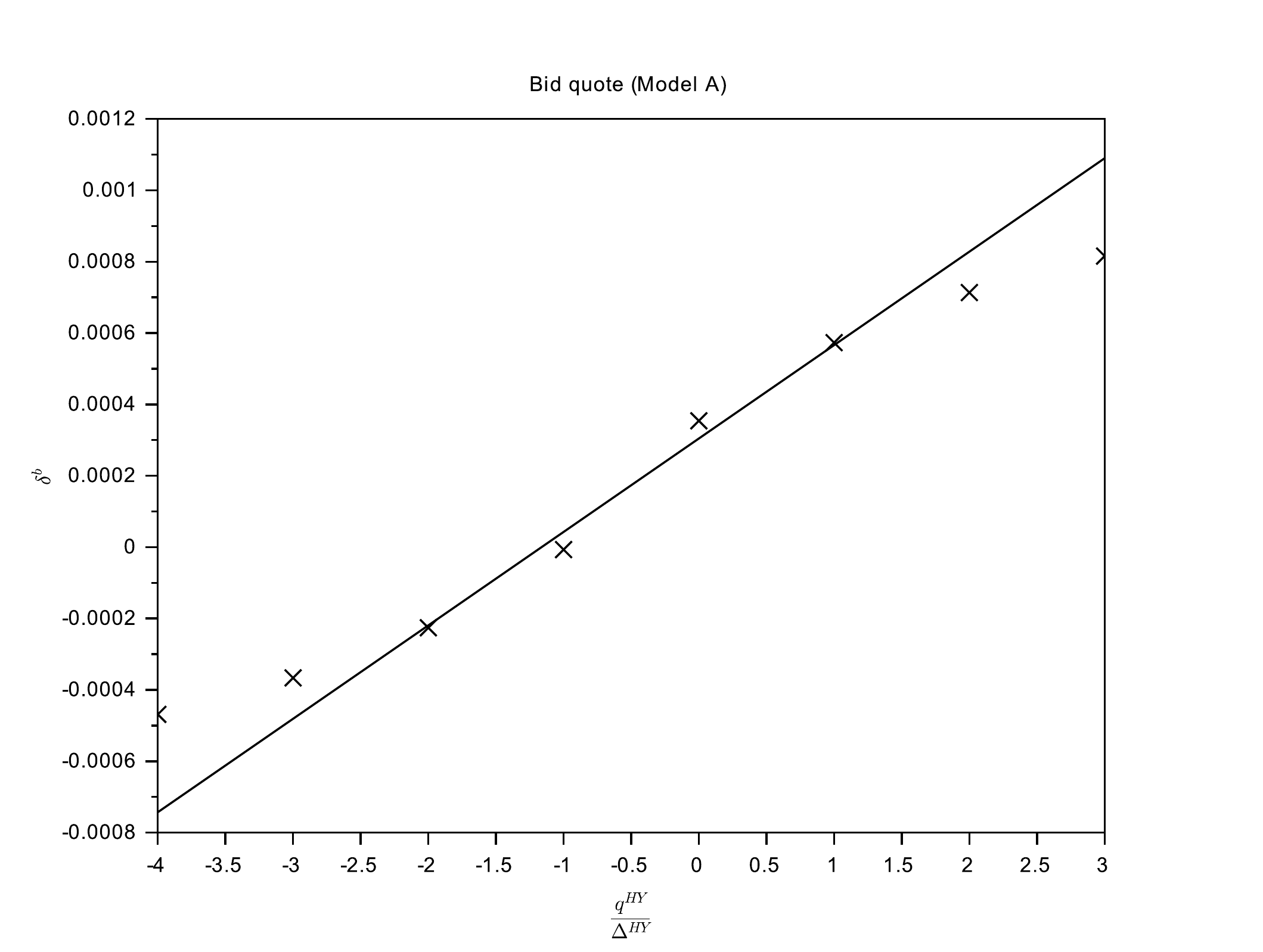}
  \caption{$q^{HY} \mapsto \delta^{HY,b}(0,q^{HY})$ (crosses) and the associated closed-form approximations (line) obtained with Eq. (\ref{sec4:glftformulab})  -- in the case of Model A.}\label{HY_bid}
\end{figure}

\begin{figure}[H]
  \centering
  \includegraphics[width=0.65\textwidth]{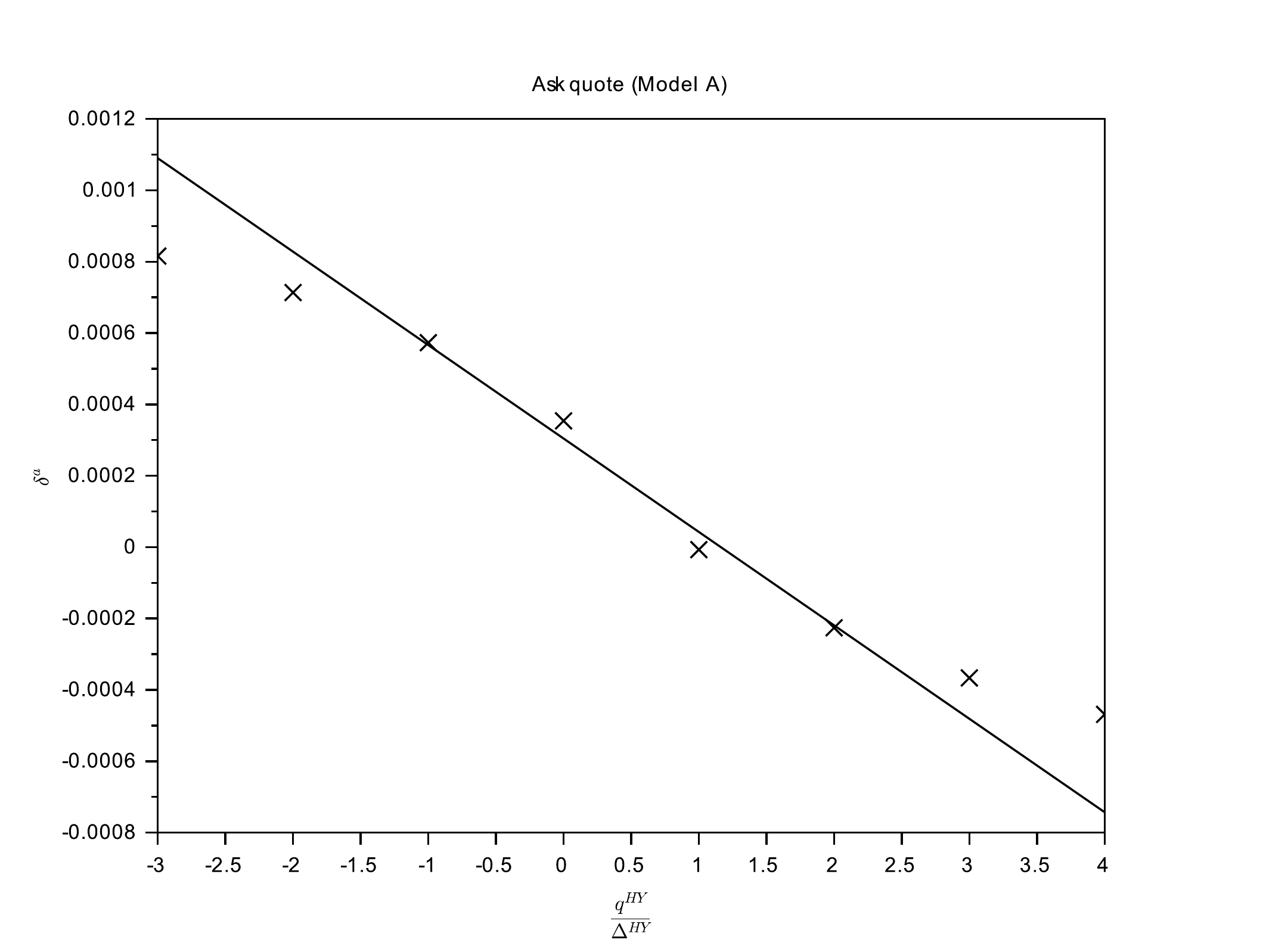}
  \caption{$q^{HY} \mapsto \delta^{HY,a}(0,q^{HY})$ (crosses) and the associated closed-form approximations (line)  obtained with Eq. (\ref{sec4:glftformulaa})  -- in the case of Model A.}\label{HY_ask}
\end{figure}

The difference between actual values and closed-form approximations can also be seen in Figures \ref{HY_spread} and \ref{HY_skew}, which represent the bid-ask spread and the skew of a market maker quoting optimally. The bid-ask spread is indeed not constant, and the skew is not linear on our example.\\

\begin{figure}[H]
  \centering
  \includegraphics[width=0.65\textwidth]{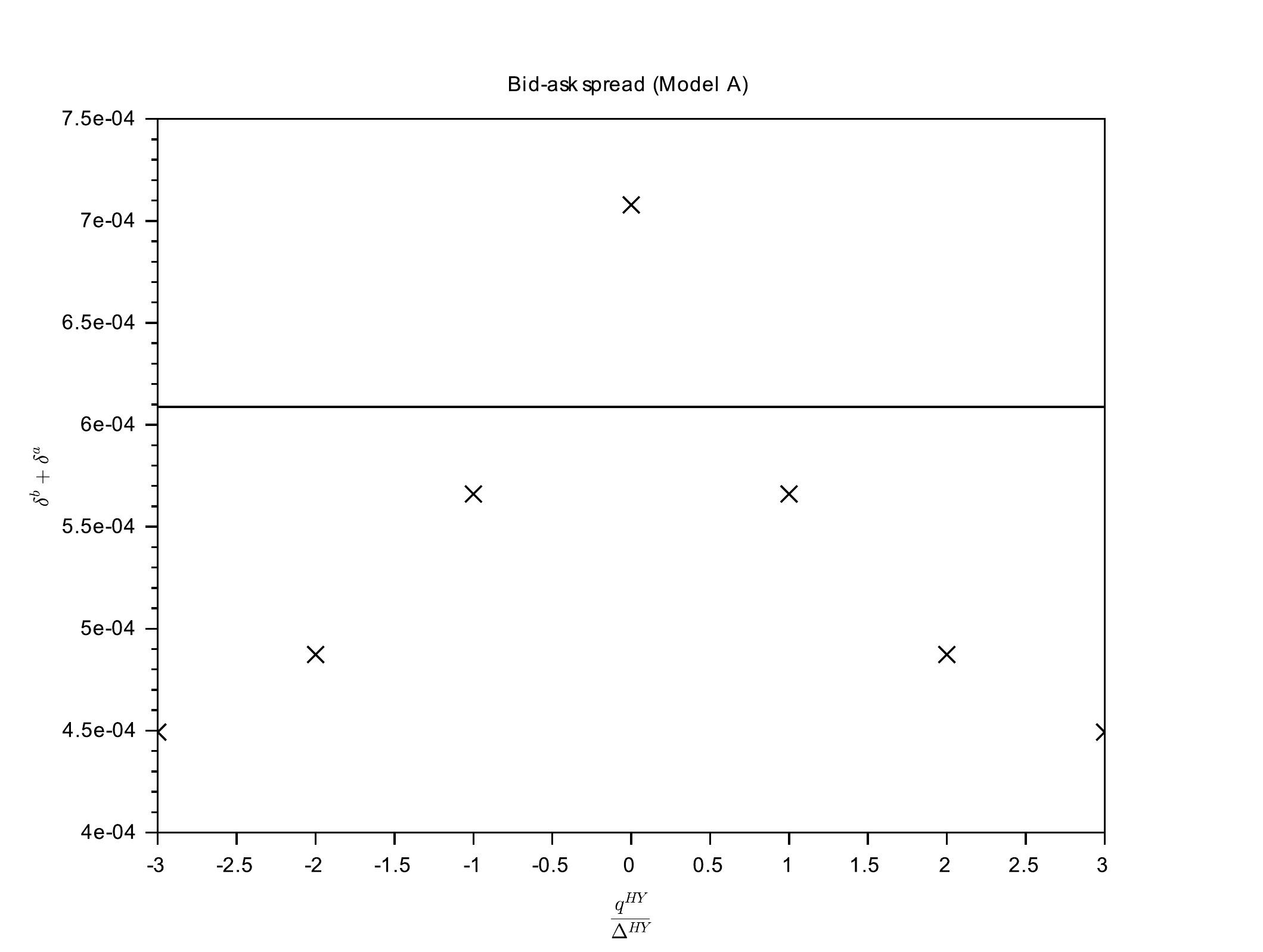}
  \caption{$q^{HY} \mapsto \delta^{HY,b}(0,q^{HY}) + \delta^{HY,a}(0,q^{HY})$ (crosses) and the associated closed-form approximations (line)  obtained with Eq. (\ref{sec4:spread})  -- in the case of Model A.}\label{HY_spread}
\end{figure}

\begin{figure}[H]
  \centering
  \includegraphics[width=0.65\textwidth]{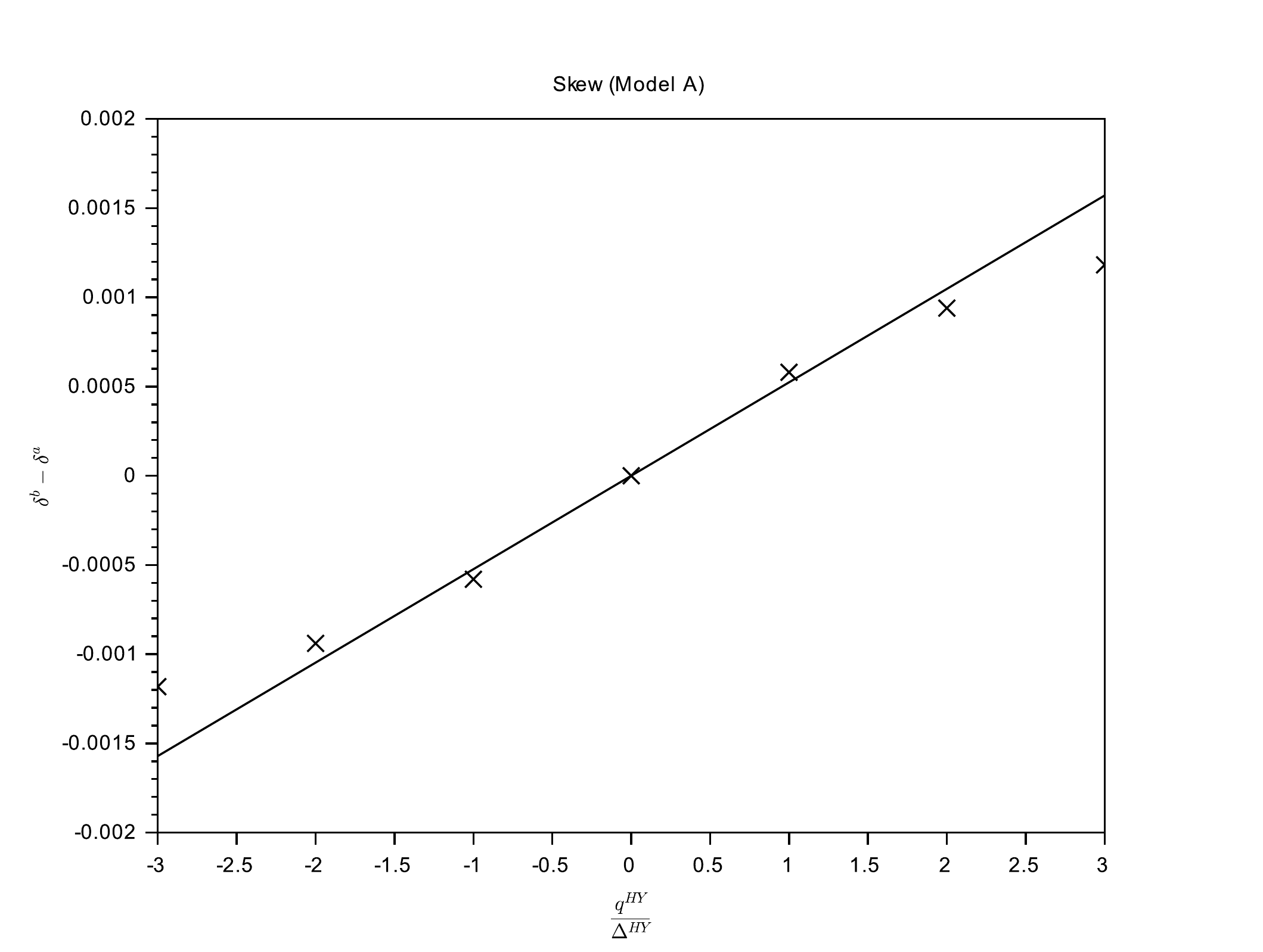}
  \caption{$q^{HY} \mapsto \delta^{HY,b}(0,q^{HY}) - \delta^{HY,a}(0,q^{HY})$ (crosses) and the associated closed-form approximations (line)  obtained with Eq. (\ref{sec4:skew})  -- in the case of Model A.}\label{HY_skew}
\end{figure}

As far as the comparison between Model A and Model B are concerned, we see in Figures \ref{comp3} and \ref{comp4} that the differences between the two models is very small, as in the case of the IG index.\\

\begin{figure}[H]
  \centering
  \includegraphics[width=0.65\textwidth]{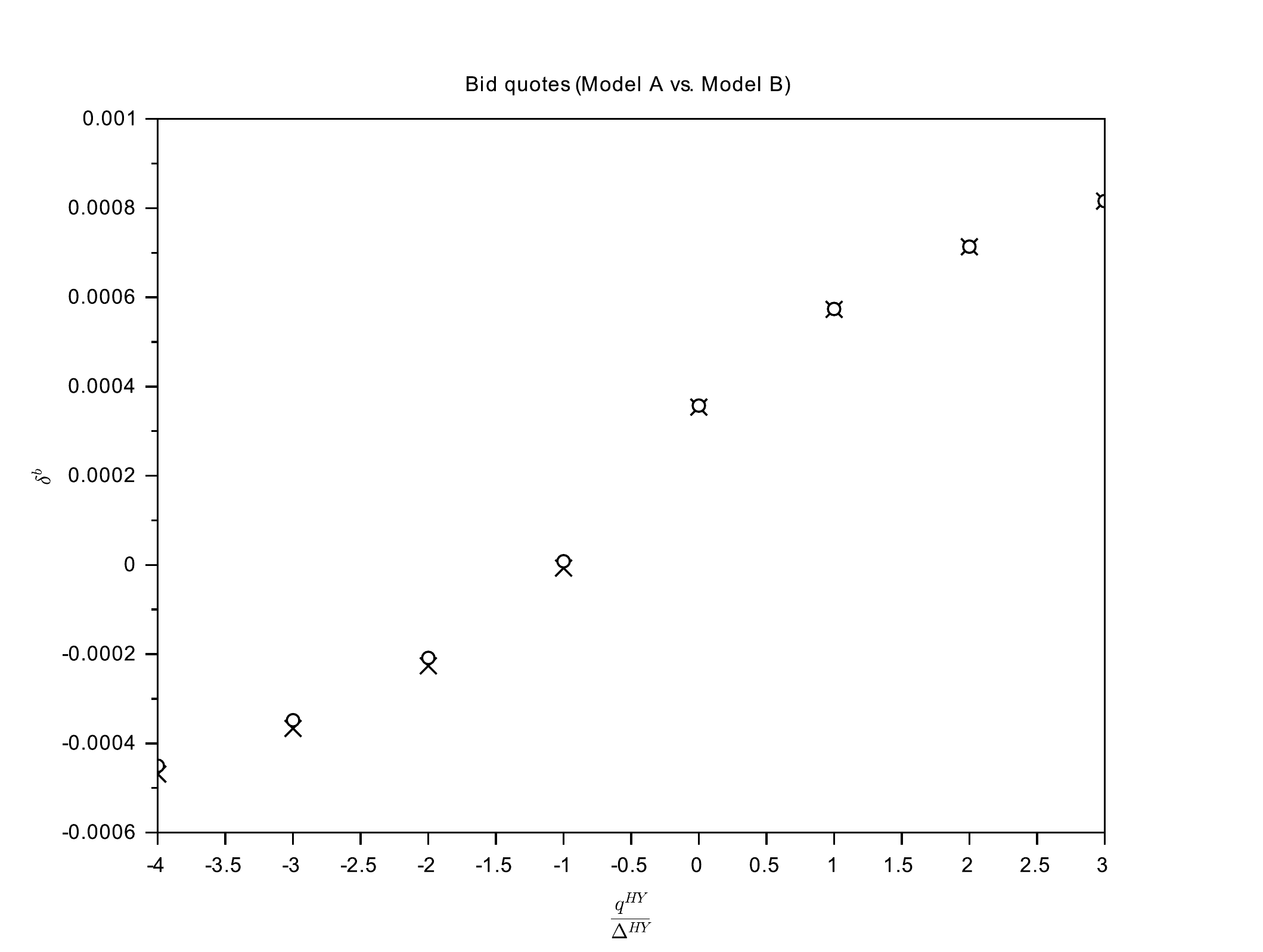}
  \caption{$q^{HY} \mapsto \delta^{HY,b}(0,q^{HY})$ in Model A (crosses) and $q^{HY} \mapsto \delta^{HY,b}(0,q^{HY})$ in Model B (circles).}\label{comp3}
\end{figure}

\begin{figure}[H]
  \centering
  \includegraphics[width=0.65\textwidth]{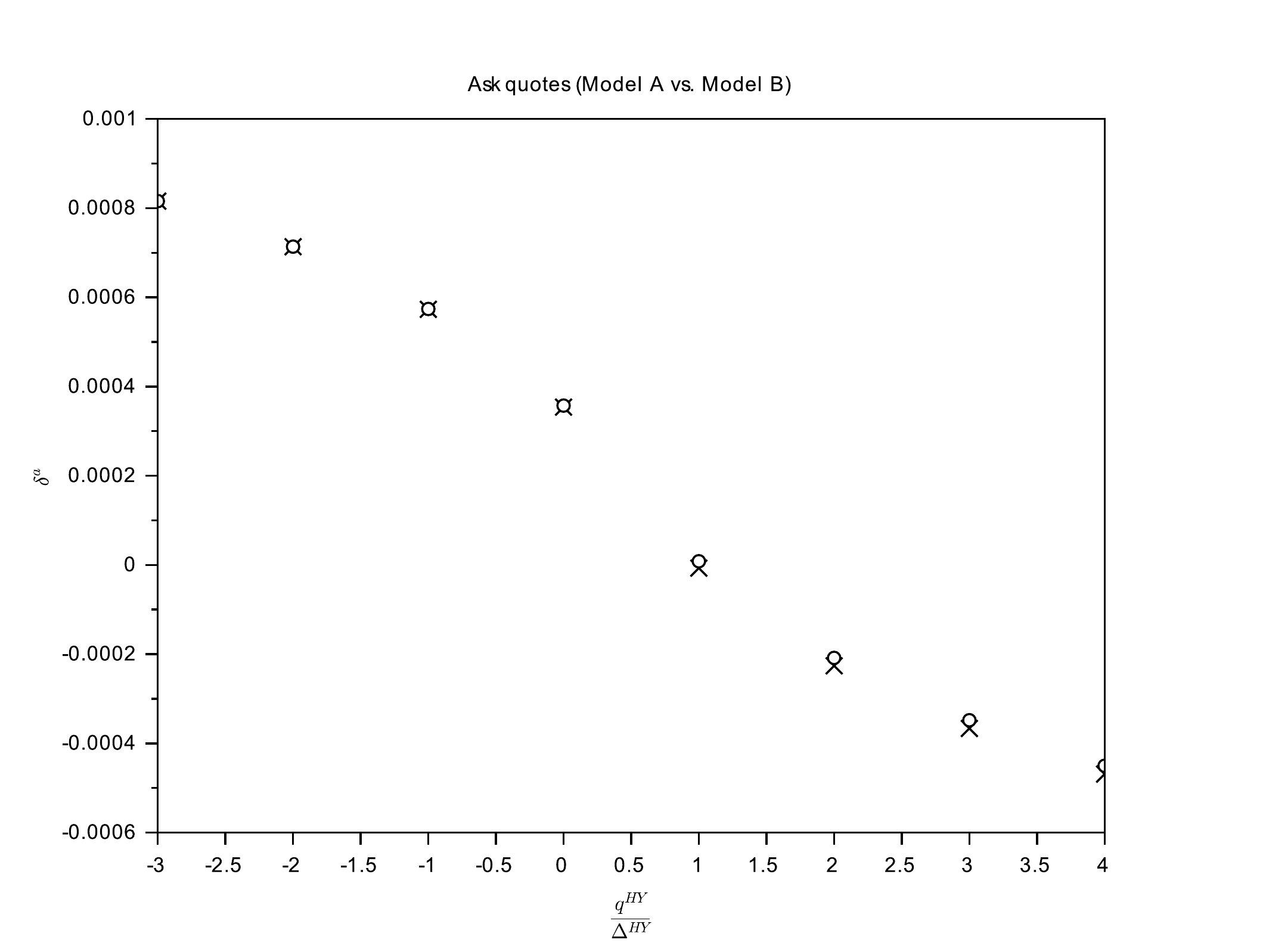}
  \caption{$q^{HY} \mapsto \delta^{HY,a}(0,q^{HY})$ in Model A (crosses) and $q^{HY} \mapsto \delta^{HY,a}(0,q^{HY})$ in Model B (circles).}\label{comp4}
\end{figure}

We can now consider the two indices together, and look at the influence of correlation for the market making of several assets at the same time. We approximated the solution $\theta$ of the systems of ODEs (\ref{sec5:thetagen}) by using an implicit scheme and a Newton's method at each time step to deal with the nonlinearity. Then we obtained the feedback control function $$(t,q^{IG},q^{HY}) \mapsto (\delta^{IG,b}(t,q^{IG},q^{HY}), \delta^{IG,a}(t,q^{IG},q^{HY}), \delta^{HY,b}(t,q^{IG},q^{HY}), \delta^{HY,a}(t,q^{IG},q^{HY}))$$ which gives the optimal bid and ask quotes at time $t$ for the two indices when $q^{IG}_{t-} = q^{IG}$ and $q^{HY}_{t-} = q^{HY}$.\\

In Figures \ref{IG_bid_2d} and \ref{HY_bid_2d}, we have plotted the optimal bid quotes for the two indices.\footnote{The results are similar, \emph{mutatis mutandis}, for the ask quotes, and are not displayed.} We see that the market maker's inventory on both indices influences his quotes. Because the correlation coefficient is positive, $(q^{IG},q^{HY}) \mapsto \delta^{IG,b}(0,q^{IG},q^{HY})$ and $(q^{IG},q^{HY}) \mapsto \delta^{HY,b}(0,q^{IG},q^{HY})$ are increasing in $q^{IG}$ and $q^{HY}$.\\

\begin{figure}[H]
  \centering
  \includegraphics[width=0.75\textwidth]{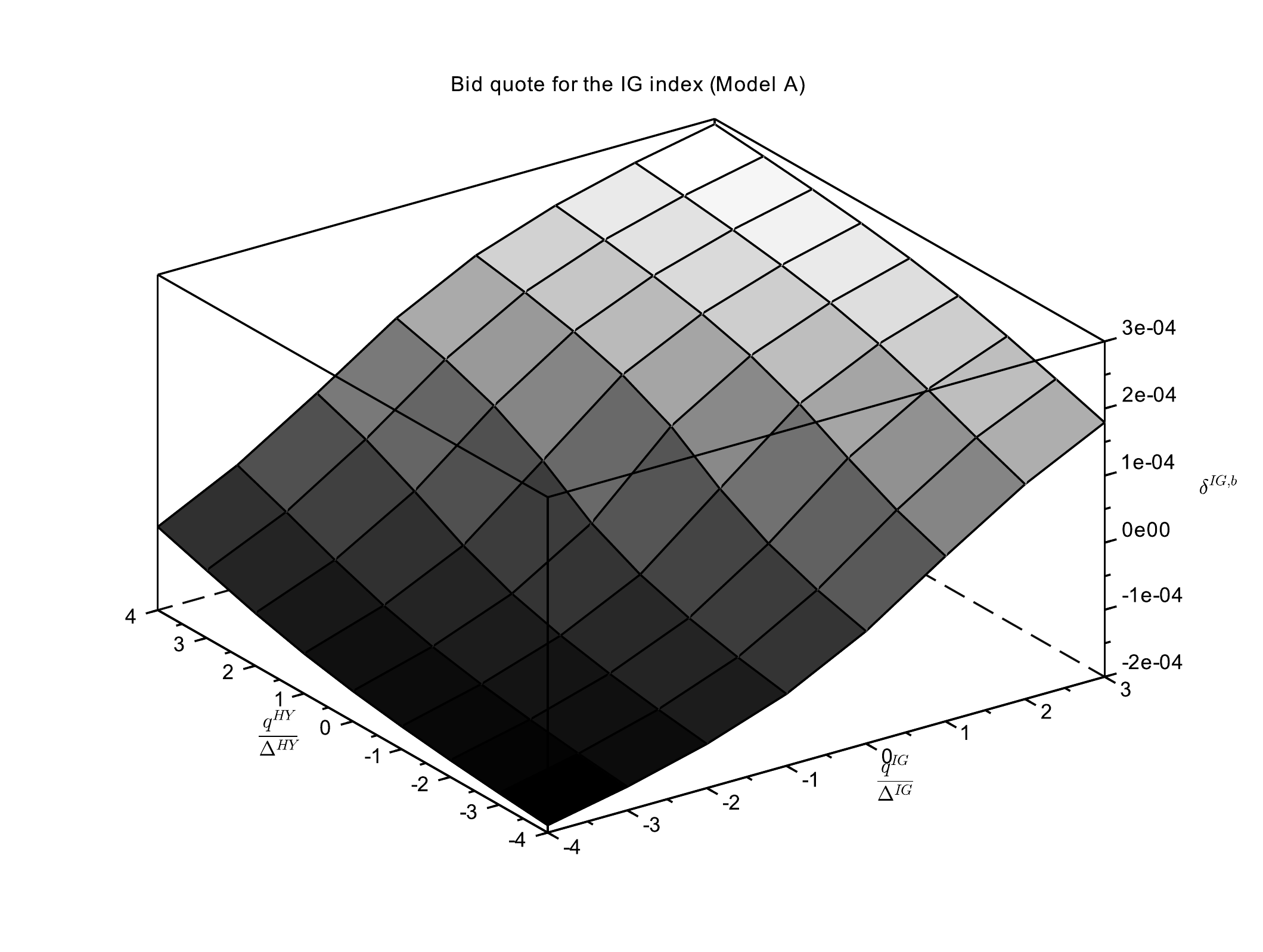}
  \caption{$(q^{IG},q^{HY}) \mapsto \delta^{IG,b}(0,q^{IG},q^{HY})$ -- in the case of Model A.}\label{IG_bid_2d}
\end{figure}

\begin{figure}[H]
  \centering
  \includegraphics[width=0.75\textwidth]{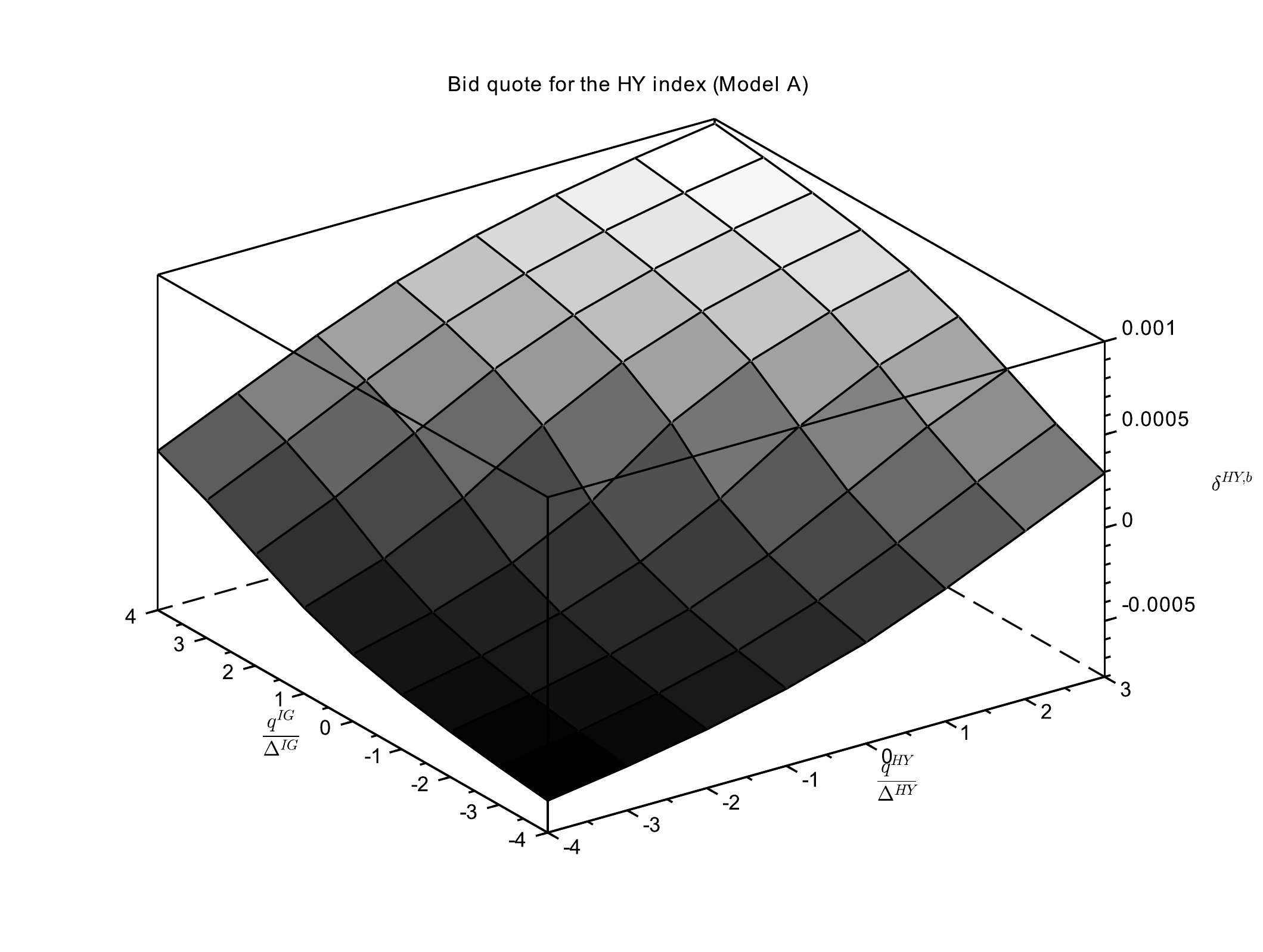}
  \caption{$(q^{IG},q^{HY}) \mapsto \delta^{HY,b}(0,q^{IG},q^{HY})$  -- in the case of Model A.}\label{HY_bid_2d}
\end{figure}

\begin{figure}[H]
  \centering
  \includegraphics[width=0.65\textwidth]{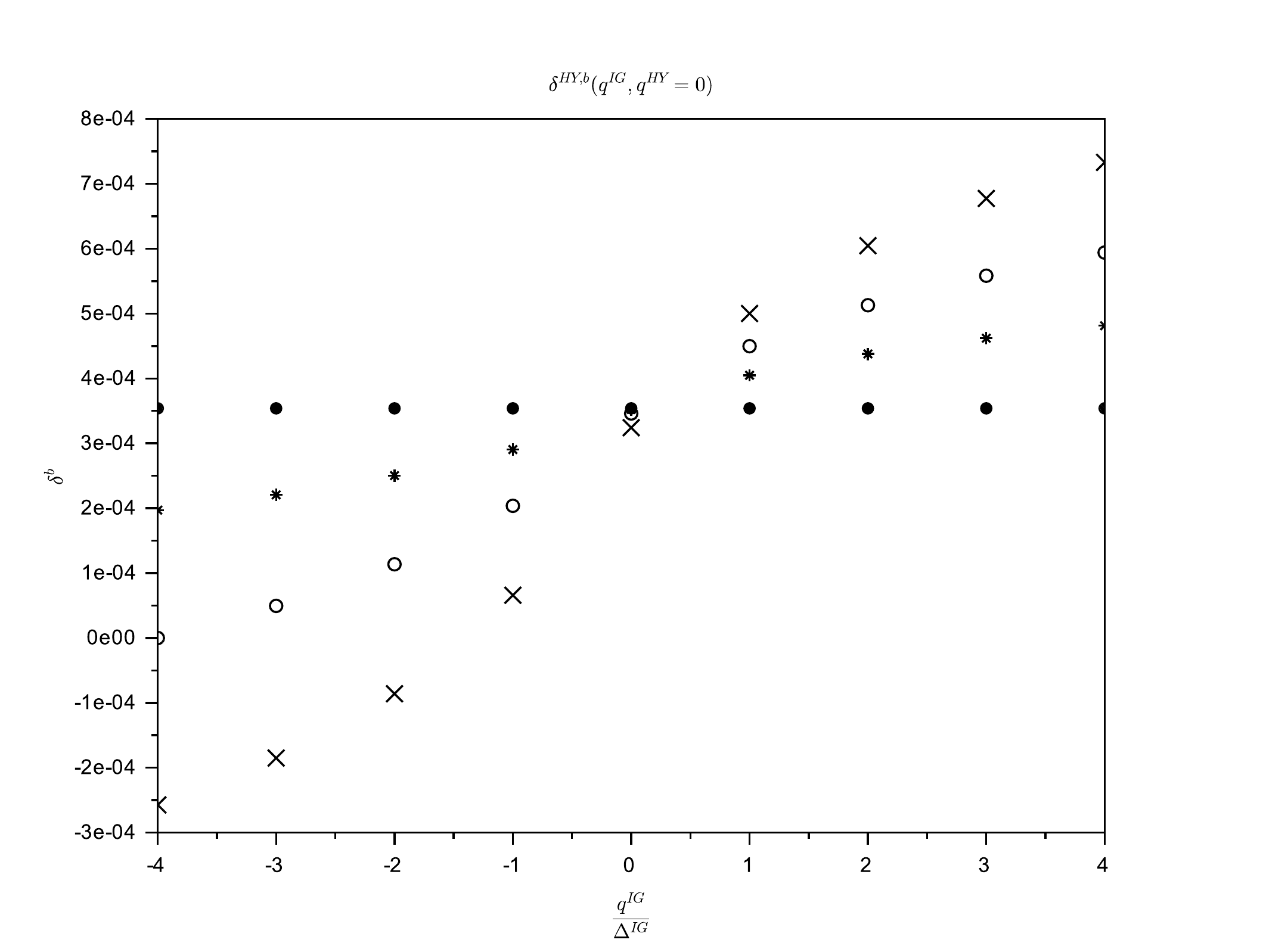}
  \caption{$q^{IG} \mapsto \delta^{HY,b}(0,q^{IG},0)$ in the case of Model A, for different values of $\rho$. $\rho=0.9$ (crosses), $\rho=0.6$ (circles), $\rho=0.3$ (stars)and $\rho=0$ (dots).}\label{correl}
\end{figure}

To see the influence of correlation, we have also computed the optimal quotes for four values of the correlation parameter: $\rho \in \lbrace 0, 0.3, 0.6, 0.9 \rbrace$. Figure \ref{correl} represents, for these different values of $\rho$, the bid quote $\delta^{HY,b}(0,q^{IG},0)$ for the HY index, when the inventory with respect to the HY index is equal to $0$, for different values of the inventory with respect to the IG index. We see that the correlation coefficient has a strong influence on the optimal quote: the more correlated the two assets, the more conservatively (respectively aggressively) the market maker should quote at the bid when he has a long (respectively short) inventory in the other asset.\\

\section*{Conclusion}

In this paper, we considered a framework \emph{\`a la} Avellaneda-Stoikov with general intensity functions, and we showed that for the different optimization criteria used in the literature, the dimensionality of the problem can be divided by 2. We also showed how to find closed-form approximations for the optimal quotes, generalizing therefore the Gu\'eant--Lehalle--Fernandez-Tapia formulas (used by many in the industry) to the two kinds of objective function used in the literature and to almost any intensity function. We also generalized our model to the multi-asset case, and showed the importance of taking account of the correlation between assets. In particular, we have derived closed-form approximations for the optimal quotes of a multi-asset market maker, an important breakthrough for practitioners who sometimes cannot solve systems of dozens or hundreds of nonlinear ODEs. The simple applications to credit indices we considered confirm the importance of the multi-asset framework.\\

\begin{center}
  \textbf{Financial support received for presenting the research results in an international conference}
\end{center}

Ce travail a \'et\'e r\'ealis\'e dans le cadre du laboratoire d'excellence ReFi port\'e par heSam Universit\'e, portant la r\'ef\'erence ANR-10-LABX-0095. Ce travail a b\'en\'efici\'e d'une aide de l'Etat g\'er\'ee par l'Agence Nationale de la Recherche au titre du projet Investissements d'Avenir Paris Nouveaux Mondes portant le r\'ef\'erence ANR-11-IDEX-0006-02.

\vspace{1cm}

\bibliographystyle{plain}

\begin{thebibliography}{10}

\bibitem{avellaneda2008high}
M.~Avellaneda and S.~Stoikov.
\newblock High-frequency trading in a limit order book.
\newblock {\em Quantitative Finance}, 8(3):217--224, 2008.

\bibitem{bayraktar2011liquidation}
E.~Bayraktar and M.~Ludkovski.
\newblock Liquidation in limit order books with controlled intensity.
\newblock \emph{Mathematical Finance}, 24(4):627--650, 2014.

\bibitem{brezis}
H. Brezis
\newblock Functional analysis, Sobolev spaces and partial differential equations.
\newblock \emph{Springer Science \& Business Media}, 2010.


\bibitem{cartea2013robust}
{\'A}.~Cartea, R.~Donnelly, and S.~Jaimungal.
\newblock Algorithmic trading with model uncertainty.
\newblock {\em Available at SSRN 2310645}, 2013.

\bibitem{cartea2013risk}
{\'A}.~Cartea and S.~Jaimungal.
\newblock Risk metrics and fine tuning of high frequency trading strategies.
\newblock {\em Mathematical Finance}, 25(3):576--611, 2013.


\bibitem{cartea2014buy}
{\'A}.~Cartea, S.~Jaimungal, and J.~Ricci.
\newblock Buy low, sell high: A high frequency trading perspective.
\newblock {\em SIAM Journal on Financial Mathematics}, 5(1):415--444, 2014.


\bibitem{cartea2015algorithmic}
{\'A}.~Cartea, S.~Jaimungal, and J.~Penalva.
\newblock {\em Algorithmic and High-Frequency Trading}.
\newblock Cambridge University Press, 2015.

\bibitem{donnelly2014ambiguity}
R.~Donnelly.
\newblock Ambiguity aversion in algorithmic and high frequency trading.
\newblock {\em PhD Thesis}, 2014.

\bibitem{egv2017}
D. Evangelista, O. Gu\'eant D. Vieira.
\newblock New closed-form approximations in multi-asset market making.
\newblock {\em Preprint}, 2017.

\bibitem{fodra2012}
P. Fodra and M. Labadie.
\newblock High-frequency market-making with inventory constraints and directional bets.
\newblock {\em arXiv preprint arXiv:1206.4810}, 2012.




\bibitem{grossman1988liquidity}
S.~Grossman and M.~Miller.
\newblock Liquidity and market structure.
\newblock {\em The Journal of Finance}, 43(3):617--633, 1988.



\bibitem{gueant2012optimal}
O.~Gu{\'e}ant, C.-A. Lehalle, and J.~Fernandez-Tapia.
\newblock Optimal portfolio liquidation with limit orders.
\newblock {\em SIAM Journal on Financial Mathematics}, 3(1):740--764, 2012.

\bibitem{gueant2013dealing}
O.~Gu{\'e}ant, C.-A. Lehalle, and J.~Fernandez-Tapia.
\newblock Dealing with the inventory risk: a solution to the market making
  problem.
\newblock {\em Mathematics and financial economics}, 7(4):477--507, 2013.

\bibitem{gueant2013general}
O.~Gu\'eant and C.-A. Lehalle.
\newblock General intensity shapes in optimal liquidation.
\newblock {\em Mathematical Finance}, 25(3):457--495, 2015.



\bibitem{gueantbook}
O.~Gu{\'e}ant.
\newblock The Financial Mathematics of Market Liquidity: from Optimal Execution to Market Making.
\newblock {\em CRC Press, Taylor and Francis}, 2016.

\bibitem{guilbaud2013optimal}
F.~Guilbaud and H.~Pham.
\newblock Optimal high-frequency trading with limit and market orders.
\newblock {\em Quantitative Finance}, 13(1):79--94, 2013.



\bibitem{ho1981optimal}
T.~Ho and H.~Stoll.
\newblock Optimal dealer pricing under transactions and return uncertainty.
\newblock {\em Journal of Financial Economics}, 9(1):47--73, 1981.

\bibitem{ho1983dynamics}
T.~Ho and H.~Stoll.
\newblock The dynamics of dealer markets under competition.
\newblock {\em The Journal of Finance}, 38(4):1053--1074, 1983.

\bibitem{huitema2012optimal}
R.~Huitema.
\newblock Optimal portfolio execution using market and limit orders.
\newblock \emph{working paper}, 2012.




\bibitem{menkveld2013high}
A.~Menkveld.
\newblock High frequency trading and the new market makers.
\newblock {\em Journal of Financial Markets}, 16(4):712--740, 2013.

\bibitem{nystrom}
K. Nystr\"om, S. M. Ould Aly, and C. Zhang.
\newblock Market making and portfolio liquidation under uncertainty.
\newblock {\em International Journal of Theoretical and Applied Finance}, 2014.




\end{thebibliography}

\end{document}